\PassOptionsToPackage{unicode}{hyperref}
\PassOptionsToPackage{hyphens}{url}
\PassOptionsToPackage{dvipsnames,svgnames,x11names}{xcolor}

\documentclass[
  12pt]{article}
\usepackage{authblk}
\usepackage{enumitem}
\usepackage{tikz-cd}
\usepackage{tikz}
\usepackage{amsthm} 
\usepackage{tabularx}
\newtheorem{coro}{Corollary}
\newtheorem{theorem}{Theorem}

\newtheorem{lemma}{Lemma}

\newtheorem{prop}{Proposition}
\newtheorem{assumption}{Assumption}

\newtheorem{remark}{Remark}
\newtheorem*{remark*}{Remark}

\newcommand{\proofpart}[1]{%
  \par\medskip
  \noindent\textbf{#1}\par
  \smallskip
}

\def\E{\mathbb E}
\def\p{\mathbb P}
\def\var{\mathrm{var}}
\def\cov{\mathrm{cov}}
\def\logit{\mathrm{logit}}
\def\tauate{\tau_{\textup{ATE}}}
\def\tauatt{\tau_{\textup{ATT}}}
\def\tauatc{\tau_{\textup{ATC}}}
\def\tauow{\tau_{\textup{ATO}}}

\def\nk{n_{[k]}}

\def\TSLS{\tau_{\mathrm{TSLS}}}

\newtheorem{example}{Example}

\usepackage{amsmath,amssymb}
\usepackage{iftex}
\ifPDFTeX
  \usepackage[T1]{fontenc}
  \usepackage[utf8]{inputenc}
  \usepackage{textcomp} % provide euro and other symbols
\else % if luatex or xetex
  \usepackage{unicode-math}
  \defaultfontfeatures{Scale=MatchLowercase}
  \defaultfontfeatures[\rmfamily]{Ligatures=TeX,Scale=1}
\fi
\usepackage{lmodern}
\ifPDFTeX\else  
\fi
\IfFileExists{upquote.sty}{\usepackage{upquote}}{}
\IfFileExists{microtype.sty}{% use microtype if available
  \usepackage[]{microtype}
  \UseMicrotypeSet[protrusion]{basicmath} % disable protrusion for tt fonts
}{}
\makeatletter
\@ifundefined{KOMAClassName}{% if non-KOMA class
  \IfFileExists{parskip.sty}{%
    \usepackage{parskip}
  }{% else
    \setlength{\parindent}{0pt}
    \setlength{\parskip}{6pt plus 2pt minus 1pt}}
}{% if KOMA class
  \KOMAoptions{parskip=half}}
\makeatother
\usepackage{xcolor}
\makeatletter
\ifx\paragraph\undefined\else
  \let\oldparagraph\paragraph
  \renewcommand{\paragraph}{
    \@ifstar
      \xxxParagraphStar
      \xxxParagraphNoStar
  }
  \newcommand{\xxxParagraphStar}[1]{\oldparagraph*{#1}\mbox{}}
  \newcommand{\xxxParagraphNoStar}[1]{\oldparagraph{#1}\mbox{}}
\fi
\ifx\subparagraph\undefined\else
  \let\oldsubparagraph\subparagraph
  \renewcommand{\subparagraph}{
    \@ifstar
      \xxxSubParagraphStar
      \xxxSubParagraphNoStar
  }
  \newcommand{\xxxSubParagraphStar}[1]{\oldsubparagraph*{#1}\mbox{}}
  \newcommand{\xxxSubParagraphNoStar}[1]{\oldsubparagraph{#1}\mbox{}}
\fi
\makeatother

\usepackage{longtable,booktabs,array}
\usepackage{calc} % for calculating minipage widths
\usepackage{etoolbox}
\makeatletter
\patchcmd\longtable{\par}{\if@noskipsec\mbox{}\fi\par}{}{}
\makeatother
\IfFileExists{footnotehyper.sty}{\usepackage{footnotehyper}}{\usepackage{footnote}}
\makesavenoteenv{longtable}
\usepackage{graphicx}
\makeatletter
\def\maxwidth{\ifdim\Gin@nat@width>\linewidth\linewidth\else\Gin@nat@width\fi}
\def\maxheight{\ifdim\Gin@nat@height>\textheight\textheight\else\Gin@nat@height\fi}
\makeatother
\setkeys{Gin}{width=\maxwidth,height=\maxheight,keepaspectratio}
\makeatletter
\def\fps@figure{htbp}
\makeatother

\makeatletter
\@ifpackageloaded{caption}{}{\usepackage{caption}}
\AtBeginDocument{%
\ifdefined\contentsname
  \renewcommand*\contentsname{Table of contents}
\else
  \newcommand\contentsname{Table of contents}
\fi
\ifdefined\listfigurename
  \renewcommand*\listfigurename{List of Figures}
\else
  \newcommand\listfigurename{List of Figures}
\fi
\ifdefined\listtablename
  \renewcommand*\listtablename{List of Tables}
\else
  \newcommand\listtablename{List of Tables}
\fi
\ifdefined\figurename
  \renewcommand*\figurename{Figure}
\else
  \newcommand\figurename{Figure}
\fi
\ifdefined\tablename
  \renewcommand*\tablename{Table}
\else
  \newcommand\tablename{Table}
\fi
}
\@ifpackageloaded{float}{}{\usepackage{float}}
\floatstyle{ruled}
\@ifundefined{c@chapter}{\newfloat{codelisting}{h}{lop}}{\newfloat{codelisting}{h}{lop}[chapter]}
\floatname{codelisting}{Listing}

\makeatother
\makeatletter
\@ifpackageloaded{caption}{}{\usepackage{caption}}
\@ifpackageloaded{subcaption}{}{\usepackage{subcaption}}
\makeatother

\makeatletter
\let\equality@label\label

\newcounter{equality}
\renewcommand{\theequality}{\roman{equality}}

\newcommand{\eqlabel}[1]{%
  \refstepcounter{equality}%
  \equality@label{#1}%
  \mathrel{\overset{\text{\normalfont(\theequality)}}{=}}%
}
\makeatother

\ifLuaTeX
  \usepackage{selnolig}  % disable illegal ligatures
\fi
\usepackage[]{natbib}
\usepackage{bookmark}

\IfFileExists{xurl.sty}{\usepackage{xurl}}{} % add URL line breaks if available
\hypersetup{
  pdftitle={Introducing the CP-plot Based on Covariance Representations of Weighted Average Treatment Effects},
  pdfauthor={Pengfei Tian; Fan Yang; Peng Ding},
  pdfkeywords={beta weight, instrumental variable, observational study, overlap weight, potential outcome, propensity score},
  colorlinks=true,
  linkcolor={blue},
  filecolor={Maroon},
  citecolor={Blue},
  urlcolor={Blue},
  pdfcreator={LaTeX via pandoc}}

\newcommand{\blind}{1}

\begin{document}

\def\spacingset#1{\renewcommand{\baselinestretch}%
{#1}\small\normalsize} \spacingset{1}

\if1\blind
{
  \title{\bf
  Introducing the CP-plot for Causal Inference with Observational Studies}

  \author[1]{Pengfei Tian}
  \author[2]{Fan Yang\textsuperscript{*}}
  \author[3]{Peng Ding\textsuperscript{*}}

  \affil[1]{\small Qiuzhen College, Tsinghua University, \texttt{\href{mailto:tpf24@mails.tsinghua.edu.cn}{tpf24@mails.tsinghua.edu.cn}}}
  \affil[2]{\small Yau Mathematical Sciences Center, Tsinghua University, \texttt{\href{mailto:yangfan1987@tsinghua.edu.cn}{yangfan1987@tsinghua.edu.cn}}}
  \affil[3]{\small Department of Statistics, University of California, Berkeley, \texttt{\href{mailto:pengdingpku@berkeley.edu}{pengdingpku@berkeley.edu}}}

  \begingroup
  \renewcommand{\thefootnote}{*}
  \footnotetext[1]{Correspondence should be addressed to Fan Yang or Peng Ding.}
  \endgroup
  \date{}
  \maketitle
}
\fi

\if0\blind
{
  \title{\bf
  Introducing the CP-plot for Causal Inference with Observational Studies}
  \date{}
  \maketitle
}
\fi

\bigskip
\begin{abstract}
Under the canonical setting of observational studies for causal inference, we derive a set of exact representations for pairwise differences among weighted average treatment effects as covariances between the conditional average treatment effect and the propensity score, up to positive scaling factors. These covariance representations bridge the two core concepts in causal inference with observational studies. They immediately imply that (i) the average treatment effect
is bracketed by the average treatment effects on the treated and on the
controls, with the direction determined by the sign of the covariance between
the conditional average treatment effect and the propensity score, and (ii) the
average treatment effect under the overlap weight, the weight that is
proportional to the conditional variance of the treatment given the covariates,
is bracketed by the average treatment effects on the treated and controls when
the corresponding covariances have a common sign within both the treated and
control groups.
 We further extend these results to weighted local average treatment effects in the instrumental variable framework. Building on this theory, we recommend the ``CP-plot'' of the estimated conditional average treatment effect against the estimated propensity score, and implement it in the R package \texttt{CPplot}.
\end{abstract}

\noindent%
{\it Keywords:} Beta weight; Instrumental variable; Observational study; Overlap weight; Propensity score.
\vfill

\newpage
\spacingset{1.7} % DON'T change the spacing!

\section{Introduction}

Causal inference with observational studies is critical in many disciplines. Under the potential outcomes framework, the average treatment effect (ATE) can be identified under two assumptions. The first is unconfoundedness, which requires the treatment to be independent of the potential outcomes given the observed covariates. The second is overlap, which requires the propensity score, the conditional probability of treatment given the observed covariates, to be bounded away from zero and one \citep{rosenbaum1983central}. Although most of the literature focuses on ATE, other weighted average treatment effects are also of scientific and policy relevance. For instance, the average treatment effect on the treated (ATT) measures the effect on the treated population, the average treatment effect on the controls (ATC) measures the effect on the control population, and the average treatment effect under the overlap weight (ATO) measures the effect on the population weighted by the conditional variance of the treatment given the observed covariates \citep{li2018balancing}. ATO has attracted growing interest for its robustness to the violation of the overlap assumption \citep{li2019addressing} and its connection to various regression-based estimators \citep{angrist1998estimating,wallace2015doubly,ding2021frisch,lee2018simple}. 
These estimands differ only in how they weight the conditional average treatment effect (CATE), the average treatment effect among units sharing the same covariate values, across the propensity-score distribution. This observation raises a natural question: how can the differences between weighted average treatment effects be represented and interpreted? We answer this question by deriving exact difference representations that express the gaps among weighted estimands through covariances between the CATE and the propensity score under suitable target distributions. These representations yield immediate ordering results and motivate a graphical diagnostic based on the relationship between the CATE and the propensity score.

First, we begin with a known covariance representation for the differences between ATT, ATE, and ATC in \citet[Chapter 13.5, Problem 13.1]{ding2024first}. These identities refine the convex-combination fact that ATE lies between ATT and ATC by showing that the directions of the gaps are governed by the covariance between the CATE and the propensity score. This covariance links two central objects in causal inference: the CATE characterizes treatment-effect heterogeneity on the outcome side, whereas the propensity score characterizes the treatment-assignment mechanism. Under ignorability, it also admits an equivalent interpretation in terms of the association between individual treatment gains and treatment receipt, or selection on treatment gains. Extending this benchmark, we derive new analogous covariance representations for the gaps involving ATO: the differences between ATT and ATO and between ATO and ATC can be expressed through covariances evaluated under the treated and control target distributions, respectively. These new representations give necessary and sufficient covariance conditions for the bracketing relation between ATO, ATT, and ATC. We then introduce the CP function, defined as the conditional mean of the CATE given the propensity score, and show that its monotonicity provides a transparent sufficient condition for signing these covariances. In the finite-strata fixed-effects setting studied by \citet{humphreys2025bounds}, the fixed-effects estimand coincides with the overlap-weighted estimand, so our ATO representation recovers the corresponding bracketing result as a special case. The main focus of this paper is therefore to extend the covariance-representation logic to the overlap-weighted estimand and to use it as the basis for the CATE--propensity-score plot (CP-plot) diagnostic.

Second, we develop analogous difference representations in the instrumental variable (IV) setting with a binary IV and a binary treatment \citep{angrist1996identification,imbens1994identification}, known as the local average treatment effect (LATE) framework. Most of this literature focuses on the LATE, defined as the ATE among compliers whose treatment status is changed by the IV. We consider weighted local estimands, including local analogs of ATT, ATC, and ATO. 
Recently, \citet{choi2024instrumental} introduced a general class of weighted LATEs; in particular, the overlap-weighted LATE has appeared frequently in the recent IV literature \citep[Section~4]{bugni2023inference}; see also \citet{borusyak2024negative}, \citet{zhao2026two} and \citet{jiang2026improving}. The same covariance logic applies after replacing the population distribution by the complier distribution and the propensity score by the IV propensity score. Our representations then imply that the overlap-weighted LATE lies between the local analogs of ATT and ATC when the relevant complier covariances share a common sign, with monotonicity of the local CP function serving as a simple sufficient condition.

Third, we extend the results to the beta-weight family  \citep{matsouaka2024causal}, a two-parameter family that includes ATE, ATT, ATC, and ATO as special cases.

Fourth, building on the theory, we recommend a practical diagnostic, the ``CP-plot'', which plots the estimated CATE against the estimated propensity score. We implement it in the R package \texttt{CPplot} and use it to revisit several observational studies and an IV study.

The rest of the paper proceeds as follows. 
Section~\ref{sec:obs} studies propensity-score weighted causal estimands, presents exact difference representations among ATE, ATT, ATO, and ATC, and establishes the bracketing relation among ATO, ATT, and ATC.
 It then discusses identification under ignorability and connects the overlap-weighted estimand to regression-based estimands. 
 Section~\ref{sec:iv} extends the difference representations and bracketing results to propensity-score weighted local causal effects in a binary IV setting and connects the overlap-weighted local estimand to the population two-stage least-squares (TSLS) coefficient. 
 Section~\ref{sec:cpplot} introduces the CP-plot and the local CP-plot.
 Section~\ref{sec:app} provides empirical illustrations based on several observational studies and an IV study.
 Section~\ref{sec:ext} develops extensions to beta weights. Section~\ref{sec:disc} concludes with a summary and practical recommendations.
 The online Supplementary Material contains additional results and proofs.

\section{Weighted average treatment effects and their relationships}\label{sec:obs}
We work in the standard potential-outcomes framework for causal inference with observational studies. For each unit, \(Z\in\{0,1\}\) indicates whether the unit receives treatment (\(Z=1\)) or control (\(Z=0\)). Let \(Y(1)\) and \(Y(0)\) denote the potential outcomes under treatment and control, respectively. The vector \(X\) collects the observed pre-treatment covariates. The conditional average treatment effect (CATE) is
\[
\tau(X)=\E\{Y(1)-Y(0)\mid X\},
\]
the average treatment effect among units sharing the same covariate value \(X\). Population-level causal estimands are obtained by averaging \(\tau(X)\) under various weighting schemes, each weight corresponding to a different target population.

Let \(e(X)=\p(Z=1\mid X)\) denote the propensity score
\citep{rosenbaum1983central}, and let
\(p=\p(Z=1)=\E\{e(X)\}\).
We impose the standard overlap and ignorability conditions.
\begin{assumption}[overlap and ignorability]\label{assume:ignor}
    \(e(X)\in(0,1)\) and $
(Y(1),Y(0))\perp\!\!\!\perp Z\mid X
$.
\end{assumption}
Under Assumption~\ref{assume:ignor}, ATE, ATT, and ATC admit the following weighted-CATE representations:
\[
\begin{aligned}
\tauate
&= \E\{Y(1)-Y(0)\}
= \E\{\tau(X)\},\\
\tauatt
&= \E\{Y(1)-Y(0)\mid Z=1\}\eqlabel{eql:att-ignor} \E\{ \tau(X)\mid Z=1 \}
= {\E\{e(X)\tau(X)\}}/p,\\
\tauatc &= \E\{Y(1)-Y(0)\mid Z=0\} \eqlabel{eql:atc-ignor} \E\{ \tau(X)\mid Z=0 \}
= {\E[\{1-e(X)\}\tau(X)]}/(1-p).
\end{aligned}
\]
Thus \(\tauatt\) and \(\tauatc\) are averages of the CATE over the
covariate distributions of the treated and control groups, respectively. 
In particular, ATE is always bracketed by ATT and ATC because it is a convex combination of the two:
\begin{align}\label{eq:ate_combine}
\tauate
=
\p(Z=1)\cdot \tauatt+\p(Z=0)\cdot\tauatc
=
p\cdot\tauatt+(1-p)\cdot\tauatc.
\end{align}

\citet{li2018balancing} proposed the overlap weight, which is proportional to the conditional variance of the treatment given the covariates, $\var(Z \mid X) = e(X)\{1-e(X)\}$. This yields ATO,
\[
\tauow = \frac{\E[e(X)\{1-e(X)\}\tau(X)]}{\E[e(X)\{1-e(X)\}]}.
\]
Unlike $e(X)$ or $1-e(X)$, the overlap weight $e(X)\{1-e(X)\}$ vanishes at both endpoints of $[0,1]$. By definition, $\tauow$ emphasizes covariate regions where both treatment arms are well represented and downweights regions with extreme propensity scores. This makes ATO a stable target estimand, especially when the treated and control groups overlap poorly. The following theorem gives exact covariance representations for four key pairwise differences involving ATE and ATO.

\begin{theorem}
    \label{thm:diff-rep}
Assume Assumption \ref{assume:ignor}.
 \begin{enumerate}[label=(\roman*),leftmargin=*, align=left]
\item
We have
\begin{align}\label{eq:ATE_diff_re}
    \tauatt-\tauate
=
\frac{\cov\{\tau(X),e(X)\}}{p},
\qquad
\tauate-\tauatc
=
\frac{\cov\{\tau(X),e(X)\}}{1-p}.
\end{align}
\item We have
\begin{align}\label{eq:ATO_diff_re}
\tauatt-\tauow
=
\frac{\cov\{\tau(X),e(X)\mid Z=1\}}
{\E\{1-e(X)\mid Z=1\}},
\qquad
\tauow-\tauatc
=
\frac{\cov\{\tau(X),e(X)\mid Z=0\}}
{\E\{e(X)\mid Z=0\}}.
\end{align}
\end{enumerate}
\end{theorem}Equation~\eqref{eq:ATE_diff_re} refines the convex-combination identity in
\eqref{eq:ate_combine}: it shows that the directions of the two ATE gaps are
governed by the marginal covariance between the CATE and the propensity score.
This representation is already known \citep[Problem 13.1]{ding2024first}.
Equation~\eqref{eq:ATO_diff_re} gives the corresponding representations for
ATO under the treated and control distributions. Because all denominators in
Theorem~\ref{thm:diff-rep} are positive, each pairwise difference has the sign
of its corresponding covariance.

To further interpret these covariances, for $z=0,1$, define the
potential-outcome selection bias
\(
B(z)
=
\E\{Y(z)\mid Z=1\}-\E\{Y(z)\mid Z=0\}
\)
\citep{heckman1997making}. The following proposition provides two useful
representations.

\begin{prop}\label{prop:covariance-interpretation}
Assume Assumption~\ref{assume:ignor}.
\begin{enumerate}[label=(\roman*),leftmargin=*,align=left]

\item
The marginal CATE--propensity covariance admits the equivalent
representations
\begin{align}\label{eq:selection-on-gains}
\cov\{\tau(X),e(X)\}
=
\cov\{Y(1)-Y(0),Z\}
=
\var(Z)\cdot \{B(1)-B(0)\}.
\end{align}

\item
The marginal covariance and the corresponding covariances within the
treated and control groups satisfy
\begin{align}\label{eq:marginal-conditional-covariance}
\cov\{\tau(X),e(X)\}
={}&
\frac{\var(Z)}{\E\{\var(Z\mid X)\}}
\Big[
p\cdot \cov\{\tau(X),e(X)\mid Z=1\}\notag\\
&\qquad\qquad+
(1-p)\cdot \cov\{\tau(X),e(X)\mid Z=0\}
\Big].
\end{align}

\end{enumerate}
\end{prop}
Proposition \ref{prop:covariance-interpretation}(i) gives three equivalent interpretations of the same covariance.
The CATE $\tau(X)$ characterizes treatment-effect heterogeneity on the outcome
side, whereas the propensity score $e(X)$ characterizes treatment assignment.
Thus, $\cov\{\tau(X),e(X)\}$ measures the alignment between treatment gains and treatment propensity: whether units more likely to receive treatment
also tend to benefit more from it. Under ignorability, it is equivalently the
covariance between the individual treatment gain $Y(1)-Y(0)$ and treatment
receipt $Z$. After normalization by $\var(Z)$, it further equals
$B(1)-B(0)$, corresponding to selection on treatment gains
\citep{heckman1997making}. Under a complete randomized experiment, this
covariance is zero; in observational studies, a nonzero value reflects
systematic selection with respect to treatment gains.

Proposition \ref{prop:covariance-interpretation}(ii) connects this marginal association with its treated-specific counterparts. The two conditional covariances characterize
the association between treatment gains and treatment propensity separately
within the treated and control populations.
Equation~\eqref{eq:marginal-conditional-covariance} shows that their weighted
combination determines the marginal covariance. Hence, common conditional
signs imply the same marginal sign, whereas the converse need not hold.

Theorem~\ref{thm:diff-rep} yields the following exact bracketing
characterization.
\begin{coro}\label{coro:exact-bracketing}
Assume Assumption \ref{assume:ignor}.
\begin{enumerate}[label=(\roman*),leftmargin=*, align=left]
\item We have
$
\tauate\in[\tauatc,\tauatt]
$
if and only if
$
\cov\{\tau(X),e(X)\}\ge 0
$.
Similarly,
$
\tauate\in[\tauatt,\tauatc]
$
if and only if
$
\cov\{\tau(X),e(X)\}\le 0$. 
\item We have
$
\tauow\in[\tauatc,\tauatt]
$
if and only if
$
\cov\{\tau(X),e(X)\mid Z=z\}\ge 0
$ for $z\in \{0,1\}$.
Similarly,
$
\tauow\in[\tauatt,\tauatc]
$
if and only if
$
\cov\{\tau(X),e(X)\mid Z=z\}\le 0$ for $z\in \{0,1\}$. 
\end{enumerate}
\end{coro}

Corollary~\ref{coro:exact-bracketing} gives necessary and sufficient conditions for these bracketing relations, expressed entirely through covariances between the CATE and the propensity score. 

We next give a simple sufficient
condition for signing these covariances.
Define the CATE--propensity-score (CP) function
\begin{align}\label{eq:def_ge}
g(e)=\E\{\tau(X)\mid e(X)=e\},
\end{align}
the conditional mean of the CATE among units with propensity score \(e\).
Because \(\tauate\), \(\tauatt\), \(\tauatc\), and \(\tauow\) are all
averages of \(\tau(X)\) weighted by functions of \(e(X)\), the CP function \(g\)
determines how these estimands relate. Moreover, for $z\in \{0,1\}$, 
\begin{align}\label{eq:balance_first}
\cov\{\tau(X),e(X)\mid Z=z\}
&=
\cov\!\left[
\E\{\tau(X)\mid e(X),Z=z\},
e(X)
\mid Z=z
\right] \\
&=
\cov\{g(e(X)),e(X)\mid Z=z\},\notag
\end{align}
where the second equality holds without ignorability because $\E\{\tau(X) \mid e(X), Z=z \} = \E\{\tau(X) \mid e(X) \}  = g(e(X))$, by the balancing property of the propensity score $
Z\perp\!\!\!\perp X\mid e(X)$ \citep{rosenbaum1983central} and the definition of $g(\cdot)$ in \eqref{eq:def_ge}.

Thus, the monotonicity of the CP function provides a transparent sufficient
condition for the signs of the tilted covariances in
Corollary~\ref{coro:exact-bracketing}. We impose the following monotonicity
condition on \(g\).

\begin{assumption}\label{assump:mono}
Suppose that
\begin{enumerate}[label=(\roman*),leftmargin=*, align=left]
    \item \(g(e)\) is a non-decreasing function, or
    \item \(g(e)\) is a non-increasing function.
\end{enumerate}
\end{assumption}

Assumption~\ref{assump:mono} holds whenever units with a higher (or lower)
propensity to receive treatment also tend to benefit more from it. For example,
doctors may be more likely to assign surgery to patients in worse condition,
who may also be expected to benefit more from treatment. The assumption is
weaker than a pointwise monotone relationship between \(\tau(X)\) and
\(e(X)\), since it constrains only the conditional mean
\(\E\{\tau(X)\mid e(X)\}\). The following example verifies Assumption~\ref{assump:mono} in a standard parametric model; we do not rely on the parametric assumptions in practice, but use the model only to illustrate the condition.

\begin{example}
Consider a parametric model in which \(\tau(X)=X^\top\beta\) and
\(e(X)=\logit^{-1}(X^\top\gamma)\), with
\(X\sim \mathcal N(\mu,\Sigma)\). Since
\((X^\top \beta,X^\top \gamma)^\top\) is bivariate normal, the
conditional expectation formula gives
\[
g(e)=\E[X^\top \beta\mid X^\top \gamma=\logit(e)]
=\beta^\top \mu
+
(\beta^\top \Sigma \gamma)(\gamma^\top \Sigma \gamma)^{-1}
\{\logit(e)-\gamma^\top \mu\}.
\]
Clearly, \(g(e)\) is monotonic in \(e\).
\end{example}

Monotonicity conditions of this type have been used elsewhere to compare causal quantities, including in bounds on fixed-effects estimands
\citep{humphreys2025bounds}, analyses of bias direction
\citep{vanderweele2008sign,chiba2009sign}, effect ordering in signed causal graphs
\citep{vanderweele2010signed}, and bias amplification by instrumental variables
\citep{ding2017instrumental}.

Under monotonicity, we have the following result, which extends \citet{humphreys2025bounds}.

\begin{coro}\label{coro:ow in atc att}
Assume Assumption \ref{assume:ignor}.
 Under Assumption~\ref{assump:mono}(i), we have
$
\tauow\in[\tauatc,\tauatt].
$
Under Assumption~\ref{assump:mono}(ii), we have
$
\tauow\in[\tauatt,\tauatc].
$
\end{coro}
We end this section with the following remarks.
\begin{remark}[Connection with \citet{humphreys2025bounds}]
In a saturated finite-strata setting, the coefficient on treatment in the
regression with stratum fixed effects equals \(\tauow\). Hence,
Corollary~\ref{coro:ow in atc att} yields a bracketing result for the
fixed-effects estimand. Our condition requires monotonicity only of the
coarsened CP function and therefore permits strata with the same propensity
score to have different treatment effects. Section~\ref{subsec:compareHumphrey}
of the Supplementary Material gives the detailed comparison.
\end{remark}

\begin{remark}[Connection with OLS]
ATO is closely related to the coefficient on \(Z\)
in the population OLS projection of \(Y\) on \((1,Z,X)\). Let \(\beta_1\) denote this coefficient. When \(e(X)\) lies in the linear span of \((1,X)\), \(\beta_1=\tauow\); without this linearity, the OLS coefficient
generally contains an additional term induced by the discrepancy between
\(e(X)\) and its linear projection on \((1,X)\). This connection complements the implied-weight perspective of
\citet{chattopadhyay2023implied}, who characterize linear regression
estimators through the weights they implicitly assign and the corresponding
target populations. Section~\ref{subsec:regression_obs}
of the Supplementary Material gives the full decomposition.
\end{remark}

\begin{remark}
Under Assumption \ref{assume:ignor}, \(\tau(X)\) admits the identification formula
\[
\tau(X)
=
\E(Y\mid Z=1,X)-\E(Y\mid Z=0,X),
\]
so \(\tauow\), \(\tauatt\), and \(\tauatc\) are functionals of the
observable data distribution, and the same bracketing relations hold for
the corresponding statistical estimands.

For the covariance algebra, ignorability is used only in
equalities~\eqref{eql:att-ignor} and~\eqref{eql:atc-ignor}. All subsequent
weighting identities follow from the propensity-score identity
\begin{align}\label{eq:exax}
\E\{a(X)\mid Z=1\}
=
\frac{\E\{e(X)a(X)\}}{p},
\qquad
\E\{a(X)\mid Z=0\}
=
\frac{\E[\{1-e(X)\}a(X)]}{1-p},
\end{align}
which holds for any integrable \(a(X)\) without ignorability.
Consequently, the covariance representations remain valid without
ignorability when ATT and ATC are interpreted as the conditional-CATE
averages appearing on the right-hand sides of
equalities~\eqref{eql:att-ignor} and~\eqref{eql:atc-ignor}.
\end{remark}

\section{Extension to local average treatment effects}\label{sec:iv}
We next develop the IV counterpart of the difference representations in
Section~\ref{sec:obs}.
 IV allows us to study causal effects when treatment receipt is not ignorable, by exploiting an external instrument that induces exogenous variation in treatment. Under the standard IV assumptions, this variation identifies
local causal effects for compliers, whose treatment status is affected by the
IV.
Consider a binary
IV \(Z\in\{0,1\}\), a binary treatment \(D\in\{0,1\}\), an outcome \(Y\),
and pre-IV covariates \(X\). For \(z\in\{0,1\}\), let \(D(z)\) denote the
potential treatment under IV value \(z\), and for \(z,d\in\{0,1\}\), let
\(Y(z,d)\) denote the potential outcome under IV value \(z\) and treatment
value \(d\). The observed treatment and outcome satisfy
\[
D=ZD(1)+(1-Z)D(0),
\qquad
Y=ZY(1,D(1))+(1-Z)Y(0,D(0)).
\]

Let $e(X)=\p(Z=1\mid X)$ denote the IV propensity score. We adopt the standard IV framework of \citet{imbens1994identification} and
\citet{angrist1996identification}, formalized in Assumption \ref{assump:iv}. 
\begin{assumption}\label{assump:iv}
Suppose that:
\begin{enumerate}[label=(\roman*),leftmargin=*, align=left]
    \item \textbf{Conditional IV exogeneity:}
    \(
    Z \perp\!\!\!\perp \{D(1),D(0),Y(1,1),Y(1,0),Y(0,1),Y(0,0)\}\mid X.
\)
    \item \textbf{Exclusion restriction:}
    \(
    Y(1,d)=Y(0,d)\) for \(d=0,1.\)
    \item \textbf{Relevance:}
    \(\E\{D(1)-D(0)\mid X\}>0\) almost surely.
    \item \textbf{Monotonicity:}
    \(D(1)\ge D(0).\)
    \item \textbf{Overlap:}
\(e(X)\in(0,1).\)
\end{enumerate}
\end{assumption}
Under the exclusion restriction in Assumption \ref{assump:iv}(ii), let \(Y(d)\) denote the common value of \(Y(1,d)\)
and \(Y(0,d)\). The IV framework induces a natural target subpopulation, namely
the compliers. To define
this population formally, let
\[
U=\bigl(D(1),D(0)\bigr)\in\{(1,1),(1,0),(0,0),(0,1)\},
\]
corresponding respectively to always-takers \((\textup{a})\), compliers
\((\textup{c})\), never-takers \((\textup{n})\), and defiers \((\textup{d})\).
Under Assumption~\ref{assump:iv}(iv), the defier stratum is absent. Define the CATE among compliers as
\[
\tau^{\textup{c}}(X)=\E\{Y(1)-Y(0)\mid U=\textup{c},X\}.
\]
Let
$$\pi^{\textup{c}}(X)=\p(U=\textup{c}\mid X)$$
denote the principal score for the complier stratum \citep{ding2017principal}. Under Assumption \ref{assump:iv}(iii) and (iv),
$
\pi^{\textup c}(X)
=
\E\{D(1)-D(0)\mid X\}>0,
$ almost surely, so the complier CATE $\tau^{\textup{c}}(X)$ is well defined.
We consider four local causal estimands, the complier analogs of ATE, ATT, ATC, and ATO, which are special cases of the weighted local average treatment
effects \citep{choi2024instrumental}:
\begin{eqnarray*}
\tau^{\textup{c}}_{\textup{ATE}}
&=&
\E\left\{Y(1)-Y(0)\mid U=\textup{c}\right\}
=
\frac{
\E\left\{\pi^{\textup{c}}(X)\tau^{\textup{c}}(X)\right\}
}{
\E\left\{\pi^{\textup{c}}(X)\right\}
},
\\
\tau^{\textup{c}}_{\textup{ATT}}
&=&
\E\left\{Y(1)-Y(0)\mid Z=1,\,U=\textup{c}\right\}
=
\frac{
\E\left\{e(X)\pi^{\textup{c}}(X)\tau^{\textup{c}}(X)\right\}
}{
\E\left\{e(X)\pi^{\textup{c}}(X)\right\}
},
\\
\tau^{\textup{c}}_{\textup{ATC}}
&=&
\E\left\{Y(1)-Y(0)\mid Z=0,\,U=\textup{c}\right\}
=
\frac{
\E\left[\{1-e(X)\}\pi^{\textup{c}}(X)\tau^{\textup{c}}(X)\right]
}{
\E\left[\{1-e(X)\}\pi^{\textup{c}}(X)\right]
},
\\
\tau^{\textup{c}}_{\textup{ATO}}
&=&
\frac{
\E\left[
e(X)\{1-e(X)\}\pi^{\textup{c}}(X)\tau^{\textup{c}}(X)
\right]
}{
\E\left[
e(X)\{1-e(X)\}\pi^{\textup{c}}(X)
\right]
}.
\end{eqnarray*}
Here \(\tau^{\textup{c}}_{\textup{ATE}}\) is the local average treatment effect (LATE) introduced by \citet{imbens1994identification}, while
\(\tau^{\textup{c}}_{\textup{ATT}}\) and \(\tau^{\textup{c}}_{\textup{ATC}}\) are local analogs of ATT and ATC. As in the observational study, \(\tau^{\textup{c}}_{\textup{ATE}}\) is always bracketed by \(\tau^{\textup{c}}_{\textup{ATT}}\) and \(\tau^{\textup{c}}_{\textup{ATC}}\),  since it is a convex combination of the two:
\[
\tau^{\textup{c}}_{\textup{ATE}}
=\p(Z=1\mid U=\textup c)\cdot 
\tau^{\textup{c}}_{\textup{ATT}}
+
\p(Z=0\mid U=\textup c)\cdot 
\tau^{\textup{c}}_{\textup{ATC}} .
\]
The quantity \(\tau^{\textup{c}}_{\textup{ATO}}\) is the overlap-weighted LATE. The following theorem gives the local counterpart of Theorem~\ref{thm:diff-rep}.

\begin{theorem}\label{thm:diff-rep-local}
Assume Assumption~\ref{assump:iv}. 
\begin{enumerate}[label=(\roman*),leftmargin=*, align=left]
\item We have
\[
\tau^{\textup c}_{\textup{ATT}}
-
\tau^{\textup c}_{\textup{ATE}}
=
\frac{
\cov\{\tau^{\textup c}(X),e(X)\mid U=\textup c\}
}{
\E\{e(X)\mid U=\textup c\}
},
\qquad
\tau^{\textup c}_{\textup{ATE}}
-
\tau^{\textup c}_{\textup{ATC}}
=
\frac{
\cov\{\tau^{\textup c}(X),e(X)\mid U=\textup c\}
}{
\E\{1-e(X)\mid U=\textup c\}
}.
\]
\item We have
\begin{eqnarray*}
\tau^{\textup{c}}_{\textup{ATT}}
-
\tau^{\textup{c}}_{\textup{ATO}}
&=&
\frac{
\cov\left\{
\tau^{\textup{c}}(X),e(X)
\mid U=\textup{c},\,Z=1
\right\}
}{
\E\left\{
1-e(X)
\mid U=\textup{c},\,Z=1
\right\}
},
\\
\tau^{\textup{c}}_{\textup{ATO}}
-
\tau^{\textup{c}}_{\textup{ATC}}
&=&
\frac{
\cov\left\{
\tau^{\textup{c}}(X),e(X)
\mid U=\textup{c},\,Z=0
\right\}
}{
\E\left\{
e(X)
\mid U=\textup{c},\,Z=0
\right\}
}.
\end{eqnarray*}
\end{enumerate}
\end{theorem}
The identities in Theorem~\ref{thm:diff-rep-local}(i) show that the
ordering of the local ATT, ATE, and ATC is governed by the covariance between the complier CATE and the IV propensity score within the complier population.
The identities in Theorem~\ref{thm:diff-rep-local}(ii) show that the
bracketing relation for \(\tau^{\textup{c}}_{\textup{ATO}}\) is governed by the
signs of two conditional covariances. Parallel to
Proposition~\ref{prop:covariance-interpretation}, the marginal complier
covariance in part~(i) admits a selection-on-treatment-gains interpretation
and an exact decomposition in terms of the two conditional complier
covariances in part~(ii); see
Proposition~\ref{prop:covariance-interpretation-local} in the Supplementary
Material. Since the denominators are positive, these
identities immediately yield the following necessary and sufficient
characterization.

\begin{coro}\label{coro:exact-bracketing-local}
Assume Assumption~\ref{assump:iv}. 
\begin{enumerate}[label=(\roman*),leftmargin=*, align=left]
\item We have $ \tau^{\textup c}_{\textup{ATE}} \in  [ \tau^{\textup c}_{\textup{ATC}}, \tau^{\textup c}_{\textup{ATT}} ] $ if and only if $ \cov\{\tau^{\textup c}(X),e(X)\mid U=\textup c\}\ge 0$.
Similarly, $\tau^{\textup c}_{\textup{ATE}} \in  [ \tau^{\textup c}_{\textup{ATT}}, \tau^{\textup c}_{\textup{ATC}} ] $
if and only if $ \cov\{\tau^{\textup c}(X),e(X)\mid U=\textup c\}\le 0$.
\item We have $ \tau^{\textup c}_{\textup{ATO}} \in  [ \tau^{\textup c}_{\textup{ATC}}, \tau^{\textup c}_{\textup{ATT}} ] $ if and only if $ \cov\{\tau^{\textup c}(X),e(X)\mid U=\textup c,Z=z\}\ge 0 $ for $z\in\{0,1\}$.
Similarly, $\tau^{\textup c}_{\textup{ATO}} \in [ \tau^{\textup c}_{\textup{ATT}}, \tau^{\textup c}_{\textup{ATC}} ] $
if and only if $ \cov\{\tau^{\textup c}(X),e(X)\mid U=\textup c,Z=z\}\le 0$ for $z\in \{0,1\}$.
\end{enumerate}
\end{coro}

We next give a simple sufficient condition for the covariance signs in Corollary~\ref{coro:exact-bracketing-local}. Define the local CP function
\begin{align}\label{eq:def_gc}
g^{\textup{c}}(e)
=
\E\{\tau^{\textup{c}}(X)\mid U=\textup{c}, e(X)=e\}=
\frac{
\E\{\pi^{\textup{c}}(X)\tau^{\textup{c}}(X)\mid e(X)=e\}
}{
\E\{\pi^{\textup{c}}(X)\mid e(X)=e\}
},
\end{align}
the CATE among compliers whose IV propensity score
equals $e$.
The four local estimands above average \(\tau^{\textup{c}}(X)\) over the
complier distribution of \(X\), with weights given by different functions of
the IV propensity score \(e(X)\). Thus \(g^{\textup{c}}(e)\) plays the
same role in the IV setting as the CP function \(g(e)\) plays in the observational study setting.

The covariance representations in
Theorem~\ref{thm:diff-rep-local} can be signed through the local CP
function. The argument parallels that in \eqref{eq:balance_first}, but
requires the balancing property of the IV propensity score within the
complier population.

Assumption~\ref{assump:iv}(i) implies
\(Z\perp\!\!\!\perp U\mid X\), because
\(U=\{D(1),D(0)\}\) is determined by the potential treatments. Hence,
\[
\p(Z=1\mid X,U=\textup c)
=
\p(Z=1\mid X)
=
e(X).
\]
Taking conditional expectations given \(e(X)\) and \(U=\textup c\) gives
\[
\begin{aligned}
\p\{Z=1\mid e(X),U=\textup c\}
=
\E\!\left[
\p(Z=1\mid X,U=\textup c)
\mid e(X),U=\textup c
\right] =
e(X).
\end{aligned}
\]
Therefore, we have
\(
Z\perp\!\!\!\perp X\mid e(X),U=\textup c.
\)
That is, \(e(X)\) remains a balancing score for \(Z\) within the
complier population.
It follows that, for \(z\in\{0,1\}\),
\[
\begin{aligned}
\E\{\tau^{\textup c}(X)
\mid U=\textup c,Z=z,e(X)\}
=
\E\{\tau^{\textup c}(X)
\mid U=\textup c,e(X)\} =
g^{\textup c}\{e(X)\}.
\end{aligned}
\]
Applying the law of total covariance conditional on
\((U=\textup c,Z=z)\), and noting that \(e(X)\) is fixed after
conditioning on itself, yields
\begin{align}\label{eq:key_tauce_to_gc}
\cov\{\tau^{\textup c}(X),e(X)
\mid U=\textup c,Z=z\} =
\cov\{g^{\textup c}(e(X)),e(X)
\mid U=\textup c,Z=z\}.
\end{align}

Thus, if \(g^{\textup c}(e)\) is non-decreasing, the covariances in
\eqref{eq:key_tauce_to_gc} are nonnegative for \(z=0,1\); if
\(g^{\textup c}(e)\) is non-increasing, they are nonpositive.
Consequently, monotonicity of the local CP function provides a transparent
sufficient condition for the signs of the tilted covariances in
Corollary~\ref{coro:exact-bracketing-local}. We impose the following
monotonicity condition on \(g^{\textup c}\).
\begin{assumption}\label{assump:iv mono}
Suppose that
\begin{enumerate}[label=(\roman*),leftmargin=*, align=left]
    \item $g^{\textup{c}}(e)$ is a non-decreasing function, or
    \item $g^{\textup{c}}(e)$ is a non-increasing function.
\end{enumerate}
\end{assumption}

Assumption~\ref{assump:iv mono} is the LATE counterpart of
Assumption~\ref{assump:mono}. It holds whenever compliers with a higher (or lower) IV propensity score also tend to have larger treatment effects.
Assumption \ref{assump:iv mono} is weaker than requiring \(\tau^{\textup{c}}(X)\) to be
pointwise monotone in \(e(X)\), since it constrains only the conditional mean \(g^{\textup{c}}({e(X)})\)
under the complier distribution.

\begin{coro}\label{coro:ow in atc att local}
Assume Assumption~\ref{assump:iv}.  Under Assumption~\ref{assump:iv mono}(i), we have $\tau^{\textup{c}}_{\textup{ATO}}
\in
[\tau^{\textup{c}}_{\textup{ATC}},\tau^{\textup{c}}_{\textup{ATT}}].$
Under Assumption~\ref{assump:iv mono}(ii), we have
$\tau^{\textup{c}}_{\textup{ATO}}
\in
[\tau^{\textup{c}}_{\textup{ATT}},\tau^{\textup{c}}_{\textup{ATC}}].$
\end{coro}
We end this section with two remarks.
\begin{remark}
As in the observational study setting, Corollary~\ref{coro:ow in atc att local}
is a statement about causal estimands. Define the conditional reduced-form
and first-stage contrasts as
\[
\tau_Y(X)
=
\E(Y\mid Z=1,X)-\E(Y\mid Z=0,X),
\qquad
\tau_D(X)
=
\E(D\mid Z=1,X)-\E(D\mid Z=0,X).
\]
Under Assumption~\ref{assump:iv}, we have
$
\pi^{\textup c}(X)=\tau_D(X),
$
and the conditional Wald formula gives
$
\tau^{\textup c}(X)
=
{\tau_Y(X)}/{\tau_D(X)}.
$
Hence these local causal estimands are identified from the observable data
distribution, and the same results carry over to the corresponding
statistical estimands.
\end{remark}
\begin{remark}[Connection with TSLS]
The overlap-weighted LATE is closely related to the population TSLS
coefficient from the IV regression of \(Y\) on \((1,D,X)\), using \(Z\) as
the excluded IV.
Let \(\TSLS\) denote this coefficient. When the IV propensity score
\(e(X)\) lies in the linear span of \((1,X)\), \(\TSLS\) equals the overlap-weighted LATE
\(\tau^{\textup c}_{\textup{ATO}}\). Without this linearity, TSLS generally
contains additional terms induced by the discrepancy between \(e(X)\) and its
linear projection and need not be a causally interpretable weighted average.
Section~\ref{subsec:regression_iv} of the Supplementary Material gives the
formal decomposition.
\end{remark}

\section{From theory to practice: CP-plot and local CP-plot}\label{sec:cpplot}

% Motivated by the covariance representations above, we now introduce the
% CP-plot (Section \ref{subsec:cp_plot}) and its local counterpart (Section \ref{subsec:local_cp_plot}).
Motivated by the covariance representations above, we now introduce the
CP-plot (Section \ref{subsec:cp_plot}) and its local counterpart (Section \ref{subsec:local_cp_plot}). The proposal complements existing
diagnostics based on the implied weights of regression adjustment
\citep{chattopadhyay2023implied} by focusing instead on the relationship
between treatment-effect heterogeneity and the propensity score.

\subsection{The CP-plot in observational studies}\label{subsec:cp_plot}

Suppose that we observe $\{(Z_i,X_i,Y_i),i=1,\ldots,n\}$ from an
observational study with a binary treatment $Z$. We construct the CP-plot as
follows.

\begin{enumerate}[label=(\arabic*),leftmargin=*]
\item Estimate the propensity score and CATE for each unit:
$
\hat e_i=\hat e(X_i),
$ and $
\hat\tau_i=\hat\tau(X_i).
$

\item Plot the points \(\{(\hat e_i,\hat\tau_i),i=1,\ldots,n\}\). 
\item Add three weighted linear fits of \(\hat\tau_i\) on
\(\hat e_i\):
\(
\texttt{lm}(\hat\tau_i\sim 1+\hat e_i,\ \texttt{weights}=w_i)
\),
using
\(
w_i\in\{1,\hat e_i,1-\hat e_i\}.
\)
\end{enumerate}

In Step (1), we can estimate the propensity score by logistic regression of
\(Z\) on \(X\). We can estimate the CATE by outcome regressions; for example, 
$
\hat\tau(X_i)=\hat\mu_1(X_i)-\hat\mu_0(X_i),
$
where $\hat\mu_z(X_i)$ estimates $\E(Y\mid Z=z,X_i)$. Other estimators can be used depending on the application.

Step (2) visualizes the joint empirical distribution of
\(\hat e_i\) and \(\hat\tau_i\): with treated and control units distinguished, the distribution of points along the horizontal axis shows the overlap in the estimated propensity scores, the spread along the vertical axis shows treatment effect heterogeneity, and the overall pattern shows the relationship between the two.

Step~(3) has a simple population interpretation. By \eqref{eq:exax}, weighting the population linear projection of \(\tau(X)\) on
\(e(X)\) by \(e(X)\) is equivalent to computing the projection under the
treated covariate distribution, whereas weighting by \(1-e(X)\) is equivalent
to computing it under the control covariate distribution. Thus, the three
population counterparts of the fits in Step~(3) are the linear projections
under the overall, treated, and control covariate distributions, respectively.
Their slopes have the same signs as
\[
\cov\{\tau(X),e(X)\},
\qquad
\cov\{\tau(X),e(X)\mid Z=1\},
\qquad
\cov\{\tau(X),e(X)\mid Z=0\},
\]
respectively. Table~\ref{tab:cp-plot-lines} summarizes the corresponding
sample interpretations.

\begin{table}[H]
\centering
\small
\caption{Interpretation of the three linear fits in the CP-plot.
The listed inequalities are reversed when the corresponding fitted slope is
negative.}
\label{tab:cp-plot-lines}
\begin{tabularx}{\textwidth}{@{}llXX@{}}
\toprule
Fit & Weight & Population covariance & Positive slope implies \\
\midrule
Unweighted
& \(1\)
& \(\cov\{\tau(X),e(X)\}\)
& \(\hat\tau_{\mathrm{ATT}}
   \ge \hat\tau_{\mathrm{ATE}}
   \ge \hat\tau_{\mathrm{ATC}}\) \\

Treated-tilted
& \(\hat e_i\)
& \(\cov\{\tau(X),e(X)\mid Z=1\}\)
& \(\hat\tau_{\mathrm{ATT}}
   \ge \hat\tau_{\mathrm{ATO}}\) \\

Control-tilted
& \(1-\hat e_i\)
& \(\cov\{\tau(X),e(X)\mid Z=0\}\)
& \(\hat\tau_{\mathrm{ATO}}
   \ge \hat\tau_{\mathrm{ATC}}\) \\
\bottomrule
\end{tabularx}
\end{table}
Theorem~\ref{thm:sample-cp-geometry} in the Supplementary Material gives
the exact finite-sample counterpart: the three fitted slopes encode the
corresponding differences among the plug-in estimates, and the pairwise
intersections recover the plug-in estimates of ATT, ATO, and ATC.

Now we present additional properties of the CP-plot based on its population version. Let
\[
L_j(e)=\alpha_j+\beta_j e,
\qquad
j\in\{\mathrm{all},\mathrm{tr},\mathrm{co}\},
\]
denote the population weighted linear projections of \(\tau(X)\) on
\(e(X)\), where
\[
(\alpha_j,\beta_j)
=
\arg\min_{a,b}
\E\left[
w_j(X)\{\tau(X)-a-be(X)\}^2
\right],
\]
with
\(
w_{\mathrm{all}}(X)=1,
w_{\mathrm{tr}}(X)=e(X),\)
and 
\(
w_{\mathrm{co}}(X)=1-e(X).
\)
The three slopes are also linked: \(\beta_{\mathrm{all}}\) is a convex
combination of \(\beta_{\mathrm{tr}}\) and \(\beta_{\mathrm{co}}\).
Proposition~\ref{prop:cp-slope-convex} in the Supplementary Material gives
the exact weights and a proof.

For
\(\star\in\{\mathrm{ATE},\mathrm{ATT},\mathrm{ATO},\mathrm{ATC}\}\),
define
\[
\bar e_\star
=
\frac{\E\{h_\star(X)e(X)\}}
{\E\{h_\star(X)\}},
\]
where
\(
h_{\mathrm{ATE}}(X)=1,
h_{\mathrm{ATT}}(X)=e(X),
h_{\mathrm{ATO}}(X)=e(X)\{1-e(X)\},
\) and \(
h_{\mathrm{ATC}}(X)=1-e(X).
\)
Thus, \((\bar e_\star,\tau_\star)\) is the centroid of
\(\{e(X),\tau(X)\}\) under the corresponding population.

\begin{theorem}
\label{thm:population-cp-geometry}
Assume Assumption~\ref{assume:ignor} and that the three weighted linear
projections are well defined.
Whenever each corresponding pair of lines is distinct, their unique
intersections are
\[
L_{\mathrm{all}}\cap L_{\mathrm{tr}}
=
(\bar e_{\mathrm{ATT}},\tau_{\mathrm{ATT}}),
\qquad
L_{\mathrm{tr}}\cap L_{\mathrm{co}}
=
(\bar e_{\mathrm{ATO}},\tau_{\mathrm{ATO}}),
\qquad
L_{\mathrm{co}}\cap L_{\mathrm{all}}
=
(\bar e_{\mathrm{ATC}},\tau_{\mathrm{ATC}}).
\]
Moreover, the slopes of the three population linear projections satisfy:
\begin{enumerate}[label=(\roman*),leftmargin=*, align=left]
\item For the unweighted projection,
\[
\beta_{\mathrm{all}}
=
\frac{
\tau_{\mathrm{ATT}}-\tau_{\mathrm{ATE}}
}{
\bar e_{\mathrm{ATT}}-\bar e_{\mathrm{ATE}}
}
=
\frac{
\tau_{\mathrm{ATE}}-\tau_{\mathrm{ATC}}
}{
\bar e_{\mathrm{ATE}}-\bar e_{\mathrm{ATC}}
}.
\]

\item For the treated- and control-weighted projections,
\[
\beta_{\mathrm{tr}}
=
\frac{
\tau_{\mathrm{ATT}}-\tau_{\mathrm{ATO}}
}{
\bar e_{\mathrm{ATT}}-\bar e_{\mathrm{ATO}}
},
\qquad
\beta_{\mathrm{co}}
=
\frac{
\tau_{\mathrm{ATO}}-\tau_{\mathrm{ATC}}
}{
\bar e_{\mathrm{ATO}}-\bar e_{\mathrm{ATC}}
}.
\]
\end{enumerate}

\end{theorem}

Theorem~\ref{thm:population-cp-geometry} gives an exact geometric
interpretation of the three population CP lines. When the pairwise
intersections are distinct, their vertical coordinates are
\(\tau_{\mathrm{ATT}}\), \(\tau_{\mathrm{ATO}}\), and
\(\tau_{\mathrm{ATC}}\), respectively. Thus, the three fitted lines encode
the levels of the target estimands, rather than merely the signs of the
corresponding covariance comparisons. The sample analog of this result implies
that the pairwise intersections of the three fitted CP lines recover the
corresponding plug-in estimates of ATT, ATO, and ATC; see
Theorem~\ref{thm:sample-cp-geometry} in the Supplementary Material.

The slope identities provide both alternative difference representations and
a quantitative interpretation of the fitted slopes. Each estimand difference
equals the relevant slope multiplied by the horizontal separation between the
corresponding target-population mean propensity scores. Consequently, even
when the intersections are visually difficult to distinguish, the magnitudes
of the estimand differences remain easy to assess. If the three slopes are
small in absolute value, the differences among ATT, ATE, ATO, and ATC must
also be small. If a slope is large in absolute value, however, this alone does
not imply a large estimand difference; one must also examine the corresponding
distance among
\(\bar e_{\mathrm{ATT}}\), \(\bar e_{\mathrm{ATO}}\),
\(\bar e_{\mathrm{ATE}}\), and \(\bar e_{\mathrm{ATC}}\).

% Add \usepackage{tabularx} in the preamble.

\subsection{The local CP-plot in IV studies}\label{subsec:local_cp_plot}
The same idea applies to the IV setting after replacing the CATE
\(\tau(X)\) by its complier counterpart \(\tau^{\textup c}(X)\).
Suppose that we observe \((Z_i,D_i,X_i,Y_i)\), where \(Z_i\) is a binary IV
and \(D_i\) is a binary treatment. Recall \(e(X)\),
\(\pi^{\textup c}(X)\), \(\tau_Y(X)\), and \(\tau_D(X)\) defined in
Section~\ref{sec:iv}. Under Assumption~\ref{assump:iv},
\(\pi^{\textup c}(X)=\tau_D(X)\) and
\(\tau^{\textup c}(X)=\tau_Y(X)/\tau_D(X)\).
We construct the local CP-plot as follows.
\begin{enumerate}[label=(\arabic*),leftmargin=*]
\item Estimate the IV propensity score:
$\hat e_i=\hat e(X_i).$ Estimate conditional contrasts:
\[
\widehat\tau_Y(X_i)
=
\widehat{\E}(Y\mid Z=1,X_i)-\widehat{\E}(Y\mid Z=0,X_i),
\]
\[
\widehat\tau_D(X_i)
=
\widehat{\E}(D\mid Z=1,X_i)-\widehat{\E}(D\mid Z=0,X_i).
\]
Set
$
\hat\pi_i^{\textup c}=\widehat\tau_D(X_i),
$ and $
\hat\tau_i^{\textup c}
=
{\widehat\tau_Y(X_i)}/{\widehat\tau_D(X_i)}.
$

\item Plot the points \(\{(\hat e_i,\hat\tau_i^{\textup c}),i=1,\ldots,n\}\).

\item Add three weighted linear fits of \(\hat\tau_i^{\textup c}\) on
\(\hat e_i\): \(\texttt{lm}(\hat\tau_i^{\textup c}\sim 1+\hat e_i,\ 
\texttt{weights}=w_i),\)
using $w_i\in\{\hat \pi_i^\textup c, \hat \pi_i^\textup c \hat e_i, \hat \pi_i^\textup c(1-\hat e_i)\}$.
\end{enumerate}

The interpretation of these weighted fits parallels that in observational
studies, with the population covariate distribution replaced by the complier distribution. Table \ref{tab:local-cp-plot-lines} summarizes the three fits.

\begin{table}[H]
\centering
\small
\caption{Interpretation of the three weighted linear fits in the local
CP-plot. The listed inequalities are reversed when the corresponding
fitted slope is negative.}
\label{tab:local-cp-plot-lines}
\resizebox{\textwidth}{!}{
\begin{tabular}{llll}
\toprule
Linear fit & Weight & Population covariance & Positive fitted slope implies \\
\midrule
Complier-weighted
& \(\hat\pi_i^{\textup c}\)
& \(\cov\{\tau^{\textup c}(X),e(X)\mid U=\textup c\}\)
& \(\hat\tau^{\textup c}_{\textup{ATT}}
\ge \hat\tau^{\textup c}_{\textup{ATE}}
\ge \hat\tau^{\textup c}_{\textup{ATC}}\) \\
Encouraged-complier-weighted
& \(\hat\pi_i^{\textup c}\hat e_i\)
& \(\cov\{\tau^{\textup c}(X),e(X)\mid U=\textup c,Z=1\}\)
& \(\hat\tau^{\textup c}_{\textup{ATT}}
\ge \hat\tau^{\textup c}_{\textup{ATO}}\) \\
Unencouraged-complier-weighted
& \(\hat\pi_i^{\textup c}(1-\hat e_i)\)
& \(\cov\{\tau^{\textup c}(X),e(X)\mid U=\textup c,Z=0\}\)
& \(\hat\tau^{\textup c}_{\textup{ATO}}
\ge \hat\tau^{\textup c}_{\textup{ATC}}\) \\
\bottomrule
\end{tabular}
}
\end{table}
We use the weighted least-squares fits in Step (3) as a unified implementation
of the local covariance comparisons in Table~\ref{tab:local-cp-plot-lines}.
Theorem~\ref{thm:sample-local-cp-geometry} in the Supplementary Material
gives the exact sample-level justification, including both the slope
identities and the pairwise intersection results.

The local CP lines also admit an exact population geometry parallel to
Theorem~\ref{thm:population-cp-geometry}. Specifically, the
complier-weighted and encouraged-complier-weighted lines intersect at the
local ATT centroid, the encouraged- and unencouraged-complier-weighted lines
intersect at the local ATO centroid, and the unencouraged-complier-weighted
and complier-weighted lines intersect at the local ATC centroid. Their slopes
similarly equal the corresponding differences between local weighted
estimands divided by the horizontal distances between the target-population
mean IV propensity scores. To save space, we give the formal result in
Theorem~\ref{thm:population-local-cp-geometry} of
Section~\ref{subsec:population-local-cp-geometry} in the Supplementary
Material; its proof is given in Section~\ref{sec:proofs}.

\section{Empirical illustrations}\label{sec:app}

We apply the CP-plot to eight observational studies and the local CP-plot to an IV study, and compare the corresponding weighted estimates.

\subsection{Observational studies}
To complement the theoretical results in Section \ref{sec:obs}, we revisit eight empirical studies. One uses the right heart catheterization (rhc) data \citep{conners1996effectiveness,harrell2019package}, where the treatment (\texttt{swang1}) indicates whether a patient received right heart catheterization, and the outcome is in-hospital death (\texttt{death}). 
A second application,  \texttt{cavities\_smoking},  uses an NHANES III pediatric dental-caries dataset \citep{cdc2010nhanes}. We restrict attention to children aged 4--11 years with at least one primary tooth, define the treatment as secondhand smoking exposure (\texttt{secondhand.smoking.exposure}), coded as 1 when serum cotinine exceeds 1 ng/mL, and take the outcome to be the number of decayed and filled surfaces in the primary teeth (\texttt{dfs}). The remaining applications are \texttt{abortion}, \texttt{adult\_services} (with three treatment definitions),
\texttt{black\_politicians}, and \texttt{ccdrug}, all from the
\texttt{causaldata} package \citep{huntingtonklein2024causaldata}.
Table~\ref{tab:app_data} summarizes the eight treatment--outcome pairs.

\begin{table}
\centering
\footnotesize
\caption{Empirical settings used in Figure~\ref{fig:rhc_tau_prop}}
\label{tab:app_data}

\setlength{\tabcolsep}{4pt}
\begin{tabularx}{\linewidth}{
@{}
>{\raggedright\arraybackslash}p{3.2cm}
>{\raggedright\arraybackslash}p{3.5cm}
>{\raggedright\arraybackslash}X
>{\raggedright\arraybackslash}X
@{}
}
\toprule
Application & Data source / sample & Treatment definition & Outcome definition \\
\midrule

\texttt{rhc}
& Right Heart Catheterization study
& Right-heart catheterization (\texttt{swang1})
& In-hospital death (\texttt{death}) \\

\texttt{cavities\_smoking}
& NHANES III children aged 4--11 with at least one primary tooth
& Secondhand smoking exposure (\texttt{SHS})\textsuperscript{a}
& Number of decayed and filled surfaces in primary teeth (\texttt{dfs}) \\

\texttt{abortion}
& Abortion legalization and sexually transmitted infections
& Early repeal of abortion prohibition (\texttt{repeal})
& Log gonorrhea cases per 100,000 among 15--19-year-olds (\texttt{lnr}) \\

\texttt{adult\_services} (hot)
& Internet-mediated sex-worker survey
& Client met in a hotel (\texttt{hot})
& Log hourly wage (\texttt{lnw}) \\

\texttt{adult\_services} (reg)
& Internet-mediated sex-worker survey
& Client was a regular client (\texttt{reg})
& Log hourly wage (\texttt{lnw}) \\

\texttt{adult\_services} (unsafe)
& Internet-mediated sex-worker survey
& Unprotected sex with client (\texttt{unsafe})
& Log hourly wage (\texttt{lnw}) \\

\texttt{black\_politicians}
& Legislators' email-response field experiment
& Email sender was from out of district (\texttt{treat\_out})
& Legislator responded to the email (\texttt{responded}) \\

\texttt{ccdrug}
& Crown Court drug-arrest data
& Defendant was male (\texttt{male})
& Defendant was taken into custody (\texttt{custody}) \\
\bottomrule
\end{tabularx}

\vspace{1em}
\begin{minipage}{0.95\linewidth}
\footnotesize
\textsuperscript{a} \texttt{SHS} denotes the raw treatment variable
\texttt{secondhand.smoking.exposure}, coded as 1 if serum cotinine exceeds
1 ng/mL.
\end{minipage}
\end{table}

Figure~\ref{fig:rhc_tau_prop} presents the CP-plots for the eight treatment-outcome pairs. For
each, we estimate \(e(X)\) by logistic regression and
\(\E(Y\mid X,Z=z)\), \(z\in\{0,1\}\), by separate linear outcome regressions within the two treatment arms.
Each panel plots the unit-level pairs \((\hat e(X_i),\hat\tau(X_i))\) and adds the 
three weighted linear fits of Table~\ref{tab:cp-plot-lines}.

The signs of the fitted slopes estimate the signs of the covariances in Table~\ref{tab:cp-plot-lines}.  In the \texttt{abortion},
\texttt{adult\_services} (unsafe), \texttt{black\_politicians}, and
\texttt{cavities\_smoking} applications, the treated tilted and control tilted  slopes are positive, implying
$
\hat\tau_{\mathrm{ATC}}
\le
\hat\tau_{\mathrm{ATO}}
\le
\hat\tau_{\mathrm{ATT}}.
$
In contrast, in the \texttt{rhc}, \texttt{adult\_services} (hot), \texttt{adult\_services} (reg), and \texttt{ccdrug} applications, both tilted slopes are negative, implying the reverse ordering
$
\hat\tau_{\mathrm{ATT}}
\le
\hat\tau_{\mathrm{ATO}}
\le
\hat\tau_{\mathrm{ATC}}.
$
By Corollary~\ref{coro:exact-bracketing}, the slope signs provide graphical evidence for the empirical
bracketing relation between ATT, ATO, and ATC. Table~\ref{tab:ATO_all} reports the estimated ATT, ATO, and ATC for all eight treatment-outcome pairs.  In each case, \(\hat\tau_{\mathrm{ATO}}\) lies between \(\hat\tau_{\mathrm{ATT}}\) and \(\hat\tau_{\mathrm{ATC}}\), with the direction determined by the signs of the tilted slopes in Figure~\ref{fig:rhc_tau_prop}.

\begin{figure}[!htbp]
\centering
\includegraphics[width=\linewidth]{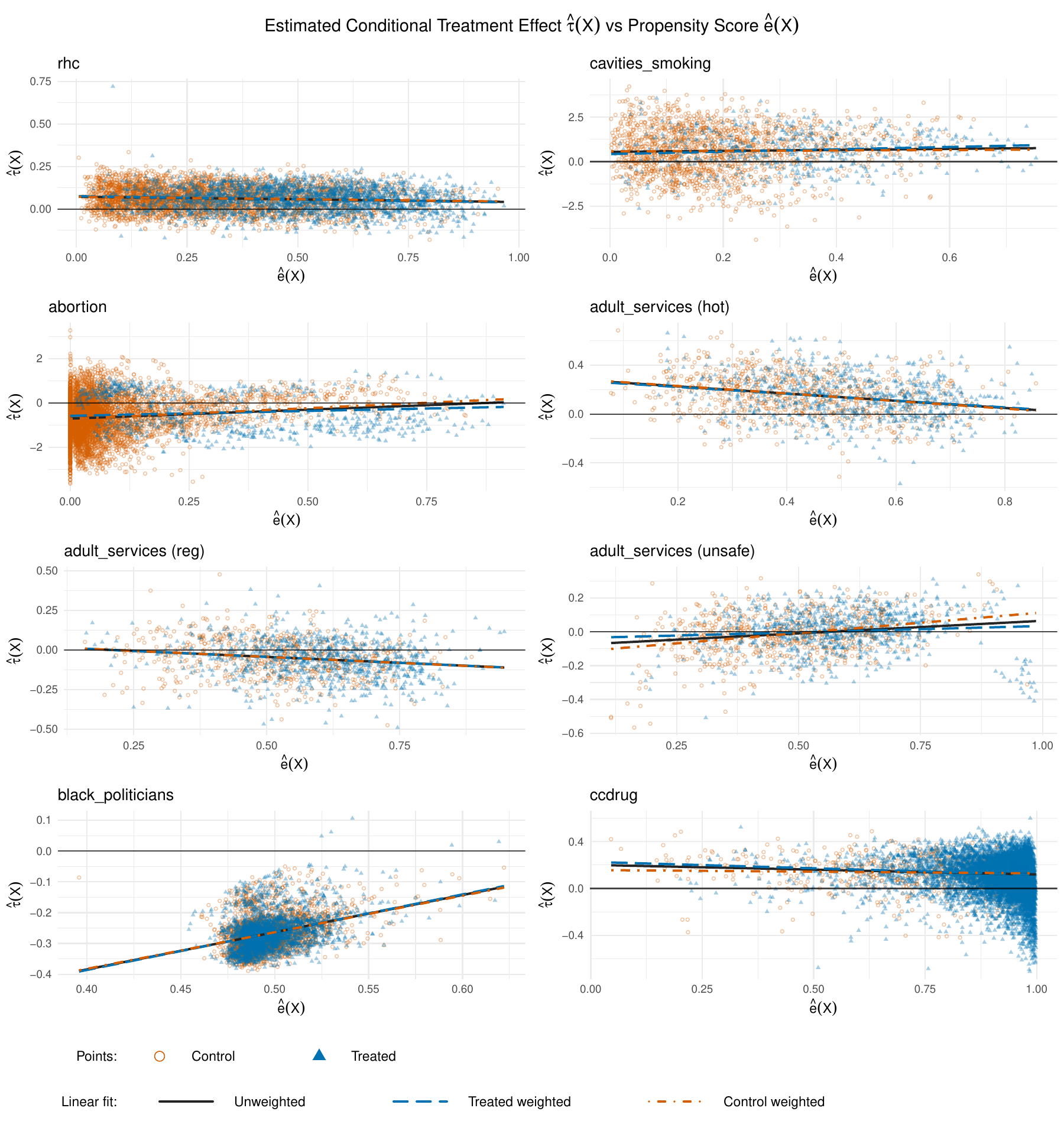}
\caption{CP-plots: estimated conditional treatment effects $\hat{\tau}(X)$ versus estimated propensity scores $\hat e(X)$ in eight empirical studies. Points are unit-level estimates, with treated and control units shown using different point types. The colored lines are three weighted linear fits of \(\hat\tau(X)\) on \(\hat e(X)\).}
\label{fig:rhc_tau_prop}
\end{figure}

\begin{table}[!htbp]
\centering
\small
\caption{Estimated average treatment effects for different target populations across eight empirical studies}
\label{tab:ATO_all}
\begin{tabular}{p{4.4cm}ccc}
\toprule
Application & $\hat\tau_{\mathrm{ATT}}$ & $\hat\tau_{\mathrm{ATO}}$ & $\hat\tau_{\mathrm{ATC}}$ \\
\midrule
\texttt{rhc} & 0.0576 & 0.0604 & 0.0635 \\
 & (0.0130) & (0.0125) & (0.0150) \\

\texttt{cavities\_smoking} & 0.6378 & 0.6140 & 0.6031 \\
 & (0.4396) & (0.4297) & (0.4947) \\

\texttt{abortion} & -0.4288 & -0.4643 & -0.6389 \\
 & (0.0401) & (0.0448) & (0.1143) \\

\texttt{adult\_services} (hot) & 0.1328 & 0.1446 & 0.1595 \\
 & (0.0257) & (0.0250) & (0.0267) \\

\texttt{adult\_services} (reg) & -0.0588 & -0.0526 & -0.0472 \\
 & (0.0259) & (0.0239) & (0.0243) \\

\texttt{adult\_services} (unsafe) & 0.0031 & -0.0010 & -0.0113 \\
 & (0.0252) & (0.0236) & (0.0246) \\

\texttt{black\_politicians} & -0.2666 & -0.2673 & -0.2678 \\
 & (0.0129) & (0.0129) & (0.0129) \\

\texttt{ccdrug} & 0.1252 & 0.1337 & 0.1348 \\
 & (0.0155) & (0.0115) & (0.0117) \\
\bottomrule
\end{tabular}

\vspace{2.5em}
\begin{minipage}{0.95\linewidth}
\footnotesize
Note. Bootstrap standard errors are shown in parentheses.
\end{minipage}
\end{table}
\subsection{An IV study}
To complement the IV results in Section \ref{sec:iv}, we revisit the 401(k) dataset analyzed by \citet{abadie2003semiparametric}. The data contain \(9{,}275\) cross-sectional observations with income and demographic information; the outcome is net financial assets (\texttt{nettfa}), the treatment is participation in a 401(k) plan (\texttt{p401k}), and the IV is eligibility for a 401(k) plan (\texttt{e401k}). As argued by \citet{abadie2003semiparametric}, eligibility is plausibly a valid IV conditional on income, age, and marital status.

For the local CP-plot, we estimate \(e(X)\) by logistic regression and \(\E(Y\mid X,Z=z)\) and \(\E(D\mid X,Z=z)\), \(z\in\{0,1\}\), by separate linear regressions within the two IV groups. All models control for family income, age, and marital status, including quadratic terms in income and age.
Figure~\ref{fig:401k_iv} plots the pairs \((\hat e(X), \hat\tau^{\textup c}(X))\) and adds the three weighted linear fits of Table~\ref{tab:local-cp-plot-lines}. All three fitted
slopes are positive: under each weighting, \(\hat\tau^{\textup c}(X)\) is positively associated with \(\hat e(X)\).

\begin{figure}[!htbp]
\centering
\includegraphics[width=.85\linewidth]{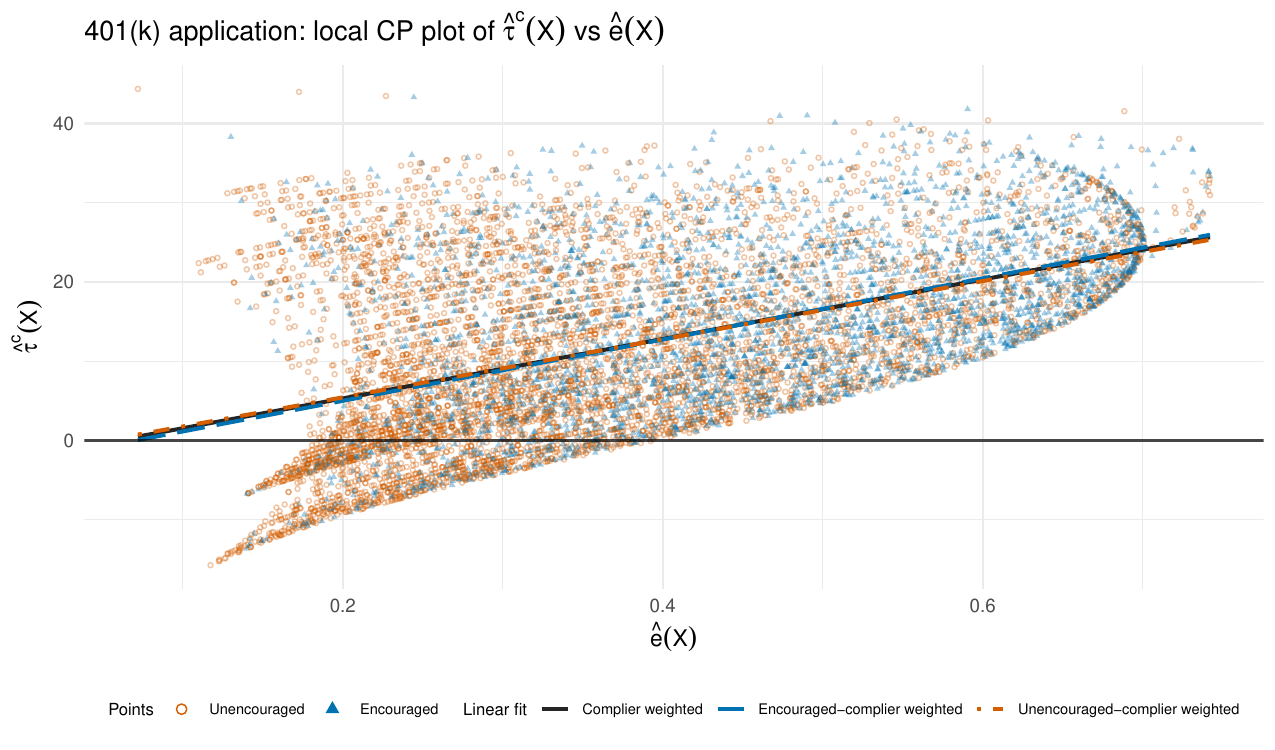}
\caption{Local CP-plot for the 401(k) application: estimated complier CATE \(\hat\tau^{\textup{c}}(X)\) versus estimated IV propensity scores \(\hat e(X)\). 
Points are unit-level plug-in estimates, with point types distinguishing the encouraged group (\(Z=1\))
from the unencouraged group (\(Z=0\)). The three colored lines are the weighted linear fits of \(\hat\tau^\textup c(X)\) on \(\hat e(X)\). }
\label{fig:401k_iv}
\end{figure}

We next estimate the local weighted estimands defined in
Section~\ref{sec:iv} using the same fitted conditional mean models used to
construct the local CP-plot. 
We use the plug-in estimator
\[
\hat\tau^{\textup c}_{h}
=
\frac{
\sum_{i=1}^n
\hat h(X_i)\hat\pi^{\textup c}(X_i)
\hat\tau^{\textup c}(X_i)
}{
\sum_{i=1}^n
\hat h(X_i)\hat\pi^{\textup c}(X_i)
}
=
\frac{
\sum_{i=1}^n
\hat h(X_i)\hat\tau_Y(X_i)
}{
\sum_{i=1}^n
\hat h(X_i)\hat\tau_D(X_i)
},
\]
with 
\(
\hat h(X)=\hat e(X),\hat e(X)\{1-\hat e(X)\},1-\hat e(X)
\)
for the local ATT, the overlap-weighted LATE, and the local ATC,
respectively.

The positive slopes of the fits weighted by
\(\hat\pi_i^{\textup c}\hat e_i\) and
\(\hat\pi_i^{\textup c}(1-\hat e_i)\) suggest that the two covariance terms
in Corollary~\ref{coro:exact-bracketing-local} are positive. The local CP-plot
therefore suggests the population ordering
\(
\tau^{\textup c}_{\textup{ATT}}
\ge
\tau^{\textup c}_{\textup{ATO}}
\ge
\tau^{\textup c}_{\textup{ATC}}.
\)

\begin{table}[H]
\centering
\caption{Plug-in estimates of local weighted treatment effects in the 401(k)
application}
\label{tab:401k_iv}
\begin{tabular}{lccc}
\toprule
& $\hat\tau^{\textup{c}}_{\textup{ATC}}$
& $\hat\tau^{\textup{c}}_{\textup{ATO}}$
& $\hat\tau^{\textup{c}}_{\textup{ATT}}$ \\
\midrule
Point estimate
& 11.5185
& 13.5882
& 15.2261 \\
Bootstrap standard error
& 2.1675
& 1.8803
& 2.3782 \\
\bottomrule
\end{tabular}
\end{table}
Table~\ref{tab:401k_iv} reports the corresponding plug-in estimates, which
follow the same ordering:
\(
\hat\tau^{\textup c}_{\textup{ATT}}
\ge
\hat\tau^{\textup c}_{\textup{ATO}}
\ge
\hat\tau^{\textup c}_{\textup{ATC}}.
\)
The estimated gaps are small relative to the bootstrap standard errors.
\section{Extension to beta weights}\label{sec:ext}
This section extends the covariance representations beyond ATE, ATT, ATC, and ATO. We study the beta-weight family, which nests these canonical estimands, and then present a general single-crossing formulation that covers other weighting schemes. For any nonnegative weight function \(w(X)\) with \(0<\E\{w(X)\}<\infty\), write $ \E_w(A)={\E\{w(X)A\}}/{\E\{w(X)\}}$ and $\cov_w(A,B)=\E_w(AB)-\E_w(A)\E_w(B)$ for the expectation and covariance under the distribution reweighted by \(w(X)\). When the weight function is notationally cumbersome, we write \(\cov(A,B;w)\) for \(\cov_w(A,B)\).
\subsection{Beta-weight estimands}\label{sec:beta}
 For real numbers $u,v\ge 1$, define the beta weights \citep{matsouaka2024causal}  \begin{align}\label{eq:betaweight}
\beta_{u,v}(X)=e(X)^{u-1}\{1-e(X)\}^{v-1}.
\end{align} 
The corresponding weighted average treatment effect is
\[
\tau_{u,v}=\frac{\E [\beta_{u,v}(X)\{Y(1)-Y(0)\}]}{\E\{\beta_{u,v}(X)\}}
=\frac{\E \{\beta_{u,v}(X)\tau(X)\}}{\E\{\beta_{u,v}(X)\}}.
\]
This family includes the four canonical estimands as special cases:
\[
\tau_{2,2}=\tauow,\qquad
\tau_{2,1}=\tauatt,\qquad
\tau_{1,2}=\tauatc,\qquad
\tau_{1,1}=\tauate.
\]

Each beta-weighted estimand averages \(\tau(X)\) under the target distribution induced by its normalized beta weight. Thus, comparing two beta-weighted estimands amounts to examining how their normalized weights reallocate mass across the propensity-score distribution. We focus on the comparison between \((u',v)\) and \((u,v)\); the comparison obtained by increasing \(v\) follows in parallel. Define the difference between the two
normalized beta weights as
\[
\Delta_{u,u';v}(X)
=
\frac{\beta_{u',v}(X)}
{\E\{\beta_{u',v}(X)\}}
-
\frac{\beta_{u,v}(X)}
{\E\{\beta_{u,v}(X)\}}.
\]
Define
\[
r_{u,u';v}
=
\left[
\frac{\E\{\beta_{u',v}(X)\}}
{\E\{\beta_{u,v}(X)\}}
\right]^{1/(u'-u)},
\qquad
r_{u;v,v'}
=
\left[
\frac{\E\{\beta_{u,v'}(X)\}}
{\E\{\beta_{u,v}(X)\}}
\right]^{1/(v'-v)}.
\]
We have
\(
\mathbf{1}\{\Delta_{u,u';v}(X)>0\}
=
\mathbf{1}\{e(X)>r_{u,u';v}\}.
\)
Thus, relative to the target population indexed by \((u,v)\), increasing
\(u\) reallocates normalized weight from units below the crossing point
\(r_{u,u';v}\) to units above it.
The crossing point determines the direction of the weight reallocation. To
also characterize its magnitude, define the positive weight
\(w_{u,u';v}(X)\) in the following Theorem \ref{thm:beta-one-coordinate-diff} so that
\[
\Delta_{u,u';v}(X)
=
\frac{w_{u,u';v}(X)}
{\E\{\beta_{u',v}(X)\}}
\{e(X)-r_{u,u';v}\}.
\]
In this decomposition, \(e(X)-r_{u,u';v}\) records the direction of the
reallocation: it is negative below the crossing point and positive above it.
The positive factor \(w_{u,u';v}(X)\) measures the local magnitude of the
reallocation. 
Because both beta weights are normalized,
\(\E\{\Delta_{u,u';v}(X)\}=0\). It implies
\(
r_{u,u';v}
=
\E_{w_{u,u';v}}\{e(X)\}.
\)
Hence, the crossing point is also the mean propensity score under the positive
weight that measures the intensity of the reallocation. It follows that
\[
\begin{aligned}
\tau_{u',v}-\tau_{u,v}
=
\E\{\Delta_{u,u';v}(X)\tau(X)\}=
\frac{\E\{w_{u,u';v}(X)\}}
{\E\{\beta_{u',v}(X)\}}
\cov\{\tau(X),e(X);w_{u,u';v}(X)\}.
\end{aligned}
\]
Thus, the covariance representation arises because the difference between
the two normalized beta weights is a zero-net-mass reallocation centered at
\(r_{u,u';v}\): the crossing point determines the direction of the
reallocation, and \(w_{u,u';v}(X)\) measures its local magnitude.

The argument for increasing \(v\) is analogous, with \(1-e(X)\) replacing
\(e(X)\). In that case, the crossing point on the propensity-score scale is
\(1-r_{u;v,v'}\), and normalized weight is reallocated from larger to smaller
propensity scores.

\begin{theorem}\label{thm:beta-one-coordinate-diff}
Assume Assumption~\ref{assume:ignor}.
\begin{enumerate}[label=(\roman*),leftmargin=*, align=left] 
\item Let \(u'>u\ge 1\) and \(v\ge 1\). Define
\[
w_{u,u';v}(X)
=
\begin{cases}
\displaystyle
\beta_{u,v}(X)
\frac{
e(X)^{u'-u}
-
r_{u,u';v}^{\,u'-u}
}{
e(X)-r_{u,u';v}
},
&
\displaystyle
e(X)\ne r_{u,u';v},
\\[1.4em]
\displaystyle
\beta_{u,v}(X)
(u'-u)r_{u,u';v}^{\,u'-u-1},
&
\displaystyle
e(X)=r_{u,u';v}.
\end{cases}
\]

Then \(w_{u,u';v}(X)>0\), and 
\[
\tau_{u',v}-\tau_{u,v}
=
\frac{
\E\{w_{u,u';v}(X)\}
}{
\E\{\beta_{u',v}(X)\}
}\cdot 
\cov\{\tau(X),e(X);w_{u,u';v}(X)\}.
\]

\item Let \(v'>v\ge 1\) and \(u\ge 1\). Define
\[
w_{u;v,v'}(X)
=
\begin{cases}
\displaystyle
\beta_{u,v}(X)
\frac{
r_{u;v,v'}^{\,v'-v}
-
\{1-e(X)\}^{v'-v}
}{
r_{u;v,v'}-\{1-e(X)\}
},
&
\displaystyle
e(X)\ne 1-r_{u;v,v'},
\\[1.4em]
\displaystyle
\beta_{u,v}(X)
(v'-v)r_{u;v,v'}^{\,v'-v-1},
&
\displaystyle
e(X)=1-r_{u;v,v'}.
\end{cases}
\]

Then \(w_{u;v,v'}(X)>0\), and
\[
\tau_{u,v}-\tau_{u,v'}
=
\frac{
\E\{w_{u;v,v'}(X)\}
}{
\E\{\beta_{u,v'}(X)\}
}\cdot \cov\{\tau(X),e(X);w_{u;v,v'}(X)\}.
\]
\end{enumerate} 
\end{theorem}
The second line in each definition gives the continuous extension at the
corresponding crossing point. In part~(i), the quotient is the divided
difference of the function \(x\mapsto x^{u'-u}\), evaluated at
\(x=e(X)\) and \(x=r_{u,u';v}\). It therefore converges to
\((u'-u)r_{u,u';v}^{\,u'-u-1}\) as
\(e(X)\to r_{u,u';v}\).
Part~(ii) follows analogously by considering the function
\(x\mapsto x^{v'-v}\), evaluated at
\(x=1-e(X)\) and \(x=r_{u;v,v'}\).
These crossing-point equalities may occur when the propensity score has
discrete support; under a continuous propensity-score distribution, they
typically hold with probability zero.

The comparison weights in Theorem~\ref{thm:beta-one-coordinate-diff} simplify for unit increments: 
\(w_{u,u+1;v}(X)=w_{u;v,v+1}(X)=\beta_{u,v}(X)\). Consequently, the same beta-weighted covariance determines both adjacent comparisons: 
\begin{eqnarray*}
\tau_{u+1,v}-\tau_{u,v}
&=&
\frac{
\E\{\beta_{u,v}(X)\}
}{
\E\{\beta_{u+1,v}(X)\}
}\cov\{\tau(X),e(X);\beta_{u,v}(X)\},
\\
\tau_{u,v}-\tau_{u,v+1}
&=&
\frac{
\E\{\beta_{u,v}(X)\}
}{
\E\{\beta_{u,v+1}(X)\}
}\cov\{\tau(X),e(X);\beta_{u,v}(X)\}.
\end{eqnarray*}

The next corollary extends Corollary~\ref{coro:ow in atc att} to the entire beta-weight family.

\begin{coro}\label{coro:beta weight}
Assume $e(X)\in(0,1)$. For real numbers $u,v\ge 1$, under Assumption~\ref{assump:mono}(i), the quantity $\tau_{u,v}$ is non-decreasing in $u$ and non-increasing in $v$; under Assumption~\ref{assump:mono}(ii), the quantity $\tau_{u,v}$ is non-increasing in $u$ and non-decreasing in $v$.
\end{coro}

Figure~\ref{fig:1} visualizes the orderings on the integer grid. 

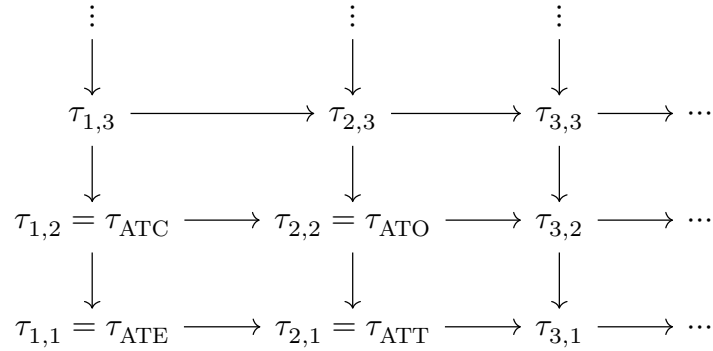
\begin{figure}[H]
    \centering
    \begin{tikzcd}
\vdots\arrow[d]&\vdots\arrow[d]&\vdots\arrow[d]&\\
\tau_{1,3}\arrow[d]\arrow[r]&\tau_{2,3}\arrow[d]\arrow[r]&\tau_{3,3}\arrow[d]\arrow[r]&\cdots\\
\tau_{1,2}=\tauatc\arrow[d]\arrow[r]&\tau_{2,2}=\tauow\arrow[d]\arrow[r]&\tau_{3,2}\arrow[d]\arrow[r]&\cdots\\
\tau_{1,1}=\tauate\arrow[r]&\tau_{2,1}=\tauatt\arrow[r]&\tau_{3,1}\arrow[r]&\cdots
    \end{tikzcd}
    \caption{Visualization of Corollary~\ref{coro:beta weight}. The notation ``$\rightarrow$'' stands for ``$\le$'' under Assumption~\ref{assump:mono}(i) and for ``$\ge$'' under Assumption~\ref{assump:mono}(ii).}
    \label{fig:1}
\end{figure}

\subsection{A generalized CP-plot for beta weights}
\label{subsec:general-cp-beta}
We extend the CP-plot to the beta-weight family: we
keep the scatter plot of \(\hat\tau_i\) against \(\hat e_i\) and add
weighted linear fits implementing the pairwise comparisons in Theorem~\ref{thm:beta-one-coordinate-diff}. 

For each unit $i$, define
$
\hat\beta_{u,v}^{(i)}
=
\hat e_i^{\,u-1}(1-\hat e_i)^{v-1}.$
For a comparison that increases the first beta parameter from \(u\) to
\(u'>u\), define
\[
\hat r_{u,u';v}
=
\left\{
\frac{\sum_{i=1}^n \hat\beta_{u',v}^{(i)}}
{\sum_{i=1}^n \hat\beta_{u,v}^{(i)}}
\right\}^{1/(u'-u)},
\qquad
\hat w_{u,u';v}^{(i)}
=
\hat\beta_{u,v}^{(i)}
\frac{
\hat e_i^{\,u'-u}
-
\hat r_{u,u';v}^{\,u'-u}
}{
\hat e_i-\hat r_{u,u';v}
},
\]
with $\hat w_{u,u';v}^{(i)}=\hat \beta_{u,v}^{(i)}(u'-u)\hat r_{u,u';v}^{u'-u-1}$ when
\(\hat e_i=\hat r_{u,u';v}\). For a comparison that increases the second beta
parameter from \(v\) to \(v'>v\), define
\[
\hat r_{u;v,v'}
=
\left\{
\frac{\sum_{i=1}^n \hat\beta_{u,v'}^{(i)}}
{\sum_{i=1}^n \hat\beta_{u,v}^{(i)}}
\right\}^{1/(v'-v)}
,\qquad
\hat w_{u;v,v'}^{(i)}
=
\hat\beta_{u,v}^{(i)}
\frac{
\hat r_{u;v,v'}^{\,v'-v}
-
(1-\hat e_i)^{v'-v}
}{
\hat r_{u;v,v'}-(1-\hat e_i)
},
\]
again with $\hat w_{u;v,v'}^{(i)}=\hat \beta_{u,v}^{(i)}(v'-v)\hat r_{u;v,v'}^{v'-v-1}$ when
\(1-\hat e_i=\hat r_{u;v,v'}\).

The generalized CP-plot for beta weights modifies Step (3) as follows:

\begin{enumerate}[label=(3),leftmargin=*]
\item For each comparison of interest, add a weighted linear fit of
\(\hat\tau_i\) on \(\hat e_i\):
\(
\texttt{lm}(\hat\tau_i\sim 1+\hat e_i,\ 
\texttt{weights}=w_i)
\),
using \(w_i=\hat w_{u,u';v}^{(i)}\) for comparing
\(\tau_{u',v}\) with \(\tau_{u,v}\), and
\(w_i=\hat w_{u;v,v'}^{(i)}\) for comparing
\(\tau_{u,v}\) with \(\tau_{u,v'}\).
\end{enumerate} 

The population slope of each weighted linear fit has the same sign as the
corresponding weighted covariance between \(\tau(X)\) and \(e(X)\). For
adjacent integer comparisons, the generalized weights simplify considerably:
both comparisons, \(\tau_{u+1,v}\) versus \(\tau_{u,v}\) and
\(\tau_{u,v}\) versus \(\tau_{u,v+1}\), are governed by
\(
\cov\{\tau(X),e(X);\beta_{u,v}(X)\}.
\)
Therefore, a single linear fit weighted by
\(\hat\beta_{u,v}^{(i)}\) diagnoses both comparisons. A positive fitted slope
suggests
\(
\tau_{u+1,v}\ge \tau_{u,v}\ge \tau_{u,v+1},
\)
whereas a negative fitted slope suggests the reverse ordering.

For the canonical comparison among ATT, ATO, and ATC, it suffices to consider the grid points \((u,v)=(2,1)\) and \((1,2)\), with beta weights \(\beta_{2,1}(X)=e(X)\) and \(\beta_{1,2}(X)=1-e(X)\): the corresponding fits are precisely the \(\hat e_i\)- and \((1-\hat e_i)\)-weighted fits of the CP-plot in Section~\ref{subsec:cp_plot}. A positive covariance under \(\beta_{2,1}\) implies \(\tau_{3,1}\ge\tauatt\ge\tauow\), and a positive covariance under \(\beta_{1,2}\) implies \(\tauow\ge\tauatc\ge\tau_{1,3}\). If both covariances are positive, we recover the bracketing relation \(\tauatc\le\tauow\le\tauatt\) of Corollary~\ref{coro:ow in atc att}, embedded in the longer chain
\[ \tau_{1,3}\le \tauatc\le \tauow\le \tauatt\le \tau_{3,1}. \]
Negative covariances reverse the corresponding inequalities.

The beta-weight comparisons and their generalized CP-plot interpretation
extend naturally to weighted local average treatment effects in the IV
framework. We omit the details here and provide
the corresponding results in Section~\ref{subsec:local_beta} in the
Supplementary Material.

\section{Summary}\label{sec:disc}
This paper studies how the relationship between treatment effect heterogeneity
and the propensity score shapes the ordering of weighted causal estimands. In observational studies, we build on known covariance representations for
ATT, ATE, and ATC and derive new analogous representations for the differences
between ATO and ATT and between ATO and ATC. These representations show that the relative positions of
the estimands are governed by covariances between the CATE and the propensity
score under appropriate target distributions. Monotonicity of the CP function
provides a transparent sufficient condition for signing these covariances and
hence for the bracketing relations. We establish analogous difference
representations and bracketing results for weighted local average treatment
effects in the IV setting, including the overlap-weighted LATE. We further extend the results to the beta-weight family.

The results suggest a practical diagnostic and reporting strategy. In
observational studies, analysts can use the CP-plot, which plots the estimated
CATE \(\hat\tau(X)\) against the estimated propensity score \(\hat e(X)\),
together with the unweighted, treated-tilted, and control-tilted linear fits
described in Section~\ref{subsec:cp_plot}. In IV studies, the local CP-plot plots plug-in estimates
\(\hat\tau^{\textup c}(X)\) against the estimated IV propensity score
\(\hat e(X)\), together with linear fits using weights
\(\hat\pi^{\textup c}(X)\),
\(\hat\pi^{\textup c}(X)\hat e(X)\), and
\(\hat\pi^{\textup c}(X)\{1-\hat e(X)\}\).
The fitted slopes summarize the empirical covariance signs in the
difference representations and therefore indicate the pairwise orderings of the weighted estimands. We therefore recommend reporting the CP-plot together with estimates of ATT,
ATO, and ATC in observational studies, and reporting the local CP-plot together
with estimates of
\(\tau^{\textup c}_{\textup{ATT}}\),
\(\tau^{\textup c}_{\textup{ATO}}\), and
\(\tau^{\textup c}_{\textup{ATC}}\)
in IV studies. Doing so makes explicit how causal conclusions vary with the
choice of target population.

\section*{Acknowledgements}
The authors thank Professors Macartan Humphreys and Hanzhong Liu for helpful discussions.
Fan Yang is also affiliated with the Yanqi Lake Beijing Institute of Mathematical Sciences and Applications.
Peng Ding acknowledges support from the U.S. National Science Foundation under Grant No.~2514234.

\bibliographystyle{main}
\bibliography{main}

@book{ding2024first,
  title={A first course in causal inference},
  author={Ding, Peng},
  year={2024},
  publisher={Chapman and Hall/CRC}
}

@article{zhao2026propensity,
  title={Propensity score weighted Cox regression for survival outcomes in observational studies with multiple or factorial treatments},
  author={Zhao, Zixian and Yang, Chengxin and Li, Fan},
  journal={arXiv preprint arXiv:2601.22572},
  year={2026}
}

@article{jiang2026improving,
    author = {Jiang, Liang and Linton, Oliver B. and Tang, Haihan and Zhang, Yichong},
    title = {Improving Estimation Efficiency via Regression-Adjustment in Covariate-Adaptive Randomizations with Imperfect Compliance},
    journal = {The Review of Economics and Statistics},
    volume = {108},
    number = {3},
    pages = {774-791},
    year = {2026},
    month = {05},
    issn = {0034-6535}
}

@article{ding2017instrumental,
  title={Instrumental variables as bias amplifiers with general outcome and confounding},
  author={Ding, Peng and VanderWeele, TJ and Robins, James M},
  journal={Biometrika},
  volume={104},
  number={2},
  pages={291--302},
  year={2017},
  publisher={Oxford University Press}
}

@article{conners1996effectiveness,
  title={The effectiveness of right heart catheterization in the initial care of critically ill patients},
  author={Conners Jr, AF and Speroff, T and Dawson, NV and Thomas, C and Harrell Jr, FE and Wagner, D and others},
  journal={Journal of the American Medical Association},
  volume={276},
  number={11},
  pages={889--897},
  year={1996}
}

@article{wallace2015doubly,
  title={Doubly-robust dynamic treatment regimen estimation via weighted least squares},
  author={Wallace, Michael P and Moodie, Erica EM},
  journal={Biometrics},
  volume={71},
  number={3},
  pages={636--644},
  year={2015},
  publisher={Oxford University Press}
}

@Manual{harrell2019package,
  title  = {{Hmisc}: Harrell Miscellaneous},
  author = {Harrell, Frank E., Jr.},
  year   = {2026},
  note   = {R package version 5.2-5},
  url    = {https://CRAN.R-project.org/package=Hmisc},
  doi    = {10.32614/CRAN.package.Hmisc}
}

@article{vanderweele2010signed,
  title={Signed directed acyclic graphs for causal inference},
  author={VanderWeele, Tyler J and Robins, James M},
  journal={Journal of the Royal Statistical Society Series B: Statistical Methodology},
  volume={72},
  number={1},
  pages={111--127},
  year={2010},
  publisher={Oxford University Press}
}

@article{vanderweele2008sign,
  title={The sign of the bias of unmeasured confounding},
  author={VanderWeele, Tyler J},
  journal={Biometrics},
  volume={64},
  number={3},
  pages={702--706},
  year={2008},
  publisher={Oxford University Press}
}

@article{chattopadhyay2023implied,
  title={On the implied weights of linear regression for causal inference},
  author={Chattopadhyay, Ambarish and Zubizarreta, Jos{\'e} R},
  journal={Biometrika},
  volume={110},
  number={3},
  pages={615--629},
  year={2023},
  publisher={Oxford University Press}
}

@article{sloczynski2022interpreting,
  title={Interpreting {OLS} estimands when treatment effects are heterogeneous: Smaller groups get larger weights},
  author={S{\l}oczy{\'n}ski, Tymon},
  journal={Review of Economics and Statistics},
  volume={104},
  number={3},
  pages={501--509},
  year={2022},
  publisher={MIT Press One Rogers Street, Cambridge, MA 02142-1209, USA journals-info~…}
}

@article{matsouaka2024causal,
  title={Causal inference in the absence of positivity: The role of overlap weights},
  author={Matsouaka, Roland A and Zhou, Yunji},
  journal={Biometrical Journal},
  volume={66},
  number={4},
  pages={2300156},
  year={2024},
  publisher={Wiley Online Library}
}

@article{humphreys2025bounds,
  title={Bounds on the fixed effects estimand in the presence of heterogeneous assignment propensities},
  author={Humphreys, Macartan},
  journal={Journal of Causal Inference},
  volume={13},
  number={1},
  pages={20240040},
  year={2025},
  publisher={De Gruyter}
}

@article{chiba2009sign,
  title={The sign of the unmeasured confounding bias under various standard populations},
  author={Chiba, Yasutaka},
  journal={Biometrical Journal},
  volume={51},
  number={4},
  pages={670--676},
  year={2009},
  publisher={Wiley Online Library}
}

@article{vansteelandt2022assumption,
  title={Assumption-lean inference for generalised linear model parameters},
  author={Vansteelandt, Stijn and Dukes, Oliver},
  journal={Journal of the Royal Statistical Society Series B: Statistical Methodology},
  volume={84},
  number={3},
  pages={657--685},
  year={2022},
  publisher={Oxford University Press}
}

@article{li2018balancing,
  title={Balancing covariates via propensity score weighting},
  author={Li, Fan and Morgan, Kari Lock and Zaslavsky, Alan M},
  journal={Journal of the American Statistical Association},
  volume={113},
  number={521},
  pages={390--400},
  year={2018},
  publisher={Taylor \& Francis}
}

@article{li2019addressing,
  title={Addressing extreme propensity scores via the overlap weights},
  author={Li, Fan and Thomas, Laine E and Li, Fan},
  journal={American Journal of Epidemiology},
  volume={188},
  number={1},
  pages={250--257},
  year={2019},
  publisher={Oxford University Press}
}

@Manual{huntingtonklein2024causaldata,
  title = {causaldata: Example Data Sets for Causal Inference Textbooks},
  author = {Nick Huntington-Klein and Malcolm Barrett},
  year = {2024},
  note = {R package version 0.1.4},
  url = {https://github.com/nickch-k/causaldata},
}

@article{ding2017principal,
  title={Principal stratification analysis using principal scores},
  author={Ding, Peng and Lu, Jiannan},
  journal={Journal of the Royal Statistical Society Series B: Statistical Methodology},
  volume={79},
  number={3},
  pages={757--777},
  year={2017},
  publisher={Oxford University Press}
}

@article{imbens1994identification,
	Author = {Imbens, Guido W and Angrist, Joshua D},
	Journal = {Econometrica},
	Number = {2},
	Pages = {467--475},
	Publisher = {JSTOR},
	Title = {Identification and estimation of local average treatment effects},
	Volume = {62},
	Year = {1994}}

@article{angrist1996identification,
	Author = {Angrist, Joshua D and Imbens, Guido W and Rubin, Donald B},
	Journal = {Journal of the American Statistical Association},
	Number = {434},
	Pages = {444--455},
	Publisher = {Taylor \&amp; Francis},
	Title = {Identification of causal effects using instrumental variables},
	Volume = {91},
	Year = {1996}}

@article{angrist1995two,
	Author = {Angrist, Joshua D and Imbens, Guido W},
	Journal = {Journal of the American Statistical Association},
	Number = {430},
	Pages = {431--442},
	Publisher = {Taylor \&amp; Francis},
	Title = {Two-stage least squares estimation of average causal effects in models with variable treatment intensity},
	Volume = {90},
	Year = {1995}}

@article{choi2024instrumental,
  title={Instrumental variable estimation of weighted local average treatment effects},
  author={Choi, Byeong Yeob},
  journal={Statistical Papers},
  volume={65},
  number={2},
  pages={737--770},
  year={2024},
  publisher={Springer}
}

@article{heckman1997making,
	Author = {Heckman, James J and Smith, Jeffrey and Clements, Nancy},
	Journal = {The Review of Economic Studies},
	Number = {4},
	Pages = {487--535},
	Publisher = {Oxford University Press},
	Title = {Making the most out of programme evaluations and social experiments: Accounting for heterogeneity in programme impacts},
	Volume = {64},
	Year = {1997}}

@article{bugni2023inference,
  title={Inference under covariate-adaptive randomization with imperfect compliance},
  author={Bugni, Federico A and Gao, Mengsi},
  journal={Journal of Econometrics},
  volume={237},
  number={1},
  pages={105497},
  year={2023},
  publisher={Elsevier}
}

@article{ding2021frisch,
  title={The Frisch--Waugh--Lovell theorem for standard errors},
  author={Ding, Peng},
  journal={Statistics \& Probability Letters},
  volume={168},
  pages={108945},
  year={2021},
  publisher={Elsevier}
}

@article{angrist1998estimating,
  title={Estimating the Labor Market Impact of Voluntary Military Service Using Social Security Data on Military Applicants},
  author={Angrist, Joshua D},
  journal={Econometrica},
  volume={66},
  number={2},
  pages={249--288},
  year={1998}
}

@article{rosenbaum1983central,
	Author = {Rosenbaum, Paul R and Rubin, Donald B},
	Journal = {Biometrika},
	Pages = {41--55},
	Publisher = {Biometrika Trust},
	Title = {The central role of the propensity score in observational studies for causal effects},
	Volume = {70},
	Year = {1983}}

@article{zhao2026two,
  title={Two-stage least squares with clustered data},
  author={Zhao, Anqi and Ding, Peng and Li, Fan},
  journal={arXiv preprint arXiv:2601.13507},
  year={2026}
}

@inproceedings{borusyak2024negative,
  title={Negative weights are no concern in design-based specifications},
  author={Borusyak, Kirill and Hull, Peter},
  booktitle={AEA Papers and Proceedings},
  volume={114},
  pages={597--600},
  year={2024},
  organization={American Economic Association 2014 Broadway, Suite 305, Nashville, TN 37203}
}

@article{abadie2003semiparametric,
  title={Semiparametric instrumental variable estimation of treatment response models},
  author={Abadie, Alberto},
  journal={Journal of Econometrics},
  volume={113},
  number={2},
  pages={231--263},
  year={2003},
  publisher={Elsevier}
}

@article{blandhol2022tsls,
    author = {Blandhol, Christine and Bonney, John and Mogstad, Magne and Torgovitsky, Alexander},
    title = {When is {TSLS} Actually {LATE}?},
    journal = {The Review of Economic Studies},
    pages = {rdag029},
    year = {2026},
    month = {05}
}

@misc{cdc2010nhanes,
  author       = {{Centers for Disease Control and Prevention}},
  title        = {{National Health and Nutrition Examination Survey Data}},
  howpublished = {Hyattsville, MD: U.S. Department of Health and Human Services, Centers for Disease Control and Prevention},
  year         = {2010},
  url          = {https://wwwn.cdc.gov/nchs/nhanes/},
}

@article{lee2018simple,
  title={Simple least squares estimator for treatment effects using propensity score residuals},
  author={Lee, Myoung-Jae},
  journal={Biometrika},
  volume={105},
  number={1},
  pages={149--164},
  year={2018},
  publisher={Oxford University Press}
}

\appendix
\phantomsection 
\setcounter{equation}{0}
\renewcommand{\theequation}{S\arabic{equation}}
\setcounter{table}{0}
\renewcommand{\thetable}{S\arabic{table}}
\setcounter{figure}{0}
\renewcommand{\thefigure}{S\arabic{figure}}
\setcounter{theorem}{0}
\renewcommand{\thetheorem}{S\arabic{theorem}}
\setcounter{lemma}{0}
\renewcommand{\thelemma}{S\arabic{lemma}}
\setcounter{condition}{0}
\renewcommand{\thecondition}{S\arabic{condition}}
\setcounter{assumption}{0}
\renewcommand{\theassumption}{S\arabic{assumption}}
\setcounter{remark}{0}
\renewcommand{\theremark}{S\arabic{remark}}
\setcounter{prop}{0}
\makeatletter
\renewcommand{\theprop}{S\arabic{prop}}
\makeatother
\setcounter{coro}{0}
\makeatletter
\renewcommand{\thecoro}{S\arabic{coro}}
\makeatother

\makeatletter
\providecommand{\theHtheorem}{S\arabic{theorem}}
\providecommand{\theHprop}{S\arabic{prop}}
\providecommand{\theHlemma}{S\arabic{lemma}}
\providecommand{\theHcoro}{S\arabic{coro}}
\providecommand{\theHassumption}{S\arabic{assumption}}
\makeatother

\noindent
\newpage

\renewcommand{\thepage}{S\arabic{page}} 
\setcounter{page}{1}

% Supplement numbering
\setcounter{theorem}{0}
\setcounter{prop}{0}
\setcounter{lemma}{0}
\setcounter{coro}{0}
\setcounter{assumption}{0}
\setcounter{figure}{0}
\setcounter{table}{0}

\renewcommand{\thetheorem}{S\arabic{theorem}}
\renewcommand{\theprop}{S\arabic{prop}}
\renewcommand{\thelemma}{S\arabic{lemma}}
\renewcommand{\thecoro}{S\arabic{coro}}
\renewcommand{\theassumption}{S\arabic{assumption}}
\renewcommand{\thefigure}{S\arabic{figure}}
\renewcommand{\thetable}{S\arabic{table}}

% Unique hyperref anchors
\makeatletter
\@ifundefined{theHtheorem}
  {\newcommand{\theHtheorem}{supp.thm.\arabic{theorem}}}
  {\renewcommand{\theHtheorem}{supp.thm.\arabic{theorem}}}

\@ifundefined{theHprop}
  {\newcommand{\theHprop}{supp.prop.\arabic{prop}}}
  {\renewcommand{\theHprop}{supp.prop.\arabic{prop}}}

\@ifundefined{theHlemma}
  {\newcommand{\theHlemma}{supp.lem.\arabic{lemma}}}
  {\renewcommand{\theHlemma}{supp.lem.\arabic{lemma}}}

\@ifundefined{theHcoro}
  {\newcommand{\theHcoro}{supp.coro.\arabic{coro}}}
  {\renewcommand{\theHcoro}{supp.coro.\arabic{coro}}}

\@ifundefined{theHassumption}
  {\newcommand{\theHassumption}{supp.assump.\arabic{assumption}}}
  {\renewcommand{\theHassumption}{supp.assump.\arabic{assumption}}}

\@ifundefined{theHfigure}
  {\newcommand{\theHfigure}{supp.fig.\arabic{figure}}}
  {\renewcommand{\theHfigure}{supp.fig.\arabic{figure}}}

\@ifundefined{theHtable}
  {\newcommand{\theHtable}{supp.tab.\arabic{table}}}
  {\renewcommand{\theHtable}{supp.tab.\arabic{table}}}
\makeatother

\begin{center}
\Huge
Supplementary Material
\end{center}

Section~\ref{sec:additional} presents the LATE counterpart of
Proposition~\ref{prop:covariance-interpretation}, the OLS and TSLS
interpretations of overlap-weighted estimands, the comparison with
\citet{humphreys2025bounds}, the sample-level geometry of the CP-plot, the
population and sample geometry of the local CP-plot, closed-form results under
a linear CP function, and local beta-weighted estimands in the IV setting.

Section~\ref{sec:proofs} provides proofs of the results in the main
paper and in the supplementary material.

\section{Additional results}\label{sec:additional}
\subsection{Interpretations of the local CATE--IV-propensity covariance}
\label{subsec:local-covariance-interpretation}

For \(d=0,1\), define the complier potential-outcome selection bias with
respect to IV status by
\begin{align}\label{eq:def-local-selection-bias}
B^{\textup c}(d)
=
\E\{Y(d)\mid Z=1,U=\textup c\}
-
\E\{Y(d)\mid Z=0,U=\textup c\}.
\end{align}
Because \(D=Z\) within the complier stratum, the two conditional means in
\eqref{eq:def-local-selection-bias} can equivalently be defined by conditioning
on \(D=1\) and \(D=0\), respectively. Let
\[
p_{\textup c}
=
\p(Z=1\mid U=\textup c)
=
\E\{e(X)\mid U=\textup c\}.
\]

\begin{prop}\label{prop:covariance-interpretation-local}
Assume Assumption~\ref{assump:iv}.
\begin{enumerate}[label=(\roman*),leftmargin=*,align=left]

\item
The marginal complier CATE--IV-propensity covariance admits the equivalent
representations
\begin{align}\label{eq:local-selection-on-gains}
\cov\{\tau^{\textup c}(X),e(X)\mid U=\textup c\}
=
&\cov\{Y(1)-Y(0),Z\mid U=\textup c\}\notag\\
=&
\cov\{Y(1)-Y(0),D\mid U=\textup c\}\notag\\
=&
\var(Z\mid U=\textup c)
\{B^{\textup c}(1)-B^{\textup c}(0)\}.
\end{align}

\item
The marginal complier covariance and the corresponding covariances within the
\(Z=1\) and \(Z=0\) complier groups satisfy
\begin{align}\label{eq:local-marginal-conditional-covariance}
&\cov\{\tau^{\textup c}(X),e(X)\mid U=\textup c\}
\notag\\
&\quad={}
\frac{\var(Z\mid U=\textup c)}
{\E[e(X)\{1-e(X)\}\mid U=\textup c]}
\Big[
 p_{\textup c}\cdot
 \cov\{\tau^{\textup c}(X),e(X)\mid U=\textup c,Z=1\}
\notag\\
&\hspace{49mm}
+
(1-p_{\textup c})\cdot
\cov\{\tau^{\textup c}(X),e(X)\mid U=\textup c,Z=0\}
\Big].
\end{align}

\end{enumerate}
\end{prop}

Proposition~\ref{prop:covariance-interpretation-local}(i) is the exact LATE
counterpart of Proposition~\ref{prop:covariance-interpretation}(i). Within the
complier population, the covariance between the complier CATE and the IV
propensity score measures whether units more likely to be encouraged also tend
to benefit more from treatment. Because treatment receipt equals IV assignment
for compliers, the same covariance is also the covariance between treatment
gains and treatment receipt within this principal stratum. Under an
unconditionally randomized IV, \(e(X)\) is constant and all these covariances
vanish.
Proposition~\ref{prop:covariance-interpretation-local}(ii) is the local
counterpart of Proposition~\ref{prop:covariance-interpretation}(ii). It shows
that a positively rescaled weighted average of the two IV-specific complier
covariances determines the marginal complier covariance. Hence, common
conditional signs imply the same marginal sign, whereas the converse need not
hold.

\subsection{Regression interpretations of overlap-weighted estimands}
\label{subsec:regression_obs}
\label{subsec:regression_iv}

\paragraph*{Observational-study interpretation through OLS.}
The overlap-weighted estimand \(\tauow\) has been studied extensively
\citep{li2018balancing,li2019addressing,zhao2026propensity}. To make its
connection with linear regression precise, define the population OLS
coefficients by
\begin{align}\label{eq:OLS}
\left(\beta_0,\beta_1,\beta_2\right)
=
\arg\min_{b_0,b_1,b_2}
\E\left(Y-b_0-b_1Z-b_2^\top X\right)^2.
\end{align}
We treat \eqref{eq:OLS} as a population linear projection and do not assume
that the corresponding additive linear outcome model is correctly specified.

\begin{prop}\label{prop:additive linear}
Assume Assumption~\ref{assume:ignor}.
\begin{enumerate}[label=(\roman*),leftmargin=*,align=left]
\item We have
\begin{align}\label{eq:beta1}
\beta_1
=
\frac{\E\left[\tilde w(X)\{Y(1)-Y(0)\}\right]}
{\E\{\tilde w(X)\}}
+
\frac{\E\left[\{e(X)-\tilde e(X)\}Y(0)\right]}
{\E\{\tilde w(X)\}},
\end{align}
where
\(\tilde w(X)=e(X)\{1-\tilde e(X)\}\) and
\(\tilde e(X)=\gamma_0+\gamma_1^\top X\) is the OLS projection of \(Z\) on
\((1,X)\), with
\[
(\gamma_0,\gamma_1)
=
\arg\min_{c_0,c_1}
\E\left(Z-c_0-c_1^\top X\right)^2.
\]
\item Furthermore, when \(e(X)\) is linear in \(X\), we have
\(\beta_1=\tauow\).
\end{enumerate}
\end{prop}

Part~(i) of Proposition~\ref{prop:additive linear} appears in
\citet{vansteelandt2022assumption}, and part~(ii) is well known; see, for
example, \citet{angrist1998estimating} and \citet{ding2024first}. More
broadly, overlap-weighted estimands arise from several regression
formulations; see \citet[Chapter~14]{ding2024first}.

A typical case in which \(e(X)\) is linear in \(X\) arises when \(X\)
contains a full set of indicators for the support of a discrete covariate. The
linearity then holds automatically, and the second term in
\eqref{eq:beta1} vanishes. This saturated finite-strata case underlies the
comparison with \citet{humphreys2025bounds} in
Section~\ref{subsec:compareHumphrey}.

\paragraph*{IV interpretation through TSLS.}
The overlap-weighted LATE admits a parallel interpretation through population
two-stage least squares (TSLS). Consider the population IV regression of \(Y\)
on \((1,D,X)\), using \(Z\) as the excluded IV for \(D\) and treating
\((1,X)\) as included exogenous regressors. Equivalently, the first stage
projects \(D\) on \((1,Z,X)\), and the second stage projects \(Y\) on
\((1,\widehat D,X)\), where \(\widehat D\) is the fitted value from the first
stage. Specifically, define
\[
(\theta_0,\theta_1,\theta_2)
=
\arg\min_{h_0,h_1,h_2}
\E\left(D-h_0-h_1Z-h_2^\top X\right)^2,
\]
and let
\(\widehat D=\theta_0+\theta_1Z+\theta_2^\top X\).
The TSLS estimand \(\TSLS\) is the coefficient on \(\widehat D\) in the
linear projection of \(Y\) on \((1,\widehat D,X)\). Recall from
\eqref{eq:OLS} that \(\beta_1\) is the coefficient on \(Z\) in the population
OLS regression of \(Y\) on \((1,Z,X)\). With a single IV for a single
endogenous treatment, the TSLS estimand is the ratio of the reduced-form and
first-stage coefficients:
\[
\TSLS=\frac{\beta_1}{\theta_1}.
\]

\begin{theorem}\label{thm:IV}
Assume \(\theta_1\ne0\). Under Assumption~\ref{assump:iv}, the following
statements hold:
\begin{enumerate}[label=(\roman*),leftmargin=*,align=left]
\item We have
\[
\TSLS
=
\frac{
\E[\tilde w(X)\{Y(1)-Y(0)\}\mid U=\textup c]\p(U=\textup c)
+
\E[\{e(X)-\tilde e(X)\}Y(D(0))]
}{
\E\{\tilde w(X)\mid U=\textup c\}\p(U=\textup c)
+
\E\{e(X)-\tilde e(X)\mid U=\textup a\}\p(U=\textup a)
},
\]
where \(\tilde w(X)=e(X)\{1-\tilde e(X)\}\) and \(\tilde e(X)\) is the
linear projection of \(Z\) on \((1,X)\).
\item Furthermore, when \(e(X)\) is linear in \(X\), we have
\[
\TSLS
=
\E\left[
\frac{e(X)\{1-e(X)\}\pi^{\textup c}(X)}
{\E[e(X)\{1-e(X)\}\pi^{\textup c}(X)]}
\tau^{\textup c}(X)
\right]
=
\tau^{\textup c}_{\textup{ATO}}.
\]
\end{enumerate}
\end{theorem}

Theorem~\ref{thm:IV} interprets additive TSLS without assuming that the
corresponding additive linear outcome model is correctly specified
\citep{angrist1995two,blandhol2022tsls}. Part~(i) shows that without linearity
of the IV propensity score, TSLS generally targets a ratio that need not have
a causal interpretation. The first term in the numerator is a weighted
average of complier treatment effects; the remaining numerator term and the
second denominator term arise from the gap between \(e(X)\) and its linear
projection \(\tilde e(X)\). Under the linearity condition in part~(ii), this
gap disappears and TSLS reduces to the overlap-weighted LATE.

\subsection{Comparison with Humphreys (2025)}
\label{subsec:compareHumphrey}

\citet{humphreys2025bounds} considers a stratified design with $K$ strata, where the propensity score is constant within each stratum. Let
\(S\in\{1,\ldots,K\}\) denote stratum membership, let $\nk$ denote the number of units in stratum $k$, $w_{[k]}=\nk/n$ the population share, and $e_{[k]}= \p(Z=1\mid S=k)\in(0,1)$ the fraction treated in stratum $k$. Define the stratum-specific ATE
$$
\tau_{[k]}=\E\{Y(1)-Y(0)\mid S=k\},\qquad k=1,\ldots,K.
$$
\citet{humphreys2025bounds} studies the population OLS regression of $Y$ on $\{Z,\mathbb I_{[1]},\ldots,\mathbb I_{[K]}\}$, where $\mathbb I_{[k]}$ indicates stratum membership. The coefficient on $Z$ defines the fixed-effects estimand
\[
\tau_{\mathrm{FE}}
=\frac{\sum_{k=1}^K w_{[k]} e_{[k]}(1-e_{[k]})\tau_{[k]}}
{\sum_{k=1}^K w_{[k]} e_{[k]}(1-e_{[k]})}.
\]
In the stratified setting, the function \(g(e)=\E\{\tau(X)\mid e(X)=e\}\)
takes the form
\[
g_{\mathrm{FE}}(e)
=
\frac{\sum_{k: e_{[k]}=e} w_{[k]} \tau_{[k]}}
{\sum_{k: e_{[k]}=e} w_{[k]}} ,
\]
that is, the average stratum-specific ATE among strata with
assignment propensity \(e\), weighted by their population shares. By Proposition \ref{prop:additive linear}, $\tau_{\mathrm{FE}}=\tauow$, and applying Corollary~\ref{coro:ow in atc att} yields the following corollary.
\begin{coro}\label{coro:FE-bounds}
Assume \(e_{[k]}\in(0,1)\) for all \(k\). If \(g_{\mathrm{FE}}(e)\) is
non-decreasing in \(e\), then
$\tau_{\mathrm{FE}}\in[\tauatc,\tauatt].$
If \(g_{\mathrm{FE}}(e)\) is non-increasing in \(e\), then $\tau_{\mathrm{FE}}\in[\tauatt,\tauatc].$
\end{coro}

For comparison, we recall Proposition~1 of \citet{humphreys2025bounds} as a corollary of Corollary \ref{coro:FE-bounds}. It imposes a monotonicity condition on the original fine strata. 

\begin{coro}[Proposition 1 of Humphreys (2025)]\label{coro:FE-bounds_old}
If for all $k,k'$, $e_{[k]}\ge e_{[k']}$ if and only if $\tau_{[k]}\ge\tau_{[k']}$, then
$\tau_{\mathrm{FE}}\in[\tauatc,\tauatt]$. If for all $k,k'$, $e_{[k]}\le e_{[k']}$ if and only if $\tau_{[k]}\ge\tau_{[k']}$, then $\tau_{\mathrm{FE}}\in[\tauatt,\tauatc]$.
\end{coro}

Corollary~\ref{coro:FE-bounds} weakens the condition in Corollary~\ref{coro:FE-bounds_old} as stated on the original fine stratification. In particular, Corollary~\ref{coro:FE-bounds_old} requires a monotone ordering of the stratum-level effects \(\tau_{[k]}\) themselves. Thus, if two fine strata have the same assignment propensity, the ``if and only if'' condition forces their treatment effects to be ordered in both directions and therefore to coincide. By contrast, Corollary~\ref{coro:FE-bounds} only requires monotonicity of the coarsened average effect \(g_{\mathrm{FE}}(e)\), and hence allows fine strata with the same assignment propensity to have different treatment effects. 
If the original fine stratification is coarsened by grouping strata with the
same assignment propensity, Corollary~\ref{coro:FE-bounds} becomes the
monotonicity condition on the coarsened strata. Let
\(\{r_1,\ldots,r_m\}\) be the distinct values of
\(\{e_{[k]}:k=1,\ldots,K\}\). For each \(r_i\), group all strata with
\(e_{[k]}=r_i\). The coarsened stratum has population share
\(
W_i=\sum_{k:e_{[k]}=r_i}w_{[k]},
\)
and average treatment effect
\[
\bar\tau_i
=
\frac{\sum_{k:e_{[k]}=r_i}w_{[k]}\tau_{[k]}}
{\sum_{k:e_{[k]}=r_i}w_{[k]}}
=
g_{\mathrm{FE}}(r_i).
\]
This coarsening leaves \(\tauatt\), \(\tauatc\), and
\(\tau_{\mathrm{FE}}\) unchanged. Hence monotonicity of \(g_{\mathrm{FE}}\) is
equivalent to monotonicity of the coarsened effects \(\bar\tau_i\) in the
coarsened propensities \(r_i\). Thus, relative to the original fine-stratum
statement, Corollary~\ref{coro:FE-bounds} allows strata with the same assignment
propensity to have different treatment effects. After coarsening, however, the
finite-strata comparison is expressed through the same object
\(g_{\mathrm{FE}}\). Our formulation makes this propensity-score conditional
mean explicit.

\paragraph*{Convex-combination relation among the population CP-plot slopes.}
Recall the population slopes \(\beta_{\mathrm{all}}\),
\(\beta_{\mathrm{tr}}\), and \(\beta_{\mathrm{co}}\) defined in
Section~\ref{subsec:cp_plot}.

\begin{prop}\label{prop:cp-slope-convex}
Assume Assumption~\ref{assume:ignor} and that the three population linear
projections are well defined. Then
\begin{align}\label{eq:cp-slope-convex}
\beta_{\mathrm{all}}
={}&
\frac{
\p(Z=1)\var\{e(X)\mid Z=1\}
}{
\p(Z=1)\var\{e(X)\mid Z=1\}
+
\p(Z=0)\var\{e(X)\mid Z=0\}
}
\beta_{\mathrm{tr}}
\notag\\
&+
\frac{
\p(Z=0)\var\{e(X)\mid Z=0\}
}{
\p(Z=1)\var\{e(X)\mid Z=1\}
+
\p(Z=0)\var\{e(X)\mid Z=0\}
}
\beta_{\mathrm{co}}.
\end{align}
In particular, \(\beta_{\mathrm{all}}\) lies between
\(\beta_{\mathrm{tr}}\) and \(\beta_{\mathrm{co}}\).
\end{prop}

\subsection{Sample-level interpretation of the CP-plot linear fits}
\label{subsec:sample-cp-lines}

This subsection provides a finite-sample justification for the weighted linear
fits used in the CP-plot. The result is purely algebraic. It treats the estimated
pairs \((\hat e_i,\hat\tau_i)\), \(i=1,\ldots,n\), as fixed numbers and does not
require any sampling approximation.

Assume \(0<\hat e_i<1\) for all \(i\). Define the plug-in estimators
\[
\hat\tau_{\mathrm{ATE}}
=
\frac1n\sum_{i=1}^n \hat\tau_i,
\qquad
\hat\tau_{\mathrm{ATT}}
=
\frac{\sum_{i=1}^n \hat e_i\hat\tau_i}
{\sum_{i=1}^n \hat e_i},
\]
\[
\hat\tau_{\mathrm{ATC}}
=
\frac{\sum_{i=1}^n (1-\hat e_i)\hat\tau_i}
{\sum_{i=1}^n (1-\hat e_i)},
\qquad
\hat\tau_{\mathrm{ATO}}
=
\frac{\sum_{i=1}^n \hat e_i(1-\hat e_i)\hat\tau_i}
{\sum_{i=1}^n \hat e_i(1-\hat e_i)}.
\]
Let
\[
\widehat L_j(e)=\widehat\alpha_j+\widehat\beta_j e,
\qquad
j\in\{\mathrm{all},\mathrm{tr},\mathrm{co}\},
\]
denote the weighted least-squares fits of \(\hat\tau_i\) on \(\hat e_i\),
using weights
\[
w_{\mathrm{all},i}=1,\qquad
w_{\mathrm{tr},i}=\hat e_i,\qquad
w_{\mathrm{co},i}=1-\hat e_i.
\]
For
\(\star\in\{\mathrm{ATE},\mathrm{ATT},\mathrm{ATO},\mathrm{ATC}\}\),
define
\[
\widehat{\bar e}_{\star}
=
\frac{\sum_{i=1}^n h_{\star,i}\hat e_i}
{\sum_{i=1}^n h_{\star,i}},
\]
where
\[
h_{\mathrm{ATE},i}=1,\qquad
h_{\mathrm{ATT},i}=\hat e_i,\qquad
h_{\mathrm{ATO},i}=\hat e_i(1-\hat e_i),\qquad
h_{\mathrm{ATC},i}=1-\hat e_i.
\]
Thus,
\((\widehat{\bar e}_{\star},\hat\tau_{\star})\) is the empirical
centroid under the corresponding weights.

\begin{theorem}
\label{thm:sample-cp-geometry}
Assume \(0<\hat e_i<1\) for all \(i\) and that the three weighted
least-squares fits are well defined. Whenever each corresponding pair of
lines is distinct, their unique intersections are
\[
\widehat L_{\mathrm{all}}\cap\widehat L_{\mathrm{tr}}
=
(\widehat{\bar e}_{\mathrm{ATT}},\hat\tau_{\mathrm{ATT}}),
\qquad
\widehat L_{\mathrm{tr}}\cap\widehat L_{\mathrm{co}}
=
(\widehat{\bar e}_{\mathrm{ATO}},\hat\tau_{\mathrm{ATO}}),
\]
\[
\widehat L_{\mathrm{co}}\cap\widehat L_{\mathrm{all}}
=
(\widehat{\bar e}_{\mathrm{ATC}},\hat\tau_{\mathrm{ATC}}).
\]
Moreover, the slopes of the three sample weighted least-squares fits
satisfy:
\begin{enumerate}[label=(\roman*),leftmargin=*,align=left]
\item For the unweighted fit,
\[
\widehat\beta_{\mathrm{all}}
=
\frac{
\hat\tau_{\mathrm{ATT}}-\hat\tau_{\mathrm{ATE}}
}{
\widehat{\bar e}_{\mathrm{ATT}}
-
\widehat{\bar e}_{\mathrm{ATE}}
}
=
\frac{
\hat\tau_{\mathrm{ATE}}-\hat\tau_{\mathrm{ATC}}
}{
\widehat{\bar e}_{\mathrm{ATE}}
-
\widehat{\bar e}_{\mathrm{ATC}}
}.
\]
\item For the treated- and control-weighted fits,
\[
\widehat\beta_{\mathrm{tr}}
=
\frac{
\hat\tau_{\mathrm{ATT}}-\hat\tau_{\mathrm{ATO}}
}{
\widehat{\bar e}_{\mathrm{ATT}}
-
\widehat{\bar e}_{\mathrm{ATO}}
},
\qquad
\widehat\beta_{\mathrm{co}}
=
\frac{
\hat\tau_{\mathrm{ATO}}-\hat\tau_{\mathrm{ATC}}
}{
\widehat{\bar e}_{\mathrm{ATO}}
-
\widehat{\bar e}_{\mathrm{ATC}}
}.
\]
\end{enumerate}
\end{theorem}

Theorem~\ref{thm:sample-cp-geometry} is the exact sample counterpart of
Theorem~\ref{thm:population-cp-geometry}. In particular, the signs of
\(\widehat\beta_{\mathrm{tr}}\) and \(\widehat\beta_{\mathrm{co}}\) determine
the directions of the two plug-in differences involving ATO, while the three
pairwise intersections recover the corresponding plug-in estimates of ATT,
ATO, and ATC.

\subsection{Population and sample geometry of the local CP-plot}
\label{subsec:population-local-cp-geometry}
\label{subsec:sample-local-cp-lines}

\paragraph*{Population geometry.}
We give the population counterpart of the three weighted linear fits used in
the local CP-plot. Let
\[
L_j^{\textup c}(e)
=
\alpha_j^{\textup c}+\beta_j^{\textup c}e,
\qquad
j\in\{\mathrm{all},\mathrm{tr},\mathrm{co}\},
\]
denote the weighted linear projections of
\(\tau^{\textup c}(X)\) on \(e(X)\), where
\[
(\alpha_j^{\textup c},\beta_j^{\textup c})
=
\arg\min_{a,b}
\E\left[
w_j^{\textup c}(X)
\{\tau^{\textup c}(X)-a-be(X)\}^2
\right],
\]
with
\[
w_{\mathrm{all}}^{\textup c}(X)
=
\pi^{\textup c}(X),
\qquad
w_{\mathrm{tr}}^{\textup c}(X)
=
e(X)\pi^{\textup c}(X),
\]
and
\[
w_{\mathrm{co}}^{\textup c}(X)
=
\{1-e(X)\}\pi^{\textup c}(X).
\]

For
\(\star\in\{\mathrm{ATE},\mathrm{ATT},\mathrm{ATO},\mathrm{ATC}\}\),
define
\[
\bar e_{\star}^{\textup c}
=
\frac{
\E\{h_\star(X)\pi^{\textup c}(X)e(X)\}
}{
\E\{h_\star(X)\pi^{\textup c}(X)\}
},
\]
where
\[
h_{\mathrm{ATE}}(X)=1,
\qquad
h_{\mathrm{ATT}}(X)=e(X),
\]
\[
h_{\mathrm{ATO}}(X)=e(X)\{1-e(X)\},
\qquad
h_{\mathrm{ATC}}(X)=1-e(X).
\]
Thus,
\[
\left(
\bar e_{\star}^{\textup c},
\tau_{\star}^{\textup c}
\right)
\]
is the centroid of
\(\{e(X),\tau^{\textup c}(X)\}\) under the corresponding local target
population.

\begin{theorem}
\label{thm:population-local-cp-geometry}
Assume Assumption~\ref{assump:iv} and that the three weighted linear
projections are well defined. Whenever each corresponding pair of lines is
distinct, their unique intersections are
\[
L_{\mathrm{all}}^{\textup c}
\cap
L_{\mathrm{tr}}^{\textup c}
=
\left(
\bar e_{\mathrm{ATT}}^{\textup c},
\tau_{\mathrm{ATT}}^{\textup c}
\right),
\quad
L_{\mathrm{tr}}^{\textup c}
\cap
L_{\mathrm{co}}^{\textup c}
=
\left(
\bar e_{\mathrm{ATO}}^{\textup c},
\tau_{\mathrm{ATO}}^{\textup c}
\right),
\quad
L_{\mathrm{co}}^{\textup c}
\cap
L_{\mathrm{all}}^{\textup c}
=
\left(
\bar e_{\mathrm{ATC}}^{\textup c},
\tau_{\mathrm{ATC}}^{\textup c}
\right).
\]

Moreover, the slopes of the three population linear projections satisfy:
\begin{enumerate}[label=(\roman*),leftmargin=*,align=left]
\item For the complier-weighted projection,
\[
\beta_{\mathrm{all}}^{\textup c}
=
\frac{
\tau_{\mathrm{ATT}}^{\textup c}
-
\tau_{\mathrm{ATE}}^{\textup c}
}{
\bar e_{\mathrm{ATT}}^{\textup c}
-
\bar e_{\mathrm{ATE}}^{\textup c}
}
=
\frac{
\tau_{\mathrm{ATE}}^{\textup c}
-
\tau_{\mathrm{ATC}}^{\textup c}
}{
\bar e_{\mathrm{ATE}}^{\textup c}
-
\bar e_{\mathrm{ATC}}^{\textup c}
}.
\]

\item For the encouraged- and unencouraged-complier-weighted projections,
\[
\beta_{\mathrm{tr}}^{\textup c}
=
\frac{
\tau_{\mathrm{ATT}}^{\textup c}
-
\tau_{\mathrm{ATO}}^{\textup c}
}{
\bar e_{\mathrm{ATT}}^{\textup c}
-
\bar e_{\mathrm{ATO}}^{\textup c}
},
\qquad
\beta_{\mathrm{co}}^{\textup c}
=
\frac{
\tau_{\mathrm{ATO}}^{\textup c}
-
\tau_{\mathrm{ATC}}^{\textup c}
}{
\bar e_{\mathrm{ATO}}^{\textup c}
-
\bar e_{\mathrm{ATC}}^{\textup c}
}.
\]
\end{enumerate}
\end{theorem}

\paragraph*{Sample-level interpretation.}
We next give the finite-sample counterpart for the local CP-plot in the
IV setting. As in Section~\ref{subsec:sample-cp-lines}, the result is purely
algebraic. It treats the estimated triples
\[
(\hat e_i,\hat\pi_i^{\textup c},\hat\tau_i^{\textup c}),
\qquad i=1,\ldots,n,
\]
as fixed numbers. Here \(\hat e_i\) is the estimated IV propensity score,
\(\hat\pi_i^{\textup c}\) is the estimated principal score for compliers, and
\(\hat\tau_i^{\textup c}\) is the estimated conditional complier treatment
effect.

Assume \(0<\hat e_i<1\) and \(\hat\pi_i^{\textup c}>0\) for all \(i\). Define
the plug-in local estimators as
\[
\hat\tau^{\textup c}_{\mathrm{ATE}}
=
\frac{\sum_{i=1}^n \hat\pi_i^{\textup c}\hat\tau_i^{\textup c}}
{\sum_{i=1}^n \hat\pi_i^{\textup c}},
\qquad
\hat\tau^{\textup c}_{\mathrm{ATT}}
=
\frac{\sum_{i=1}^n \hat e_i\hat\pi_i^{\textup c}\hat\tau_i^{\textup c}}
{\sum_{i=1}^n \hat e_i\hat\pi_i^{\textup c}},
\]
\[
\hat\tau^{\textup c}_{\mathrm{ATC}}
=
\frac{\sum_{i=1}^n (1-\hat e_i)\hat\pi_i^{\textup c}\hat\tau_i^{\textup c}}
{\sum_{i=1}^n (1-\hat e_i)\hat\pi_i^{\textup c}},
\]
and
\[
\hat\tau^{\textup c}_{\mathrm{ATO}}
=
\frac{
\sum_{i=1}^n
\hat e_i(1-\hat e_i)\hat\pi_i^{\textup c}\hat\tau_i^{\textup c}
}{
\sum_{i=1}^n
\hat e_i(1-\hat e_i)\hat\pi_i^{\textup c}
}.
\]
Let
\[
\widehat L_j^{\textup c}(e)
=
\widehat\alpha_j^{\textup c}
+
\widehat\beta_j^{\textup c}e,
\qquad
j\in\{\mathrm{all},\mathrm{tr},\mathrm{co}\},
\]
denote the weighted least-squares fits of
\(\hat\tau_i^{\textup c}\) on \(\hat e_i\), using weights
\[
w_{\mathrm{all},i}^{\textup c}=\hat\pi_i^{\textup c},
\qquad
w_{\mathrm{tr},i}^{\textup c}=\hat e_i\hat\pi_i^{\textup c},
\qquad
w_{\mathrm{co},i}^{\textup c}
=(1-\hat e_i)\hat\pi_i^{\textup c}.
\]
For
\(\star\in\{\mathrm{ATE},\mathrm{ATT},\mathrm{ATO},\mathrm{ATC}\}\),
define
\[
\widehat{\bar e}_{\star}^{\textup c}
=
\frac{
\sum_{i=1}^n h_{\star,i}\hat\pi_i^{\textup c}\hat e_i
}{
\sum_{i=1}^n h_{\star,i}\hat\pi_i^{\textup c}
},
\]
where
\[
h_{\mathrm{ATE},i}=1,\qquad
h_{\mathrm{ATT},i}=\hat e_i,\qquad
h_{\mathrm{ATO},i}=\hat e_i(1-\hat e_i),\qquad
h_{\mathrm{ATC},i}=1-\hat e_i.
\]
Thus,
\[
\left(
\widehat{\bar e}_{\star}^{\textup c},
\hat\tau_{\star}^{\textup c}
\right)
\]
is the corresponding empirical local weighted centroid.

\begin{theorem}
\label{thm:sample-local-cp-geometry}
Assume \(0<\hat e_i<1\) and \(\hat\pi_i^{\textup c}>0\) for all \(i\), and
that the three weighted least-squares fits are well defined. Whenever each
corresponding pair of lines is distinct, their unique intersections are
\[
\widehat L_{\mathrm{all}}^{\textup c}
\cap
\widehat L_{\mathrm{tr}}^{\textup c}
=
\left(
\widehat{\bar e}_{\mathrm{ATT}}^{\textup c},
\hat\tau_{\mathrm{ATT}}^{\textup c}
\right),
\]
\[
\widehat L_{\mathrm{tr}}^{\textup c}
\cap
\widehat L_{\mathrm{co}}^{\textup c}
=
\left(
\widehat{\bar e}_{\mathrm{ATO}}^{\textup c},
\hat\tau_{\mathrm{ATO}}^{\textup c}
\right),
\]
\[
\widehat L_{\mathrm{co}}^{\textup c}
\cap
\widehat L_{\mathrm{all}}^{\textup c}
=
\left(
\widehat{\bar e}_{\mathrm{ATC}}^{\textup c},
\hat\tau_{\mathrm{ATC}}^{\textup c}
\right).
\]
Moreover, the slopes of the three sample weighted least-squares fits
satisfy:
\begin{enumerate}[label=(\roman*),leftmargin=*,align=left]
\item For the complier-weighted fit,
\[
\widehat\beta_{\mathrm{all}}^{\textup c}
=
\frac{
\hat\tau_{\mathrm{ATT}}^{\textup c}
-
\hat\tau_{\mathrm{ATE}}^{\textup c}
}{
\widehat{\bar e}_{\mathrm{ATT}}^{\textup c}
-
\widehat{\bar e}_{\mathrm{ATE}}^{\textup c}
}
=
\frac{
\hat\tau_{\mathrm{ATE}}^{\textup c}
-
\hat\tau_{\mathrm{ATC}}^{\textup c}
}{
\widehat{\bar e}_{\mathrm{ATE}}^{\textup c}
-
\widehat{\bar e}_{\mathrm{ATC}}^{\textup c}
}.
\]
\item For the encouraged- and unencouraged-complier-weighted fits,
\[
\widehat\beta_{\mathrm{tr}}^{\textup c}
=
\frac{
\hat\tau_{\mathrm{ATT}}^{\textup c}
-
\hat\tau_{\mathrm{ATO}}^{\textup c}
}{
\widehat{\bar e}_{\mathrm{ATT}}^{\textup c}
-
\widehat{\bar e}_{\mathrm{ATO}}^{\textup c}
},
\qquad
\widehat\beta_{\mathrm{co}}^{\textup c}
=
\frac{
\hat\tau_{\mathrm{ATO}}^{\textup c}
-
\hat\tau_{\mathrm{ATC}}^{\textup c}
}{
\widehat{\bar e}_{\mathrm{ATO}}^{\textup c}
-
\widehat{\bar e}_{\mathrm{ATC}}^{\textup c}
}.
\]
\end{enumerate}
\end{theorem}

Theorem~\ref{thm:sample-local-cp-geometry} is the exact sample counterpart
of Theorem~\ref{thm:population-local-cp-geometry}. Its slope identities and
intersection results provide the finite-sample algebraic justification for
the three weighted fits in the local CP-plot.

\subsection{Closed-form results under a linear CP function}

We next consider a special case in which the CP function is linear.

\begin{assumption}\label{assump:linear}
Suppose that
$
g(e)=a+be$
for some constants \(a\) and \(b\).
\end{assumption}

Assumption~\ref{assump:linear} is a special case of Assumption~\ref{assump:mono}, with $b\ge 0$ corresponding to (i) and $b\le0$ corresponding to (ii). Under Assumption~\ref{assump:linear}, $\tauow$ is a convex combination of $\tauatt$ and $\tauatc$.

\begin{prop}\label{prop:linear-ow}
Assume $e(X)\in(0,1)$ and Assumption~\ref{assump:linear}. Then
\begin{align}\label{eq:ow= l atc+ 1-l att}
\tauow=\lambda\cdot \tauatt+(1-\lambda)\cdot \tauatc,
\end{align}
where $\lambda\in[0,1]$. When $e(X)$ is constant, we have $\tauow=\tauatt=\tauatc$; when $e(X)$ is not constant, we have 
\[
\lambda
=\frac{\E\{ e(X)\}\E\{ 1-e(X)\}}
{\E \{e(X)(1-e(X))\}\var\{e(X)\}}
\E[\{1-e(X)\}\{e(X)-c^*\}^2]\in [0, 1],
\]
where
\[
c^*
=
\arg\min_{c\in\mathbb R}
\E\big[\{1-e(X)\}\{e(X)-c\}^2\big]=
\frac{\E [e(X)\{1-e(X)\}]}{\E\{1-e(X)\}}.
\]
\end{prop}
The displayed formula for \(\lambda\) makes \(\lambda\ge 0\) immediate. To see
that \(\lambda\le 1\), interchange \(e(X)\) and \(1-e(X)\). Then \(1-\lambda\)
has the same form:
\[
1-\lambda
=
\frac{\E\{ e(X)\}\E\{ 1-e(X)\}}
{\E \{e(X)(1-e(X))\}\var\{e(X)\}}
\E\big[e(X)\{1-e(X)-c^\dagger\}^2\big],
\]
where
$c^\dagger
={\E[e(X)\{1-e(X)\}]}/{p}.
$
Hence \(1-\lambda\ge 0\), and therefore \(\lambda\in[0,1]\).

In the fixed-effects setting, \citet{humphreys2025bounds} derives a similar result, where \(e(X)\) takes finitely many values \(e_{[k]}\), and the expectations above reduce to finite weighted sums over strata. Under a linear relation between the stratum-level average effects and the assignment shares, the resulting mixing weight is the population propensity-score version of the finite-strata expression in \citet[Equation~(11)]{humphreys2025bounds}.
Proposition~\ref{prop:linear-ow} is related to, but distinct from, the result in \citet{sloczynski2022interpreting}, which studies the coefficient of \(Z\) from the population least-squares regression of \(Y\) on an intercept, \(Z\), and \(X\). Under his conditions, the coefficient of $Z$ also admits a convex-combination representation in $\tauatt$ and $\tauatc$, but with a mixing weight that depends primarily on the overall treated and control shares, where the smaller group receives more weight. Here, in contrast, the mixing weight \(\lambda\) depends on the full propensity-score distribution through
$
\E[\{1-e(X)\}\{e(X)-c^*\}^2],
$ so \(\tauow\) depends not only on the proportion of treated, but also on how propensity score varies across the population. Two populations with the same treated share can therefore produce different values of \(\lambda\) if their propensity-score distributions differ.

Under the linear specification in Assumption~\ref{assump:linear}, adjacent beta-weighted estimands satisfy a convex-combination relation.

\begin{prop}\label{prop:linear beta}
Assume \(e(X)\in(0,1)\) and Assumption~\ref{assump:linear}. For real numbers
\(u,v\ge1\),
\begin{align}\label{eq:tau u+1 v+1 lambda tau u+1 v +(1-lambda) tau u, v+1}
\tau_{u+1,v+1}
=\lambda\cdot \tau_{u+1,v}+(1-\lambda)\cdot \tau_{u,v+1},
\end{align}
where $\lambda\in[0,1]$. When $e(X)$ is constant, we have $\tau_{u+1,v+1}=\tau_{u+1,v}=\tau_{u,v+1}$;
when $e(X)$ is not constant, we have
\begin{align}\label{eq:lambda beta}
\lambda=\frac{\dfrac{\E[\beta_{u+2,v+1}(X)]}{\E[\beta_{u+1,v+1}(X)]}-\dfrac{\E[\beta_{u+1,v+1}(X)]}{\E[\beta_{u,v+1}(X)]}}{\dfrac{\E[\beta_{u+2,v}(X)]}{\E[\beta_{u+1,v}(X)]}-\dfrac{\E[\beta_{u+1,v+1}(X)]}{\E[\beta_{u,v+1}(X)]}}\in [0, 1].
\end{align}
\end{prop}
The weight \(\lambda\) depends only on the distribution of the propensity score.
It measures where the beta-weighted mean propensity score under
\(\beta_{u+1,v+1}\) lies between the corresponding means under
\(\beta_{u,v+1}\) and \(\beta_{u+1,v}\). The ordering of these
beta-weighted mean propensity scores ensures that the interpolation is convex,
so \(\lambda\in[0,1]\).

\subsection{Local beta weights in the IV setting}\label{subsec:local_beta}
The beta-weight construction also extends to local causal effects. For real numbers $u,v\ge 1$, define
\[
\tau^{\textup{c}}_{u,v}
=\frac{\E\!\left[\beta_{u,v}(X)\,\pi^{\textup{c}}(X)\,\tau^{\textup{c}}(X)\right]}
{\E\!\left[\beta_{u,v}(X)\,\pi^{\textup{c}}(X)\right]}.
\]
This family includes the local ATE-, ATT-, ATC-, and ATO-type estimands introduced earlier:
\[
\tau^{\textup{c}}_{1,1}=\tau^{\textup{c}}_{\textup{ATE}},\qquad
\tau^{\textup{c}}_{2,1}=\tau^{\textup{c}}_{\textup{ATT}},\qquad
\tau^{\textup{c}}_{1,2}=\tau^{\textup{c}}_{\textup{ATC}},\qquad
\tau^{\textup{c}}_{2,2}=\tau^{\textup{c}}_{\textup{ATO}}.
\]
For notational simplicity, define
\[
r^{\textup c}_{u,u';v}
=
\left[
\frac{\E\{\beta_{u',v}(X)\pi^{\textup c}(X)\}}
{\E\{\beta_{u,v}(X)\pi^{\textup c}(X)\}}
\right]^{1/(u'-u)},
\qquad
r^{\textup c}_{u;v,v'}
=
\left[
\frac{\E\{\beta_{u,v'}(X)\pi^{\textup c}(X)\}}
{\E\{\beta_{u,v}(X)\pi^{\textup c}(X)\}}
\right]^{1/(v'-v)}.
\]

\begin{theorem}\label{thm:local-beta-one-coordinate-diff}
Assume Assumption~\ref{assump:iv} and \(e(X)\in(0,1)\).
\begin{enumerate}[label=(\roman*),leftmargin=*, align=left]
\item Let \(u'>u\ge 1\) and \(v\ge 1\). Define
\[
w^{\textup c}_{u,u';v}(X)
=
\begin{cases}
\displaystyle
\beta_{u,v}(X)\pi^{\textup c}(X)
\frac{
e(X)^{u'-u}
-
\{r^{\textup c}_{u,u';v}\}^{u'-u}
}{
e(X)-r^{\textup c}_{u,u';v}
},
&
\displaystyle
e(X)\ne r^{\textup c}_{u,u';v},
\\[1.4em]
\displaystyle
\beta_{u,v}(X)\pi^{\textup c}(X)
(u'-u)
\{r^{\textup c}_{u,u';v}\}^{u'-u-1},
&
\displaystyle
e(X)=r^{\textup c}_{u,u';v}.
\end{cases}
\]
Then \(w^{\textup c}_{u,u';v}(X)>0\), and
\[
\tau^{\textup c}_{u',v}-\tau^{\textup c}_{u,v}
=
\frac{
\E\{w^{\textup c}_{u,u';v}(X)\}\,
\cov_{w^{\textup c}_{u,u';v}}\{\tau^{\textup c}(X),e(X)\}
}{
\E\{\beta_{u',v}(X)\pi^{\textup c}(X)\}
}.
\]

\item Let \(v'>v\ge 1\) and \(u\ge 1\). Define
\[
w^{\textup c}_{u;v,v'}(X)
=
\begin{cases}
\displaystyle
\beta_{u,v}(X)\pi^{\textup c}(X)
\frac{
\{r^{\textup c}_{u;v,v'}\}^{v'-v}
-
\{1-e(X)\}^{v'-v}
}{
r^{\textup c}_{u;v,v'}-\{1-e(X)\}
},
&
\displaystyle
e(X)\ne 1-r^{\textup c}_{u;v,v'},
\\[1.4em]
\displaystyle
\beta_{u,v}(X)\pi^{\textup c}(X)
(v'-v)
\{r^{\textup c}_{u;v,v'}\}^{v'-v-1},
&
\displaystyle
e(X)=1-r^{\textup c}_{u;v,v'}.
\end{cases}
\]
Then \(w^{\textup c}_{u;v,v'}(X)>0\), and
\[
\tau^{\textup c}_{u,v}-\tau^{\textup c}_{u,v'}
=
\frac{
\E\{w^{\textup c}_{u;v,v'}(X)\}\,
\cov_{w^{\textup c}_{u;v,v'}}\{\tau^{\textup c}(X),e(X)\}
}{
\E\{\beta_{u,v'}(X)\pi^{\textup c}(X)\}
}.
\]
\end{enumerate}
\end{theorem}

Theorem~\ref{thm:local-beta-one-coordinate-diff} is the local analog of
Theorem~\ref{thm:beta-one-coordinate-diff}. It shows that one-coordinate
changes in the beta-weighted local estimands are governed by weighted
covariances between the complier CATE and the IV propensity score. Because the
weights in Theorem~\ref{thm:local-beta-one-coordinate-diff} are proportional to
\(\pi^{\textup c}(X)\) times functions of \(e(X)\), monotonicity of the local CP
function \(g^{\textup c}(e)\) determines the signs of these covariances. This
yields the following monotonicity result for the beta-weighted LATEs.

\begin{coro}\label{coro:cate beta weight}
Assume Assumption~\ref{assump:iv} and \(e(X)\in(0,1)\). For real numbers \(u,v\ge 1\), under Assumption~\ref{assump:iv mono}(i), the quantity \(\tau_{u,v}^{\textup{c}}\) is non-decreasing in \(u\) and non-increasing in \(v\); under Assumption~\ref{assump:iv mono}(ii), the quantity \(\tau_{u,v}^{\textup{c}}\) is non-increasing in \(u\) and non-decreasing in \(v\).
\end{coro}

\section{Proofs}\label{sec:proofs}

Define 
\begin{align}\label{eq:el}
    e_l=\E [e^l(X)],\qquad
    \overline {e}_l=\E[\{1-e(X)\}^l]
\end{align} 
as the $l$-th moments of $e(X)$ and $1-e(X)$, $l=1,2,3$.
\begin{lemma}\label{lemma:z1ex}
    For any function \(A(X)\), we have
\[
\E\{A(X)\mid Z=1\}
=
\frac{\E\{e(X)A(X)\}}{p},
\qquad
\E\{A(X)\mid Z=0\}
=
\frac{\E[\{1-e(X)\}A(X)]}{1-p}.
\]
\end{lemma}
\begin{proof}
By iterated expectations, \(\E\{ZA(X)\}=\E[e(X)A(X)]\) and \(\E(Z)=\E\{e(X)\}\), which gives the first identity; the second follows analogously by replacing \(Z\) with \(1-Z\).
\end{proof}
\subsection{Proof of Theorem~\ref{thm:diff-rep}}

\begin{proof}
We prove parts (i) and (ii) separately.

\proofpart{Part (i): Comparisons involving ATE}
First, we compare \(\tauatt\) and \(\tauate\):
\[
\begin{aligned}
\tauatt-\tauate
&=
\frac{\E\{e(X)\tau(X)\}}{\E\{e(X)\}}
-
\E\{\tau(X)\} \\
&=
\frac{
\E\{e(X)\tau(X)\}
-
\E\{e(X)\}\E\{\tau(X)\}
}{\E\{e(X)\}} \\
&=
\frac{\cov\{\tau(X),e(X)\}}{\E\{e(X)\}}.
\end{aligned}
\]

Next, we compare \(\tauate\) and \(\tauatc\):
\[
\begin{aligned}
\tauate-\tauatc
&=
\E\{\tau(X)\}
-
\frac{\E[\{1-e(X)\}\tau(X)]}{\E\{1-e(X)\}} \\
&=
\frac{
\E\{1-e(X)\}\E\{\tau(X)\}
-
\E[\{1-e(X)\}\tau(X)]
}{\E\{1-e(X)\}} \\
&=
\frac{
\E\{e(X)\tau(X)\}
-
\E\{e(X)\}\E\{\tau(X)\}
}{\E\{1-e(X)\}} \\
&=
\frac{\cov\{\tau(X),e(X)\}}{\E\{1-e(X)\}}.
\end{aligned}
\]
This proves part (i).

\proofpart{Part (ii): Comparisons involving ATO}
We first compare \(\tauatt\) and \(\tauow\). By
Lemma~\ref{lemma:z1ex},
\[
\tauatt=\E\{\tau(X)\mid Z=1\},
\qquad
\tauow
=
\frac{\E[\{1-e(X)\}\tau(X)\mid Z=1]}
{\E\{1-e(X)\mid Z=1\}}.
\]
Therefore,
\[
\begin{aligned}
\tauatt-\tauow
&=
\E\{\tau(X)\mid Z=1\}
-
\frac{\E[\{1-e(X)\}\tau(X)\mid Z=1]}
{\E\{1-e(X)\mid Z=1\}}\\
&=
\frac{
\E\{\tau(X)\mid Z=1\}\E\{1-e(X)\mid Z=1\}
-
\E[\{1-e(X)\}\tau(X)\mid Z=1]
}{
\E\{1-e(X)\mid Z=1\}
}\\
&=
\frac{
\E\{e(X)\tau(X)\mid Z=1\}
-
\E\{e(X)\mid Z=1\}\E\{\tau(X)\mid Z=1\}
}{
\E\{1-e(X)\mid Z=1\}
}\\
&=
\frac{
\cov\{\tau(X),e(X)\mid Z=1\}
}{
\E\{1-e(X)\mid Z=1\}
}.
\end{aligned}
\]

We next compare \(\tauow\) and \(\tauatc\). Again by
Lemma~\ref{lemma:z1ex},
\[
\tauatc=\E\{\tau(X)\mid Z=0\},
\qquad
\tauow
=
\frac{\E\{e(X)\tau(X)\mid Z=0\}}
{\E\{e(X)\mid Z=0\}}.
\]
It follows that
\[
\begin{aligned}
\tauow-\tauatc
&=
\frac{\E\{e(X)\tau(X)\mid Z=0\}}
{\E\{e(X)\mid Z=0\}}
-
\E\{\tau(X)\mid Z=0\}\\
&=
\frac{
\E\{e(X)\tau(X)\mid Z=0\}
-
\E\{e(X)\mid Z=0\}\E\{\tau(X)\mid Z=0\}
}{
\E\{e(X)\mid Z=0\}
}\\
&=
\frac{
\cov\{\tau(X),e(X)\mid Z=0\}
}{
\E\{e(X)\mid Z=0\}
}.
\end{aligned}
\]
This proves part (ii), and hence all four identities.
\end{proof}
\subsection{Proof of Proposition~\ref{prop:covariance-interpretation}}
\begin{proof}
We prove the two parts in turn.

For part (i), let \(\Delta=Y(1)-Y(0)\). By the law of total covariance and
Assumption~\ref{assume:ignor},
\[
\begin{aligned}
\cov(\Delta,Z)
&=
\E\{\cov(\Delta,Z\mid X)\}
+
\cov\{\E(\Delta\mid X),\E(Z\mid X)\}\\
&=
\cov\{\tau(X),e(X)\}.
\end{aligned}
\]
Moreover, for any random variable \(W\) and binary \(Z\),
\[
\cov(W,Z)
=
\var(Z)\big[\E(W\mid Z=1)-\E(W\mid Z=0)\big].
\]
Applying this identity to \(W=Y(1)\) and \(W=Y(0)\) gives
\[
\cov\{Y(1)-Y(0),Z\}
=
\var(Z)\{B(1)-B(0)\},
\]
which proves \eqref{eq:selection-on-gains}.

For part (ii), the law of total covariance gives
\begin{align}\label{eq:proof-total-cov-decomp}
\cov\{\tau(X),e(X)\}
={}&
\E\{e(X)\}\cov\{\tau(X),e(X)\mid Z=1\}\notag\\
&+
\E\{1-e(X)\}\cov\{\tau(X),e(X)\mid Z=0\}\notag\\
&+
\cov\!\left[\E\{\tau(X)\mid Z\},\E\{e(X)\mid Z\}\right].
\end{align}
Because \(Z\) is binary,
\[
\E\{\tau(X)\mid Z=1\}-\E\{\tau(X)\mid Z=0\}
=
\frac{\cov\{\tau(X),Z\}}{\var(Z)}
=
\frac{\cov\{\tau(X),e(X)\}}{\var(Z)},
\]
where the last equality follows because \(\tau(X)\) is a function of \(X\).
Similarly,
\[
\E\{e(X)\mid Z=1\}-\E\{e(X)\mid Z=0\}
=
\frac{\cov\{e(X),Z\}}{\var(Z)}
=
\frac{\var\{e(X)\}}{\var(Z)}.
\]
Therefore,
\[
\cov\!\left[\E\{\tau(X)\mid Z\},\E\{e(X)\mid Z\}\right]
=
\frac{\var\{e(X)\}}{\var(Z)}
\cov\{\tau(X),e(X)\}.
\]
Substituting this identity into \eqref{eq:proof-total-cov-decomp} and using
\[
\var(Z)
=
\E\{\var(Z\mid X)\}+\var\{e(X)\}
\]
yields \eqref{eq:marginal-conditional-covariance}.
\end{proof}

\subsection{Proof of Corollary \ref{coro:exact-bracketing}}
\begin{proof}
All denominators in Theorem~\ref{thm:diff-rep} are positive. Therefore,
each displayed estimand difference has the sign of its corresponding
covariance, which gives both equivalences.
\end{proof}

\subsection{Proof of Corollary \ref{coro:ow in atc att}}

\begin{proof}
Under Assumption~\ref{assump:mono}(i), \(g(e)\) is non-decreasing in \(e\).
For \(z=0,1\), because \(Z\) is a binary variable with propensity score
\(e(X)\), conditioning on \(e(X)\) makes \(Z\) independent of any function of
\(X\). Therefore,
\[
\E\{\tau(X)\mid e(X),Z=z\}
=
\E\{\tau(X)\mid e(X)\}
=
g(e(X)).
\]
It follows that
\[
\cov\{\tau(X),e(X)\mid Z=z\}
=
\cov\{g(e(X)),e(X)\mid Z=z\},
\qquad z=0,1.
\]
Since \(g(e)\) is non-decreasing in \(e\), the covariance between
\(g(e(X))\) and \(e(X)\) is nonnegative under the conditional distribution
given \(Z=z\). Hence
\[
\cov\{\tau(X),e(X)\mid Z=1\}\ge 0,
\qquad
\cov\{\tau(X),e(X)\mid Z=0\}\ge 0.
\]
By Corollary~\ref{coro:exact-bracketing},
\[
\tauow\in[\tauatc,\tauatt].
\]

Under Assumption~\ref{assump:mono}(ii), the same argument applies with
\(g(e)\) non-increasing in \(e\). Thus
\[
\cov\{\tau(X),e(X)\mid Z=1\}\le 0,
\qquad
\cov\{\tau(X),e(X)\mid Z=0\}\le 0,
\]
and Corollary~\ref{coro:exact-bracketing} gives
\[
\tauow\in[\tauatt,\tauatc].
\]
\end{proof}

\subsection{Proof of Theorem~\ref{thm:diff-rep-local}}
\begin{proof}
\proofpart{Preliminary identities}
We first record several useful identities. Recall that
\(\pi^{\textup c}(X)=\p(U=\textup c\mid X)\). By the definition of
\(\pi^{\textup c}(X)\) and conditional IV exogeneity,
\(Z\perp\!\!\!\perp U\mid X\), for any integrable function \(A(X)\),
\[
\E\{A(X)\mid U=\textup c\}
=
\frac{\E\{\pi^{\textup c}(X)A(X)\}}
{\E\{\pi^{\textup c}(X)\}},
\]
\[
\E\{A(X)\mid U=\textup c,Z=1\}
=
\frac{\E\{e(X)\pi^{\textup c}(X)A(X)\}}
{\E\{e(X)\pi^{\textup c}(X)\}},
\]
and
\[
\E\{A(X)\mid U=\textup c,Z=0\}
=
\frac{\E[\{1-e(X)\}\pi^{\textup c}(X)A(X)]}
{\E[\{1-e(X)\}\pi^{\textup c}(X)]}.
\]
It follows that
\[
\tau^{\textup c}_{\textup{ATE}}
=
\E\{\tau^{\textup c}(X)\mid U=\textup c\},
\]
\[
\tau^{\textup c}_{\textup{ATT}}
=
\frac{\E\{e(X)\tau^{\textup c}(X)\mid U=\textup c\}}
{\E\{e(X)\mid U=\textup c\}}
=
\E\{\tau^{\textup c}(X)\mid U=\textup c,Z=1\},
\]
and
\[
\tau^{\textup c}_{\textup{ATC}}
=
\frac{\E[\{1-e(X)\}\tau^{\textup c}(X)\mid U=\textup c]}
{\E\{1-e(X)\mid U=\textup c\}}
=
\E\{\tau^{\textup c}(X)\mid U=\textup c,Z=0\}.
\]
\proofpart{Part (i): Comparisons involving the local ATE}
We first compare \(\tau^{\textup c}_{\textup{ATT}}\) and
\(\tau^{\textup c}_{\textup{ATE}}\). We have
\[
\begin{aligned}
\tau^{\textup c}_{\textup{ATT}}
-
\tau^{\textup c}_{\textup{ATE}}
&=
\frac{\E\{e(X)\tau^{\textup c}(X)\mid U=\textup c\}}
{\E\{e(X)\mid U=\textup c\}}
-
\E\{\tau^{\textup c}(X)\mid U=\textup c\} \\
&=
\frac{
\E\{e(X)\tau^{\textup c}(X)\mid U=\textup c\}
-
\E\{e(X)\mid U=\textup c\}
\E\{\tau^{\textup c}(X)\mid U=\textup c\}
}{
\E\{e(X)\mid U=\textup c\}
} \\
&=
\frac{
\cov\{\tau^{\textup c}(X),e(X)\mid U=\textup c\}
}{
\E\{e(X)\mid U=\textup c\}
}.
\end{aligned}
\]
We next compare \(\tau^{\textup c}_{\textup{ATE}}\) and
\(\tau^{\textup c}_{\textup{ATC}}\):
\[
\begin{aligned}
\tau^{\textup c}_{\textup{ATE}}
-
\tau^{\textup c}_{\textup{ATC}}
&=
\E\{\tau^{\textup c}(X)\mid U=\textup c\}
-
\frac{\E[\{1-e(X)\}\tau^{\textup c}(X)\mid U=\textup c]}
{\E\{1-e(X)\mid U=\textup c\}}\\
&=
\frac{
\E\{e(X)\tau^{\textup c}(X)\mid U=\textup c\}
-
\E\{e(X)\mid U=\textup c\}
\E\{\tau^{\textup c}(X)\mid U=\textup c\}
}{
\E\{1-e(X)\mid U=\textup c\}
}\\
&=
\frac{
\cov\{\tau^{\textup c}(X),e(X)\mid U=\textup c\}
}{
\E\{1-e(X)\mid U=\textup c\}
}.
\end{aligned}
\]
This proves part (i).
\proofpart{Part (ii): Comparisons involving the local ATO}
We first compare \(\tau^{\textup c}_{\textup{ATT}}\) and
\(\tau^{\textup c}_{\textup{ATO}}\). Under the conditional distribution
given \((U=\textup c,Z=1)\),
\[
\tau^{\textup c}_{\textup{ATT}}
=
\E\{\tau^{\textup c}(X)\mid U=\textup c,Z=1\},
\]
and
\[
\tau^{\textup c}_{\textup{ATO}}
=
\frac{
\E[\{1-e(X)\}\tau^{\textup c}(X)\mid U=\textup c,Z=1]
}{
\E\{1-e(X)\mid U=\textup c,Z=1\}
}.
\]
Therefore,
\[
\begin{aligned}
\tau^{\textup c}_{\textup{ATT}}
-
\tau^{\textup c}_{\textup{ATO}}
&=
\E\{\tau^{\textup c}(X)\mid U=\textup c,Z=1\}
-
\frac{
\E[\{1-e(X)\}\tau^{\textup c}(X)\mid U=\textup c,Z=1]
}{
\E\{1-e(X)\mid U=\textup c,Z=1\}
}\\
&=
\frac{
\E\{e(X)\tau^{\textup c}(X)\mid U=\textup c,Z=1\}
-
\E\{e(X)\mid U=\textup c,Z=1\}
\E\{\tau^{\textup c}(X)\mid U=\textup c,Z=1\}
}{
\E\{1-e(X)\mid U=\textup c,Z=1\}
}\\
&=
\frac{
\cov\{\tau^{\textup c}(X),e(X)\mid U=\textup c,Z=1\}
}{
\E\{1-e(X)\mid U=\textup c,Z=1\}
}.
\end{aligned}
\]
We next compare \(\tau^{\textup c}_{\textup{ATO}}\) and
\(\tau^{\textup c}_{\textup{ATC}}\). Under the conditional distribution
given \((U=\textup c,Z=0)\),
\[
\tau^{\textup c}_{\textup{ATC}}
=
\E\{\tau^{\textup c}(X)\mid U=\textup c,Z=0\},
\]
and
\[
\tau^{\textup c}_{\textup{ATO}}
=
\frac{
\E\{e(X)\tau^{\textup c}(X)\mid U=\textup c,Z=0\}
}{
\E\{e(X)\mid U=\textup c,Z=0\}
}.
\]
Thus,
\[
\begin{aligned}
\tau^{\textup c}_{\textup{ATO}}
-
\tau^{\textup c}_{\textup{ATC}}
&=
\frac{
\E\{e(X)\tau^{\textup c}(X)\mid U=\textup c,Z=0\}
}{
\E\{e(X)\mid U=\textup c,Z=0\}
}
-
\E\{\tau^{\textup c}(X)\mid U=\textup c,Z=0\}\\
&=
\frac{
\E\{e(X)\tau^{\textup c}(X)\mid U=\textup c,Z=0\}
-
\E\{e(X)\mid U=\textup c,Z=0\}
\E\{\tau^{\textup c}(X)\mid U=\textup c,Z=0\}
}{
\E\{e(X)\mid U=\textup c,Z=0\}
}\\
&=
\frac{
\cov\{\tau^{\textup c}(X),e(X)\mid U=\textup c,Z=0\}
}{
\E\{e(X)\mid U=\textup c,Z=0\}
}.
\end{aligned}
\]
This proves part (ii), and hence all four identities.
\end{proof}

\subsection{Proof of Proposition~\ref{prop:covariance-interpretation-local}}
\begin{proof}
We prove the two parts in turn. Let
\(\Delta=Y(1)-Y(0)\).

For part~(i), Assumption~\ref{assump:iv}(i) and the fact that
\(U=(D(1),D(0))\) is determined by the potential treatments imply
\[
Z\perp\!\!\!\perp (\Delta,U)\mid X.
\]
Consequently,
\(Z\perp\!\!\!\perp\Delta\mid X,U=\textup c\) and
\(\E(Z\mid X,U=\textup c)=e(X)\). The law of total covariance within the
complier population therefore gives
\[
\begin{aligned}
\cov(\Delta,Z\mid U=\textup c)
={}&
\E\{\cov(\Delta,Z\mid X,U=\textup c)\mid U=\textup c\}\\
&+
\cov\{\E(\Delta\mid X,U=\textup c),
       \E(Z\mid X,U=\textup c)\mid U=\textup c\}\\
={}&
\cov\{\tau^{\textup c}(X),e(X)\mid U=\textup c\}.
\end{aligned}
\]
Moreover, \(D=Z\) whenever \(U=\textup c\), so
\[
\cov(\Delta,D\mid U=\textup c)
=
\cov(\Delta,Z\mid U=\textup c).
\]
For any random variable \(W\) and binary \(Z\), conditional on
\(U=\textup c\),
\[
\cov(W,Z\mid U=\textup c)
=
\var(Z\mid U=\textup c)
\big[
\E(W\mid Z=1,U=\textup c)
-
\E(W\mid Z=0,U=\textup c)
\big].
\]
Applying this identity to \(W=Y(1)\) and \(W=Y(0)\), and using
\(D=Z\) within the complier stratum, yields
\[
\cov(\Delta,Z\mid U=\textup c)
=
\var(Z\mid U=\textup c)
\{B^{\textup c}(1)-B^{\textup c}(0)\}.
\]
This proves \eqref{eq:local-selection-on-gains}.

For part~(ii), write
\[
C_{\textup c}
=
\cov\{\tau^{\textup c}(X),e(X)\mid U=\textup c\},
\]
and, for \(z=0,1\),
\[
C_{\textup c,z}
=
\cov\{\tau^{\textup c}(X),e(X)\mid U=\textup c,Z=z\}.
\]
The law of total covariance, applied within the complier population and then
conditioning on \(Z\), gives
\begin{align}\label{eq:proof-local-total-cov-decomp}
C_{\textup c}
={}&
 p_{\textup c}C_{\textup c,1}
+
(1-p_{\textup c})C_{\textup c,0}
\notag\\
&+
\cov\!\left[
\E\{\tau^{\textup c}(X)\mid Z,U=\textup c\},
\E\{e(X)\mid Z,U=\textup c\}
\mid U=\textup c
\right].
\end{align}
Because \(Z\) is binary within the complier population,
\[
\begin{aligned}
&\E\{\tau^{\textup c}(X)\mid Z=1,U=\textup c\}
-
\E\{\tau^{\textup c}(X)\mid Z=0,U=\textup c\}\\
&\qquad=
\frac{\cov\{\tau^{\textup c}(X),Z\mid U=\textup c\}}
{\var(Z\mid U=\textup c)}
=
\frac{C_{\textup c}}{\var(Z\mid U=\textup c)},
\end{aligned}
\]
where the last equality follows from
\(\E(Z\mid X,U=\textup c)=e(X)\). Similarly,
\[
\E\{e(X)\mid Z=1,U=\textup c\}
-
\E\{e(X)\mid Z=0,U=\textup c\}
=
\frac{\var\{e(X)\mid U=\textup c\}}
{\var(Z\mid U=\textup c)}.
\]
Because the covariance of two functions of a binary variable equals its
variance times the product of the corresponding two mean differences, the last
covariance in \eqref{eq:proof-local-total-cov-decomp} equals
\[
\frac{\var\{e(X)\mid U=\textup c\}}
{\var(Z\mid U=\textup c)}
C_{\textup c}.
\]
Finally, conditional IV exogeneity implies
\(\var(Z\mid X,U=\textup c)=e(X)\{1-e(X)\}\). Hence the conditional law of
total variance gives
\[
\var(Z\mid U=\textup c)
=
\E[e(X)\{1-e(X)\}\mid U=\textup c]
+
\var\{e(X)\mid U=\textup c\}.
\]
Substituting the preceding two identities into
\eqref{eq:proof-local-total-cov-decomp} and solving for
\(C_{\textup c}\) yields
\eqref{eq:local-marginal-conditional-covariance}.
\end{proof}

\subsection{Proof of Corollary \ref{coro:exact-bracketing-local}}
\begin{proof}
All denominators in Theorem~\ref{thm:diff-rep-local} are positive. Therefore,
each displayed local estimand difference has the sign of its corresponding
covariance, which gives both equivalences.
\end{proof}

\subsection{Proof of Corollary~\ref{coro:ow in atc att local}}
\label{subsec:proof_coro_owinatcatt}

\begin{proof}
Equation~\eqref{eq:key_tauce_to_gc} shows that, for \(z\in\{0,1\}\),
\[
\cov\{\tau^{\textup c}(X),e(X)
\mid U=\textup c,Z=z\}
=
\cov\{g^{\textup c}(e(X)),e(X)
\mid U=\textup c,Z=z\}.
\]

Under Assumption~\ref{assump:iv mono}(i),
\(g^{\textup c}(e)\) is non-decreasing in \(e\). Hence, under each
conditional distribution given \((U=\textup c,Z=z)\),
\[
\cov\{g^{\textup c}(e(X)),e(X)
\mid U=\textup c,Z=z\}\ge 0,
\qquad z=0,1,
\]
because the covariance between a scalar random variable and a
non-decreasing function of itself is nonnegative. Therefore,
\[
\cov\{\tau^{\textup c}(X),e(X)
\mid U=\textup c,Z=z\}\ge 0,
\qquad z=0,1.
\]
Corollary~\ref{coro:exact-bracketing-local} then gives
\[
\tau^{\textup c}_{\textup{ATO}}
\in
\bigl[
\tau^{\textup c}_{\textup{ATC}},
\tau^{\textup c}_{\textup{ATT}}
\bigr].
\]

Under Assumption~\ref{assump:iv mono}(ii), the same argument gives
\[
\cov\{\tau^{\textup c}(X),e(X)
\mid U=\textup c,Z=z\}\le 0,
\qquad z=0,1,
\]
and hence
\[
\tau^{\textup c}_{\textup{ATO}}
\in
\bigl[
\tau^{\textup c}_{\textup{ATT}},
\tau^{\textup c}_{\textup{ATC}}
\bigr].
\]
\end{proof}
\subsection{Proof of Theorem~\ref{thm:population-cp-geometry}}

\begin{proof}

\proofpart{Preliminary notation and weighted centroids}

Write
\[
E=e(X),
\qquad
T=\tau(X).
\]
For any integrable function \(a(X)\), define
\[
\E_{\mathrm{all}}\{a(X)\}=\E\{a(X)\},
\]
\[
\E_{\mathrm{tr}}\{a(X)\}
=
\frac{\E\{E a(X)\}}{\E(E)}
=
\E\{a(X)\mid Z=1\},
\]
and
\[
\E_{\mathrm{co}}\{a(X)\}
=
\frac{\E\{(1-E)a(X)\}}{\E(1-E)}
=
\E\{a(X)\mid Z=0\}.
\]
Let \(\cov_j\) and \(\var_j\) denote covariance and variance under
\(\E_j\), for
\(j\in\{\mathrm{all},\mathrm{tr},\mathrm{co}\}\).

The normal equations for a weighted linear projection imply that each fitted
line passes through its corresponding weighted centroid. Therefore,
\[
L_{\mathrm{all}}(\bar e_{\mathrm{ATE}})
=
\tau_{\mathrm{ATE}},
\qquad
L_{\mathrm{tr}}(\bar e_{\mathrm{ATT}})
=
\tau_{\mathrm{ATT}},
\qquad
L_{\mathrm{co}}(\bar e_{\mathrm{ATC}})
=
\tau_{\mathrm{ATC}},
\]
and the slopes are
\[
\beta_{\mathrm{all}}
=
\frac{\cov_{\mathrm{all}}(T,E)}
{\var_{\mathrm{all}}(E)},
\qquad
\beta_{\mathrm{tr}}
=
\frac{\cov_{\mathrm{tr}}(T,E)}
{\var_{\mathrm{tr}}(E)},
\qquad
\beta_{\mathrm{co}}
=
\frac{\cov_{\mathrm{co}}(T,E)}
{\var_{\mathrm{co}}(E)}.
\]

\proofpart{Part (i): Slope identities for the unweighted line}

Direct calculation gives
\[
\bar e_{\mathrm{ATT}}-\bar e_{\mathrm{ATE}}
=
\frac{\var_{\mathrm{all}}(E)}
{\E_{\mathrm{all}}(E)},
\qquad
\bar e_{\mathrm{ATE}}-\bar e_{\mathrm{ATC}}
=
\frac{\var_{\mathrm{all}}(E)}
{\E_{\mathrm{all}}(1-E)}.
\]
The corresponding treatment-effect differences satisfy
\[
\tau_{\mathrm{ATT}}-\tau_{\mathrm{ATE}}
=
\frac{\cov_{\mathrm{all}}(T,E)}
{\E_{\mathrm{all}}(E)},
\qquad
\tau_{\mathrm{ATE}}-\tau_{\mathrm{ATC}}
=
\frac{\cov_{\mathrm{all}}(T,E)}
{\E_{\mathrm{all}}(1-E)}.
\]
Dividing each treatment-effect difference by the corresponding
propensity-score difference yields
\[
\frac{\tau_{\mathrm{ATT}}-\tau_{\mathrm{ATE}}}
{\bar e_{\mathrm{ATT}}-\bar e_{\mathrm{ATE}}}
=
\frac{\cov_{\mathrm{all}}(T,E)}
{\var_{\mathrm{all}}(E)}
=
\beta_{\mathrm{all}},
\]
and
\[
\frac{\tau_{\mathrm{ATE}}-\tau_{\mathrm{ATC}}}
{\bar e_{\mathrm{ATE}}-\bar e_{\mathrm{ATC}}}
=
\frac{\cov_{\mathrm{all}}(T,E)}
{\var_{\mathrm{all}}(E)}
=
\beta_{\mathrm{all}}.
\]
This proves part (i).

\proofpart{Part (ii): Slope identities for the treated- and control-weighted lines}

Because the ATO distribution can be obtained by tilting the treated
distribution by \(1-E\), or equivalently by tilting the control distribution
by \(E\), we have
\[
\bar e_{\mathrm{ATT}}-\bar e_{\mathrm{ATO}}
=
\frac{\var_{\mathrm{tr}}(E)}
{\E_{\mathrm{tr}}(1-E)},
\qquad
\bar e_{\mathrm{ATO}}-\bar e_{\mathrm{ATC}}
=
\frac{\var_{\mathrm{co}}(E)}
{\E_{\mathrm{co}}(E)}.
\]
The corresponding treatment-effect differences satisfy
\[
\tau_{\mathrm{ATT}}-\tau_{\mathrm{ATO}}
=
\frac{\cov_{\mathrm{tr}}(T,E)}
{\E_{\mathrm{tr}}(1-E)},
\qquad
\tau_{\mathrm{ATO}}-\tau_{\mathrm{ATC}}
=
\frac{\cov_{\mathrm{co}}(T,E)}
{\E_{\mathrm{co}}(E)}.
\]
Therefore,
\[
\frac{\tau_{\mathrm{ATT}}-\tau_{\mathrm{ATO}}}
{\bar e_{\mathrm{ATT}}-\bar e_{\mathrm{ATO}}}
=
\frac{\cov_{\mathrm{tr}}(T,E)}
{\var_{\mathrm{tr}}(E)}
=
\beta_{\mathrm{tr}},
\]
and
\[
\frac{\tau_{\mathrm{ATO}}-\tau_{\mathrm{ATC}}}
{\bar e_{\mathrm{ATO}}-\bar e_{\mathrm{ATC}}}
=
\frac{\cov_{\mathrm{co}}(T,E)}
{\var_{\mathrm{co}}(E)}
=
\beta_{\mathrm{co}}.
\]
This proves part (ii).

\proofpart{Part (iii): Intersections of the three fitted lines}

Since the unweighted line passes through
\((\bar e_{\mathrm{ATE}},\tau_{\mathrm{ATE}})\), part (i) implies
\[
\begin{aligned}
L_{\mathrm{all}}(\bar e_{\mathrm{ATT}})
&=
\tau_{\mathrm{ATE}}
+
\beta_{\mathrm{all}}
(\bar e_{\mathrm{ATT}}-\bar e_{\mathrm{ATE}})\\
&=
\tau_{\mathrm{ATT}},
\end{aligned}
\]
and
\[
\begin{aligned}
L_{\mathrm{all}}(\bar e_{\mathrm{ATC}})
&=
\tau_{\mathrm{ATE}}
+
\beta_{\mathrm{all}}
(\bar e_{\mathrm{ATC}}-\bar e_{\mathrm{ATE}})\\
&=
\tau_{\mathrm{ATC}}.
\end{aligned}
\]
Because
\[
L_{\mathrm{tr}}(\bar e_{\mathrm{ATT}})
=
\tau_{\mathrm{ATT}},
\qquad
L_{\mathrm{co}}(\bar e_{\mathrm{ATC}})
=
\tau_{\mathrm{ATC}},
\]
the ATT point lies on both \(L_{\mathrm{all}}\) and \(L_{\mathrm{tr}}\),
whereas the ATC point lies on both \(L_{\mathrm{co}}\) and
\(L_{\mathrm{all}}\).

Next, part (ii) gives
\[
\begin{aligned}
L_{\mathrm{tr}}(\bar e_{\mathrm{ATO}})
&=
\tau_{\mathrm{ATT}}
+
\beta_{\mathrm{tr}}
(\bar e_{\mathrm{ATO}}-\bar e_{\mathrm{ATT}})\\
&=
\tau_{\mathrm{ATO}},
\end{aligned}
\]
and
\[
\begin{aligned}
L_{\mathrm{co}}(\bar e_{\mathrm{ATO}})
&=
\tau_{\mathrm{ATC}}
+
\beta_{\mathrm{co}}
(\bar e_{\mathrm{ATO}}-\bar e_{\mathrm{ATC}})\\
&=
\tau_{\mathrm{ATO}}.
\end{aligned}
\]
Thus,
\[
(\bar e_{\mathrm{ATT}},\tau_{\mathrm{ATT}})
\in L_{\mathrm{all}}\cap L_{\mathrm{tr}},
\]
\[
(\bar e_{\mathrm{ATO}},\tau_{\mathrm{ATO}})
\in L_{\mathrm{tr}}\cap L_{\mathrm{co}},
\]
and
\[
(\bar e_{\mathrm{ATC}},\tau_{\mathrm{ATC}})
\in L_{\mathrm{co}}\cap L_{\mathrm{all}}.
\]
Whenever each corresponding pair of lines is distinct, two affine lines have
at most one common point. Hence these common points are their unique
intersections. This proves part (iii).
\end{proof}

\subsection{Proof of Proposition~\ref{prop:cp-slope-convex}}
\begin{proof}
Write
\[
E=e(X),\qquad T=\tau(X),
\]
and let
\[
p=\p(Z=1),\qquad q=\p(Z=0).
\]
For brevity, define
\[
C=\cov(T,E),\qquad V=\var(E),
\]
\[
C_1=\cov(T,E\mid Z=1),\qquad
C_0=\cov(T,E\mid Z=0),
\]
and
\[
V_1=\var(E\mid Z=1),\qquad
V_0=\var(E\mid Z=0).
\]
By the definitions of the three population regression slopes,
\[
\beta_{\mathrm{all}}=\frac{C}{V},\qquad
\beta_{\mathrm{tr}}=\frac{C_1}{V_1},\qquad
\beta_{\mathrm{co}}=\frac{C_0}{V_0}.
\]

Equation~\eqref{eq:marginal-conditional-covariance} gives
\begin{align}\label{eq:proof-slope-cov-relation}
C
=
\frac{\var(Z)}{\E\{\var(Z\mid X)\}}
\{pC_1+qC_0\}.
\end{align}
The same total-covariance calculation, with the first argument also equal to
\(E=e(X)\), gives
\begin{align}\label{eq:proof-slope-var-relation}
V
=
\frac{\var(Z)}{\E\{\var(Z\mid X)\}}
\{pV_1+qV_0\}.
\end{align}
Dividing \eqref{eq:proof-slope-cov-relation} by
\eqref{eq:proof-slope-var-relation}, we obtain
\[
\begin{aligned}
\beta_{\mathrm{all}}
&=
\frac{pC_1+qC_0}{pV_1+qV_0}\\
&=
\frac{pV_1}{pV_1+qV_0}\,\beta_{\mathrm{tr}}
+
\frac{qV_0}{pV_1+qV_0}\,\beta_{\mathrm{co}}.
\end{aligned}
\]
Because the treated- and control-weighted projections are well defined,
\(V_1>0\) and \(V_0>0\). Also \(p,q>0\) under overlap. Hence the two
coefficients above are positive and sum to one. This proves
\eqref{eq:cp-slope-convex} and the convex-combination claim.
\end{proof}

\subsection{Proof of Theorem~\ref{thm:beta-one-coordinate-diff}}

\begin{proof}

\proofpart{Part (i): Increasing the first beta-weight parameter}

\textit{Step 1: Positivity and centering of \(w_{u,u';v}\).}
Let \(a=u'-u>0\), and write
\[
r
=
\left[
\frac{\E\{\beta_{u',v}(X)\}}
{\E\{\beta_{u,v}(X)\}}
\right]^{1/a}.
\]
Because \(e(X)\in(0,1)\) and
\(\beta_{u',v}(X)=\beta_{u,v}(X)e(X)^a\), we have
\(r\in(0,1)\). Moreover, the function \(t\mapsto t^a\) is strictly
increasing on \((0,1)\), so
\[
\frac{e(X)^a-r^a}{e(X)-r}>0
\]
whenever \(e(X)\ne r\), and its value at \(e(X)=r\) equals
\(ar^{a-1}>0\). Therefore,
\[
w_{u,u';v}(X)>0.
\]

By the definition of \(w_{u,u';v}(X)\),
\[
w_{u,u';v}(X)\{e(X)-r\}
=
\beta_{u,v}(X)
\left[
e(X)^a
-
\frac{\E\{\beta_{u',v}(X)\}}
{\E\{\beta_{u,v}(X)\}}
\right].
\]
Taking expectations gives
\[
\begin{aligned}
\E\left[w_{u,u';v}(X)\{e(X)-r\}\right]
=
\E\{\beta_{u',v}(X)\}-
\frac{\E\{\beta_{u',v}(X)\}}
{\E\{\beta_{u,v}(X)\}}
\E\{\beta_{u,v}(X)\}=0.
\end{aligned}
\]
Hence,
\[
\E_{w_{u,u';v}}\{e(X)\}=r.
\]

\textit{Step 2: Covariance representation.}
We have
\[
\begin{aligned}
\tau_{u',v}-\tau_{u,v}
&=
\E\left[
\left\{
\frac{\beta_{u',v}(X)}
{\E\{\beta_{u',v}(X)\}}
-
\frac{\beta_{u,v}(X)}
{\E\{\beta_{u,v}(X)\}}
\right\}
\tau(X)
\right]\\
&=
\frac{1}{\E\{\beta_{u',v}(X)\}}
\E\left[
\beta_{u,v}(X)
\left\{
e(X)^a
-
\frac{\E\{\beta_{u',v}(X)\}}
{\E\{\beta_{u,v}(X)\}}
\right\}
\tau(X)
\right]\\
&=
\frac{1}{\E\{\beta_{u',v}(X)\}}
\E\left[
w_{u,u';v}(X)\{e(X)-r\}\tau(X)
\right]\\
&=
\frac{\E\{w_{u,u';v}(X)\}}
{\E\{\beta_{u',v}(X)\}}
\E_{w_{u,u';v}}
\left[
\{e(X)-\E_{w_{u,u';v}}(e(X))\}\tau(X)
\right]\\
&=
\frac{
\E\{w_{u,u';v}(X)\}
\cov_{w_{u,u';v}}\{\tau(X),e(X)\}
}{
\E\{\beta_{u',v}(X)\}
}.
\end{aligned}
\]
This proves part (i).

\proofpart{Part (ii): Increasing the second beta-weight parameter}

\textit{Step 1: Positivity and centering of \(w_{u;v,v'}\).}
Let \(b=v'-v>0\), and write
\[
s
=
\left[
\frac{\E\{\beta_{u,v'}(X)\}}
{\E\{\beta_{u,v}(X)\}}
\right]^{1/b}.
\]
Because
\[
\beta_{u,v'}(X)
=
\beta_{u,v}(X)\{1-e(X)\}^b,
\]
we have \(s\in(0,1)\). The corresponding crossing point in \(e(X)\) is
\(1-s\). Since \(t\mapsto(1-t)^b\) is strictly decreasing,
\[
\frac{s^b-\{1-e(X)\}^b}{e(X)-(1-s)}>0
\]
whenever \(e(X)\ne1-s\), and its value at \(e(X)=1-s\) equals
\(bs^{b-1}>0\). Therefore,
\[
w_{u;v,v'}(X)>0.
\]

By the definition of \(w_{u;v,v'}(X)\),
\[
w_{u;v,v'}(X)\{e(X)-(1-s)\}
=
\beta_{u,v}(X)
\left[
\frac{\E\{\beta_{u,v'}(X)\}}
{\E\{\beta_{u,v}(X)\}}
-
\{1-e(X)\}^b
\right].
\]
Taking expectations gives
\[
\begin{aligned}
\E\left[
w_{u;v,v'}(X)\{e(X)-(1-s)\}
\right]
=
\frac{\E\{\beta_{u,v'}(X)\}}
{\E\{\beta_{u,v}(X)\}}
\E\{\beta_{u,v}(X)\}-
\E\{\beta_{u,v'}(X)\}
=0.
\end{aligned}
\]
Hence,
\[
\E_{w_{u;v,v'}}\{e(X)\}=1-s.
\]

\textit{Step 2: Covariance representation.}
We have
\[
\begin{aligned}
\tau_{u,v}-\tau_{u,v'}
&=
\E\left[
\left\{
\frac{\beta_{u,v}(X)}
{\E\{\beta_{u,v}(X)\}}
-
\frac{\beta_{u,v'}(X)}
{\E\{\beta_{u,v'}(X)\}}
\right\}
\tau(X)
\right]\\
&=
\frac{1}{\E\{\beta_{u,v'}(X)\}}
\E\left[
\beta_{u,v}(X)
\left\{
\frac{\E\{\beta_{u,v'}(X)\}}
{\E\{\beta_{u,v}(X)\}}
-
\{1-e(X)\}^b
\right\}
\tau(X)
\right]\\
&=
\frac{1}{\E\{\beta_{u,v'}(X)\}}
\E\left[
w_{u;v,v'}(X)
\{e(X)-(1-s)\}\tau(X)
\right]\\
&=
\frac{\E\{w_{u;v,v'}(X)\}}
{\E\{\beta_{u,v'}(X)\}}
\E_{w_{u;v,v'}}
\left[
\{e(X)-\E_{w_{u;v,v'}}(e(X))\}\tau(X)
\right]\\
&=
\frac{
\E\{w_{u;v,v'}(X)\}
\cov_{w_{u;v,v'}}\{\tau(X),e(X)\}
}{
\E\{\beta_{u,v'}(X)\}
}.
\end{aligned}
\]
This proves part (ii) and completes the proof.
\end{proof}
\subsection{Proof of Corollary \ref{coro:beta weight}}
\begin{proof}
Each comparison weight in Theorem~\ref{thm:beta-one-coordinate-diff} is a
positive function of \(e(X)\). Hence, under any such weight \(w\),
\[
\cov\{\tau(X),e(X);w(X)\}
=
\cov\{g(e(X)),e(X);w(X)\}.
\]
Under Assumption~\ref{assump:mono}(i), this covariance is nonnegative.
The two representations in Theorem~\ref{thm:beta-one-coordinate-diff}
therefore imply that \(\tau_{u,v}\) is non-decreasing in \(u\) and
non-increasing in \(v\). Under Assumption~\ref{assump:mono}(ii), the
covariance signs reverse, yielding the reverse monotonicities.
\end{proof}

\subsection{Proof of Proposition~\ref{prop:additive linear}}

\begin{proof}

\proofpart{Preliminary notation and projection identities}

Let \(\tilde e(X)=\gamma_0+\gamma_1^\top X\) be the population OLS
projection of \(Z\) on \((1,X)\), and define the residual
\[
\tilde Z=Z-\tilde e(X).
\]
By the normal equations for this projection,
\begin{equation}\label{eq:proj-orth}
\E(\tilde Z)=0
\quad\text{and}\quad
\E(\tilde Z X)=0.
\end{equation}

\proofpart{Part (i): Decomposition of the OLS coefficient}

\textit{Residual representation.}
The population OLS coefficients \((\beta_0,\beta_1,\beta_2)\) satisfy
\begin{align*}
\E\!\left\{Y-\beta_0-\beta_1 Z-\beta_2^\top X\right\}&=0,\\
\E\!\left[X\{Y-\beta_0-\beta_1 Z-\beta_2^\top X\}\right]&=0,\\
\E\!\left[Z\{Y-\beta_0-\beta_1 Z-\beta_2^\top X\}\right]&=0.
\end{align*}
Because \(\tilde Z\) is a linear combination of \(1\), \(Z\), and \(X\),
these normal equations imply
\[
0
=
\E\!\left[
\tilde Z\{Y-\beta_0-\beta_1 Z-\beta_2^\top X\}
\right].
\]
Using \eqref{eq:proj-orth}, this reduces to
\[
0=\E(\tilde ZY)-\beta_1\E(\tilde Z Z),
\]
and hence
\begin{equation}\label{eq:beta1-resid}
\beta_1
=
\frac{\E(\tilde ZY)}
{\E(\tilde Z Z)}.
\end{equation}

\textit{Denominator.}
Conditioning on \(X\) gives
\[
\begin{aligned}
\E(\tilde Z Z\mid X)
&=
\E\!\left[\{Z-\tilde e(X)\}Z\mid X\right]\\
&=
\E(Z\mid X)-\tilde e(X)\E(Z\mid X)\\
&=
e(X)\{1-\tilde e(X)\}.
\end{aligned}
\]
Therefore,
\begin{equation}\label{eq:denom}
\E(\tilde Z Z)
=
\E\!\left[e(X)\{1-\tilde e(X)\}\right]
=
\E\{\tilde w(X)\},
\end{equation}
where
\[
\tilde w(X)=e(X)\{1-\tilde e(X)\}.
\]

\textit{Numerator.}
Write
\[
\mu_z(X)=\E\{Y(z)\mid X\},
\qquad
\tau(X)=\mu_1(X)-\mu_0(X).
\]
Under ignorability,
\[
\E(Y\mid Z=1,X)=\mu_1(X),
\qquad
\E(Y\mid Z=0,X)=\mu_0(X).
\]
Hence,
\[
\E(ZY\mid X)
=
\E\!\left[Z\,\E(Y\mid Z,X)\mid X\right]
=
e(X)\mu_1(X),
\]
and
\[
\E(Y\mid X)
=
e(X)\mu_1(X)
+
\{1-e(X)\}\mu_0(X).
\]
It follows that
\[
\begin{aligned}
\E(\tilde ZY\mid X)
&=
\E\!\left[\{Z-\tilde e(X)\}Y\mid X\right]\\
&=
\E(ZY\mid X)-\tilde e(X)\E(Y\mid X)\\
&=
e(X)\mu_1(X)
-
\tilde e(X)
\left[
e(X)\mu_1(X)+\{1-e(X)\}\mu_0(X)
\right]\\
&=
e(X)\{1-\tilde e(X)\}\mu_1(X)
-
\tilde e(X)\{1-e(X)\}\mu_0(X).
\end{aligned}
\]
Using \(\mu_1(X)=\mu_0(X)+\tau(X)\), we obtain
\begin{align}
\E(\tilde ZY\mid X)
&=
e(X)\{1-\tilde e(X)\}\tau(X)
+
\{e(X)-\tilde e(X)\}\mu_0(X)
\nonumber\\
&=
\tilde w(X)\tau(X)
+
\{e(X)-\tilde e(X)\}\mu_0(X).
\label{eq:num-cond}
\end{align}
Taking expectations in \eqref{eq:num-cond} gives
\begin{equation}\label{eq:num}
\E(\tilde ZY)
=
\E\{\tilde w(X)\tau(X)\}
+
\E\!\left[
\{e(X)-\tilde e(X)\}\mu_0(X)
\right].
\end{equation}

Combining \eqref{eq:beta1-resid}, \eqref{eq:denom}, and \eqref{eq:num},
we obtain
\[
\beta_1
=
\frac{\E\{\tilde w(X)\tau(X)\}}
{\E\{\tilde w(X)\}}
+
\frac{
\E\!\left[
\{e(X)-\tilde e(X)\}\mu_0(X)
\right]
}{
\E\{\tilde w(X)\}
}.
\]
This proves part (i).

\proofpart{Part (ii): Linear propensity score}

If \(e(X)\) lies in the linear span of \((1,X)\), then its population
linear projection satisfies
\[
\tilde e(X)=e(X)
\quad\text{almost surely}.
\]
Consequently, the second term in part (i) vanishes and
\[
\tilde w(X)=e(X)\{1-e(X)\}.
\]
It follows that
\[
\beta_1
=
\frac{
\E\!\left[e(X)\{1-e(X)\}\tau(X)\right]
}{
\E\!\left[e(X)\{1-e(X)\}\right]
}
=
\tauow.
\]
This proves part (ii).
\end{proof}
\subsection{Proof of Theorem~\ref{thm:IV}}

\begin{proof}

\proofpart{Preliminary TSLS decomposition}

We apply Proposition~\ref{prop:additive linear} twice: first to the
reduced-form regression of \(Y\) on \((1,Z,X)\), whose potential outcomes
with respect to \(Z\) are \(Y\{D(1)\}\) and \(Y\{D(0)\}\), and second to the
first-stage regression of \(D\) on \((1,Z,X)\), whose potential outcomes are
\(D(1)\) and \(D(0)\). Conditional instrument exogeneity justifies the
required ignorability condition.

It follows that
\[
\begin{aligned}
\TSLS
&=
\frac{
\bigl[
\E\{\tilde w(X)[Y\{D(1)\}-Y\{D(0)\}]\}
+
\E[\{e(X)-\tilde e(X)\}Y\{D(0)\}]
\bigr]/\E\{\tilde w(X)\}
}{
\bigl[
\E\{\tilde w(X)\{D(1)-D(0)\}\}
+
\E[\{e(X)-\tilde e(X)\}D(0)]
\bigr]/\E\{\tilde w(X)\}
}\\
&=
\frac{
\E\{\tilde w(X)[Y\{D(1)\}-Y\{D(0)\}]\}
+
\E[\{e(X)-\tilde e(X)\}Y\{D(0)\}]
}{
\E\{\tilde w(X)\{D(1)-D(0)\}\}
+
\E[\{e(X)-\tilde e(X)\}D(0)]
},
\end{aligned}
\]
where
\[
\tilde w(X)=e(X)\{1-\tilde e(X)\}.
\]

\proofpart{Part (i): General decomposition}

By the exclusion restriction and monotonicity,
\[
D(1)-D(0)=\mathbf 1(U=\textup c),
\qquad
Y\{D(1)\}-Y\{D(0)\}
=
\{Y(1)-Y(0)\}\mathbf 1(U=\textup c).
\]
Therefore,
\[
\begin{aligned}
\E\{\tilde w(X)\{D(1)-D(0)\}\}
&=
\E\{\tilde w(X)\mathbf 1(U=\textup c)\}\\
&=
\E\{\tilde w(X)\mid U=\textup c\}
\p(U=\textup c),
\end{aligned}
\]
and
\[
\begin{aligned}
&\E\{\tilde w(X)[Y\{D(1)\}-Y\{D(0)\}]\}\\
&\qquad=
\E\!\left[
\tilde w(X)\{Y(1)-Y(0)\}\mathbf 1(U=\textup c)
\right]\\
&\qquad=
\E\!\left[
\tilde w(X)\{Y(1)-Y(0)\}\mid U=\textup c
\right]
\p(U=\textup c).
\end{aligned}
\]

Moreover, under monotonicity,
\[
D(0)=\mathbf 1(U=\textup a),
\]
so
\[
\begin{aligned}
\E[\{e(X)-\tilde e(X)\}D(0)]
&=
\E\!\left[
\{e(X)-\tilde e(X)\}\mathbf 1(U=\textup a)
\right]\\
&=
\E[\{e(X)-\tilde e(X)\}\mid U=\textup a]
\p(U=\textup a).
\end{aligned}
\]
Substituting these identities into the preliminary TSLS decomposition yields
the expression in part~(i).

\proofpart{Part (ii): Linear IV propensity score}

If \(e(X)\) lies in the linear span of \((1,X)\), then
\[
\tilde e(X)=e(X)
\quad\text{almost surely}.
\]
Hence, all terms involving \(e(X)-\tilde e(X)\) vanish and
\[
\tilde w(X)=e(X)\{1-e(X)\}.
\]
Part~(i) therefore reduces to
\[
\begin{aligned}
\TSLS
&=
\frac{
\E\!\left[
e(X)\{1-e(X)\}\{Y(1)-Y(0)\}
\mid U=\textup c
\right]
}{
\E\!\left[
e(X)\{1-e(X)\}
\mid U=\textup c
\right]
}\\
&=
\E\left[
\frac{
e(X)\{1-e(X)\}
}{
\E[e(X)\{1-e(X)\}\mid U=\textup c]
}
\{Y(1)-Y(0)\}
\mathrel{\Big|} U=\textup c
\right].
\end{aligned}
\]
This proves the first representation in part~(ii).

\textit{Equivalent complier-CATE representation.}
The preceding expression can equivalently be written as
\[
\TSLS
=
\frac{
\E\!\left[
e(X)\{1-e(X)\}\{Y(1)-Y(0)\}
\mathbf 1(U=\textup c)
\right]
}{
\E\!\left[
e(X)\{1-e(X)\}
\mathbf 1(U=\textup c)
\right]
}.
\]
Taking conditional expectations given \(X\) gives
\[
\begin{aligned}
&\E\!\left[
e(X)\{1-e(X)\}\{Y(1)-Y(0)\}
\mathbf 1(U=\textup c)
\mid X
\right]\\
&\qquad=
e(X)\{1-e(X)\}
\pi^{\textup c}(X)
\tau^{\textup c}(X),
\end{aligned}
\]
and
\[
\E\!\left[
e(X)\{1-e(X)\}
\mathbf 1(U=\textup c)
\mid X
\right]
=
e(X)\{1-e(X)\}\pi^{\textup c}(X),
\]
where
\[
\tau^{\textup c}(X)
=
\E\{Y(1)-Y(0)\mid U=\textup c,X\}.
\]
Therefore,
\[
\TSLS
=
\E\left[
\frac{
e(X)\{1-e(X)\}\pi^{\textup c}(X)
}{
\E[e(X)\{1-e(X)\}\pi^{\textup c}(X)]
}
\tau^{\textup c}(X)
\right].
\]
This proves the second representation in part~(ii) and completes the proof.
\end{proof}
\subsection{Proof of Corollary \ref{coro:FE-bounds}}
\begin{proof}
In the stratified setting, \(e(X)=e_{[k]}\) for units in stratum \(k\), and
\[
\tau_{\mathrm{FE}}
=
\frac{\sum_k w_{[k]}e_{[k]}(1-e_{[k]})\tau_{[k]}}
{\sum_k w_{[k]}e_{[k]}(1-e_{[k]})}
=\tauow.
\]
Moreover, the conditional mean \(g(e)=\E\{\tau(X)\mid e(X)=e\}\) becomes
\(g_{\mathrm{FE}}(e)\). The result follows directly from
Corollary~\ref{coro:ow in atc att}.
\end{proof}

\subsection{Proof of Corollary \ref{coro:FE-bounds_old}}
\begin{proof}
If two strata have the same propensity score, either ``if and only if''
condition forces their treatment effects to coincide. The first condition
therefore makes \(g_{\mathrm{FE}}(e)\) non-decreasing, and the second makes it
non-increasing. The two conclusions follow from
Corollary~\ref{coro:FE-bounds}.
\end{proof}

\subsection{Proof of Theorem \ref{thm:sample-cp-geometry}}
\begin{proof}
For a generic nonnegative weight vector \(w\), define
\[
\bar e_w
=
\frac{\sum_{i=1}^n w_i\hat e_i}{\sum_{i=1}^n w_i},
\qquad
\hat V_w
=
\sum_{i=1}^n w_i(\hat e_i-\bar e_w)^2.
\]
Use the subscripts \(\mathrm{all}\), \(\mathrm{tr}\), and \(\mathrm{co}\)
for the three weight choices in the theorem.
The fitted slope satisfies
\[
\hat\beta_w
=
\frac{
\sum_{i=1}^n w_i(\hat e_i-\bar e_w)\hat\tau_i
}{
\sum_{i=1}^n w_i(\hat e_i-\bar e_w)^2
}
=
\frac{
\sum_{i=1}^n w_i(\hat e_i-\bar e_w)\hat\tau_i
}{\hat V_w},
\]
because \(\sum_i w_i(\hat e_i-\bar e_w)=0\).

First consider the unweighted fit, \(w_i=1\). Let
\[
\bar e=\frac1n\sum_{i=1}^n\hat e_i,
\qquad
\bar\tau=\frac1n\sum_{i=1}^n\hat\tau_i.
\]
Then
\[
\hat\tau_{\mathrm{ATT}}-\hat\tau_{\mathrm{ATE}}
=
\frac{\sum_i \hat e_i\hat\tau_i}{\sum_i\hat e_i}
-
\bar\tau
=
\frac{\sum_i(\hat e_i-\bar e)\hat\tau_i}{\sum_i\hat e_i}
=
\frac{\hat\beta_{\mathrm{all}}\hat V_{\mathrm{all}}}
{\sum_i\hat e_i}.
\]
Similarly,
\[
\hat\tau_{\mathrm{ATE}}-\hat\tau_{\mathrm{ATC}}
=
\bar\tau
-
\frac{\sum_i(1-\hat e_i)\hat\tau_i}{\sum_i(1-\hat e_i)}
=
\frac{\sum_i(\hat e_i-\bar e)\hat\tau_i}{\sum_i(1-\hat e_i)}
=
\frac{\hat\beta_{\mathrm{all}}\hat V_{\mathrm{all}}}
{\sum_i(1-\hat e_i)}.
\]

Next consider the fit with \(w_i=\hat e_i\). Let
\[
\bar e_{\mathrm{tr}}
=
\frac{\sum_i \hat e_i^2}{\sum_i\hat e_i}.
\]
Then
\[
\hat\tau_{\mathrm{ATT}}-\hat\tau_{\mathrm{ATO}}
=
\frac{\sum_i \hat e_i\hat\tau_i}{\sum_i\hat e_i}
-
\frac{\sum_i \hat e_i(1-\hat e_i)\hat\tau_i}
{\sum_i \hat e_i(1-\hat e_i)}.
\]
Rearranging the right-hand side gives
\[
\hat\tau_{\mathrm{ATT}}-\hat\tau_{\mathrm{ATO}}
=
\frac{
\sum_i \hat e_i(\hat e_i-\bar e_{\mathrm{tr}})\hat\tau_i
}{
\sum_i \hat e_i(1-\hat e_i)
}
=
\frac{\hat\beta_{\mathrm{tr}}\hat V_{\mathrm{tr}}}
{\sum_i \hat e_i(1-\hat e_i)}.
\]

Finally consider the fit with \(w_i=1-\hat e_i\). Let
\[
\bar e_{\mathrm{co}}
=
\frac{\sum_i (1-\hat e_i)\hat e_i}
{\sum_i(1-\hat e_i)}.
\]
Then
\[
\hat\tau_{\mathrm{ATO}}-\hat\tau_{\mathrm{ATC}}
=
\frac{\sum_i \hat e_i(1-\hat e_i)\hat\tau_i}
{\sum_i \hat e_i(1-\hat e_i)}
-
\frac{\sum_i (1-\hat e_i)\hat\tau_i}{\sum_i(1-\hat e_i)}.
\]
Rearranging gives
\[
\hat\tau_{\mathrm{ATO}}-\hat\tau_{\mathrm{ATC}}
=
\frac{
\sum_i(1-\hat e_i)(\hat e_i-\bar e_{\mathrm{co}})\hat\tau_i
}{
\sum_i \hat e_i(1-\hat e_i)
}
=
\frac{\hat\beta_{\mathrm{co}}\hat V_{\mathrm{co}}}
{\sum_i \hat e_i(1-\hat e_i)}.
\]

The corresponding horizontal differences are
\[
\widehat{\bar e}_{\mathrm{ATT}}
-
\widehat{\bar e}_{\mathrm{ATE}}
=
\frac{\hat V_{\mathrm{all}}}{\sum_i\hat e_i},
\qquad
\widehat{\bar e}_{\mathrm{ATE}}
-
\widehat{\bar e}_{\mathrm{ATC}}
=
\frac{\hat V_{\mathrm{all}}}{\sum_i(1-\hat e_i)},
\]
and
\[
\widehat{\bar e}_{\mathrm{ATT}}
-
\widehat{\bar e}_{\mathrm{ATO}}
=
\frac{\hat V_{\mathrm{tr}}}
{\sum_i\hat e_i(1-\hat e_i)},
\qquad
\widehat{\bar e}_{\mathrm{ATO}}
-
\widehat{\bar e}_{\mathrm{ATC}}
=
\frac{\hat V_{\mathrm{co}}}
{\sum_i\hat e_i(1-\hat e_i)}.
\]
Dividing each vertical difference by its corresponding horizontal difference
gives the slope identities in parts~(i) and~(ii).

The weighted least-squares normal equations imply
\[
\widehat L_{\mathrm{all}}
(\widehat{\bar e}_{\mathrm{ATE}})
=
\hat\tau_{\mathrm{ATE}},
\qquad
\widehat L_{\mathrm{tr}}
(\widehat{\bar e}_{\mathrm{ATT}})
=
\hat\tau_{\mathrm{ATT}},
\qquad
\widehat L_{\mathrm{co}}
(\widehat{\bar e}_{\mathrm{ATC}})
=
\hat\tau_{\mathrm{ATC}}.
\]
The slope identities then imply that the ATT centroid also lies on
\(\widehat L_{\mathrm{all}}\), the ATO centroid lies on both
\(\widehat L_{\mathrm{tr}}\) and \(\widehat L_{\mathrm{co}}\), and the ATC
centroid also lies on \(\widehat L_{\mathrm{all}}\). These are therefore the
three displayed common points. Whenever each corresponding pair of lines is
distinct, the common point is unique.
\end{proof}

\subsection{Proof of Theorem \ref{thm:population-local-cp-geometry}}

\begin{proof}
Write
\[
E=e(X),
\qquad
T=\tau^{\textup c}(X),
\qquad
\Pi=\pi^{\textup c}(X).
\]
For any integrable function \(a(X)\), define expectation under the complier
distribution by
\[
\E_{\textup c}\{a(X)\}
=
\frac{\E\{\Pi a(X)\}}{\E(\Pi)}.
\]
Under Assumption~\ref{assump:iv}, this equals
\(\E\{a(X)\mid U=\textup c\}\). Define the encouraged- and
unencouraged-complier expectations by
\[
\E_{\mathrm{tr}}^{\textup c}\{a(X)\}
=
\frac{\E_{\textup c}\{E a(X)\}}
{\E_{\textup c}(E)},
\qquad
\E_{\mathrm{co}}^{\textup c}\{a(X)\}
=
\frac{\E_{\textup c}\{(1-E)a(X)\}}
{\E_{\textup c}(1-E)}.
\]
Let \(\cov_j^{\textup c}\) and \(\var_j^{\textup c}\) denote covariance
and variance under the corresponding expectation, where
\(j\in\{\mathrm{all},\mathrm{tr},\mathrm{co}\}\), with
\(\E_{\mathrm{all}}^{\textup c}=\E_{\textup c}\).

The normal equations for weighted linear projections imply that each line
passes through its corresponding weighted centroid. Hence,
\[
L_{\mathrm{all}}^{\textup c}
(\bar e_{\mathrm{ATE}}^{\textup c})
=
\tau_{\mathrm{ATE}}^{\textup c},
\]
\[
L_{\mathrm{tr}}^{\textup c}
(\bar e_{\mathrm{ATT}}^{\textup c})
=
\tau_{\mathrm{ATT}}^{\textup c},
\qquad
L_{\mathrm{co}}^{\textup c}
(\bar e_{\mathrm{ATC}}^{\textup c})
=
\tau_{\mathrm{ATC}}^{\textup c}.
\]
Moreover, their slopes are
\[
\beta_{\mathrm{all}}^{\textup c}
=
\frac{
\cov_{\mathrm{all}}^{\textup c}(T,E)
}{
\var_{\mathrm{all}}^{\textup c}(E)
},
\qquad
\beta_{\mathrm{tr}}^{\textup c}
=
\frac{
\cov_{\mathrm{tr}}^{\textup c}(T,E)
}{
\var_{\mathrm{tr}}^{\textup c}(E)
},
\]
and
\[
\beta_{\mathrm{co}}^{\textup c}
=
\frac{
\cov_{\mathrm{co}}^{\textup c}(T,E)
}{
\var_{\mathrm{co}}^{\textup c}(E)
}.
\]

We first prove the slope identities. Under the complier distribution,
\[
\bar e_{\mathrm{ATT}}^{\textup c}
=
\frac{\E_{\textup c}(E^2)}
{\E_{\textup c}(E)},
\qquad
\bar e_{\mathrm{ATE}}^{\textup c}
=
\E_{\textup c}(E),
\]
so
\[
\bar e_{\mathrm{ATT}}^{\textup c}
-
\bar e_{\mathrm{ATE}}^{\textup c}
=
\frac{
\var_{\textup c}(E)
}{
\E_{\textup c}(E)
}.
\]
Likewise,
\[
\bar e_{\mathrm{ATC}}^{\textup c}
=
\frac{
\E_{\textup c}\{E(1-E)\}
}{
\E_{\textup c}(1-E)
},
\]
which gives
\[
\bar e_{\mathrm{ATE}}^{\textup c}
-
\bar e_{\mathrm{ATC}}^{\textup c}
=
\frac{
\var_{\textup c}(E)
}{
\E_{\textup c}(1-E)
}.
\]

The corresponding differences between the local weighted estimands are
\[
\tau_{\mathrm{ATT}}^{\textup c}
-
\tau_{\mathrm{ATE}}^{\textup c}
=
\frac{
\cov_{\textup c}(T,E)
}{
\E_{\textup c}(E)
},
\]
and
\[
\tau_{\mathrm{ATE}}^{\textup c}
-
\tau_{\mathrm{ATC}}^{\textup c}
=
\frac{
\cov_{\textup c}(T,E)
}{
\E_{\textup c}(1-E)
}.
\]
Dividing the treatment-effect differences by the corresponding
propensity-score differences yields
\[
\frac{
\tau_{\mathrm{ATT}}^{\textup c}
-
\tau_{\mathrm{ATE}}^{\textup c}
}{
\bar e_{\mathrm{ATT}}^{\textup c}
-
\bar e_{\mathrm{ATE}}^{\textup c}
}
=
\frac{
\cov_{\textup c}(T,E)
}{
\var_{\textup c}(E)
}
=
\beta_{\mathrm{all}}^{\textup c},
\]
and
\[
\frac{
\tau_{\mathrm{ATE}}^{\textup c}
-
\tau_{\mathrm{ATC}}^{\textup c}
}{
\bar e_{\mathrm{ATE}}^{\textup c}
-
\bar e_{\mathrm{ATC}}^{\textup c}
}
=
\frac{
\cov_{\textup c}(T,E)
}{
\var_{\textup c}(E)
}
=
\beta_{\mathrm{all}}^{\textup c}.
\]
This proves part~(i).

We next consider the encouraged-complier distribution. The local ATO
distribution is obtained by tilting this distribution by \(1-E\). Therefore,
\[
\bar e_{\mathrm{ATO}}^{\textup c}
=
\frac{
\E_{\mathrm{tr}}^{\textup c}\{E(1-E)\}
}{
\E_{\mathrm{tr}}^{\textup c}(1-E)
}.
\]
Since
\(
\bar e_{\mathrm{ATT}}^{\textup c}
=
\E_{\mathrm{tr}}^{\textup c}(E)
\),
a direct calculation gives
\[
\bar e_{\mathrm{ATT}}^{\textup c}
-
\bar e_{\mathrm{ATO}}^{\textup c}
=
\frac{
\var_{\mathrm{tr}}^{\textup c}(E)
}{
\E_{\mathrm{tr}}^{\textup c}(1-E)
}.
\]
Similarly,
\[
\tau_{\mathrm{ATT}}^{\textup c}
-
\tau_{\mathrm{ATO}}^{\textup c}
=
\frac{
\cov_{\mathrm{tr}}^{\textup c}(T,E)
}{
\E_{\mathrm{tr}}^{\textup c}(1-E)
}.
\]
Consequently,
\[
\frac{
\tau_{\mathrm{ATT}}^{\textup c}
-
\tau_{\mathrm{ATO}}^{\textup c}
}{
\bar e_{\mathrm{ATT}}^{\textup c}
-
\bar e_{\mathrm{ATO}}^{\textup c}
}
=
\frac{
\cov_{\mathrm{tr}}^{\textup c}(T,E)
}{
\var_{\mathrm{tr}}^{\textup c}(E)
}
=
\beta_{\mathrm{tr}}^{\textup c}.
\]

Equivalently, the local ATO distribution can be obtained by tilting the
unencouraged-complier distribution by \(E\). Hence,
\[
\bar e_{\mathrm{ATO}}^{\textup c}
-
\bar e_{\mathrm{ATC}}^{\textup c}
=
\frac{
\var_{\mathrm{co}}^{\textup c}(E)
}{
\E_{\mathrm{co}}^{\textup c}(E)
},
\]
and
\[
\tau_{\mathrm{ATO}}^{\textup c}
-
\tau_{\mathrm{ATC}}^{\textup c}
=
\frac{
\cov_{\mathrm{co}}^{\textup c}(T,E)
}{
\E_{\mathrm{co}}^{\textup c}(E)
}.
\]
It follows that
\[
\frac{
\tau_{\mathrm{ATO}}^{\textup c}
-
\tau_{\mathrm{ATC}}^{\textup c}
}{
\bar e_{\mathrm{ATO}}^{\textup c}
-
\bar e_{\mathrm{ATC}}^{\textup c}
}
=
\frac{
\cov_{\mathrm{co}}^{\textup c}(T,E)
}{
\var_{\mathrm{co}}^{\textup c}(E)
}
=
\beta_{\mathrm{co}}^{\textup c}.
\]
This proves part~(ii).

It remains to establish the intersection results. Because
\(L_{\mathrm{all}}^{\textup c}\) passes through the local ATE centroid,
part~(i) implies
\[
\begin{aligned}
L_{\mathrm{all}}^{\textup c}
(\bar e_{\mathrm{ATT}}^{\textup c})
&=
\tau_{\mathrm{ATE}}^{\textup c}
+
\beta_{\mathrm{all}}^{\textup c}
\left(
\bar e_{\mathrm{ATT}}^{\textup c}
-
\bar e_{\mathrm{ATE}}^{\textup c}
\right)=
\tau_{\mathrm{ATT}}^{\textup c},
\end{aligned}
\]
and
\[
\begin{aligned}
L_{\mathrm{all}}^{\textup c}
(\bar e_{\mathrm{ATC}}^{\textup c})
&=
\tau_{\mathrm{ATE}}^{\textup c}
+
\beta_{\mathrm{all}}^{\textup c}
\left(
\bar e_{\mathrm{ATC}}^{\textup c}
-
\bar e_{\mathrm{ATE}}^{\textup c}
\right)=
\tau_{\mathrm{ATC}}^{\textup c}.
\end{aligned}
\]
Since
\[
L_{\mathrm{tr}}^{\textup c}
(\bar e_{\mathrm{ATT}}^{\textup c})
=
\tau_{\mathrm{ATT}}^{\textup c},
\qquad
L_{\mathrm{co}}^{\textup c}
(\bar e_{\mathrm{ATC}}^{\textup c})
=
\tau_{\mathrm{ATC}}^{\textup c},
\]
the local ATT point lies on both
\(L_{\mathrm{all}}^{\textup c}\) and
\(L_{\mathrm{tr}}^{\textup c}\), whereas the local ATC point lies on both
\(L_{\mathrm{co}}^{\textup c}\) and
\(L_{\mathrm{all}}^{\textup c}\).

Finally, part~(ii) gives
\[
\begin{aligned}
L_{\mathrm{tr}}^{\textup c}
(\bar e_{\mathrm{ATO}}^{\textup c})
&=
\tau_{\mathrm{ATT}}^{\textup c}
+
\beta_{\mathrm{tr}}^{\textup c}
\left(
\bar e_{\mathrm{ATO}}^{\textup c}
-
\bar e_{\mathrm{ATT}}^{\textup c}
\right)=
\tau_{\mathrm{ATO}}^{\textup c},
\end{aligned}
\]
and
\[
\begin{aligned}
L_{\mathrm{co}}^{\textup c}
(\bar e_{\mathrm{ATO}}^{\textup c})
&=
\tau_{\mathrm{ATC}}^{\textup c}
+
\beta_{\mathrm{co}}^{\textup c}
\left(
\bar e_{\mathrm{ATO}}^{\textup c}
-
\bar e_{\mathrm{ATC}}^{\textup c}
\right)=
\tau_{\mathrm{ATO}}^{\textup c}.
\end{aligned}
\]
Therefore,
\[
\left(
\bar e_{\mathrm{ATT}}^{\textup c},
\tau_{\mathrm{ATT}}^{\textup c}
\right)
\in
L_{\mathrm{all}}^{\textup c}
\cap
L_{\mathrm{tr}}^{\textup c},
\]
\[
\left(
\bar e_{\mathrm{ATO}}^{\textup c},
\tau_{\mathrm{ATO}}^{\textup c}
\right)
\in
L_{\mathrm{tr}}^{\textup c}
\cap
L_{\mathrm{co}}^{\textup c},
\]
and
\[
\left(
\bar e_{\mathrm{ATC}}^{\textup c},
\tau_{\mathrm{ATC}}^{\textup c}
\right)
\in
L_{\mathrm{co}}^{\textup c}
\cap
L_{\mathrm{all}}^{\textup c}.
\]
Whenever each corresponding pair of lines is distinct, two affine lines have
at most one common point. Hence, the displayed common points are their unique
intersections.
\end{proof}

\subsection{Proof of Theorem \ref{thm:sample-local-cp-geometry}}
\begin{proof}
For a generic nonnegative weight vector \(w\), define
\[
\bar e_w
=
\frac{\sum_{i=1}^n w_i\hat e_i}{\sum_{i=1}^n w_i},
\qquad
\hat V_w
=
\sum_{i=1}^n w_i(\hat e_i-\bar e_w)^2.
\]
For the three local weight choices in the theorem, write
\(\hat V_j^{\textup c}=\hat V_{w_j^{\textup c}}\) for
\(j\in\{\mathrm{all},\mathrm{tr},\mathrm{co}\}\). The fitted slope satisfies
\[
\hat\beta_w^{\textup c}
=
\frac{
\sum_{i=1}^n w_i(\hat e_i-\bar e_w)\hat\tau_i^{\textup c}
}{
\sum_{i=1}^n w_i(\hat e_i-\bar e_w)^2
}
=
\frac{
\sum_{i=1}^n w_i(\hat e_i-\bar e_w)\hat\tau_i^{\textup c}
}{
\hat V_w
},
\]
because \(\sum_i w_i(\hat e_i-\bar e_w)=0\).

First consider the complier-weighted fit, \(w_i=\hat\pi_i^{\textup c}\). Let
\[
\bar e_{\mathrm{all}}^{\textup c}
=
\frac{
\sum_i \hat\pi_i^{\textup c}\hat e_i
}{
\sum_i \hat\pi_i^{\textup c}
}.
\]
Then
\[
\hat\tau^{\textup c}_{\mathrm{ATT}}
-
\hat\tau^{\textup c}_{\mathrm{ATE}}
=
\frac{\sum_i \hat e_i\hat\pi_i^{\textup c}\hat\tau_i^{\textup c}}
{\sum_i \hat e_i\hat\pi_i^{\textup c}}
-
\frac{\sum_i \hat\pi_i^{\textup c}\hat\tau_i^{\textup c}}
{\sum_i \hat\pi_i^{\textup c}}
=
\frac{
\sum_i
\hat\pi_i^{\textup c}
(\hat e_i-\bar e_{\mathrm{all}}^{\textup c})
\hat\tau_i^{\textup c}
}{
\sum_i \hat e_i\hat\pi_i^{\textup c}
}
=
\frac{
\hat\beta_{\mathrm{all}}^{\textup c}
\hat V_{\mathrm{all}}^{\textup c}
}{
\sum_i \hat e_i\hat\pi_i^{\textup c}
}.
\]
Similarly,
\[
\hat\tau^{\textup c}_{\mathrm{ATE}}
-
\hat\tau^{\textup c}_{\mathrm{ATC}}
=
\frac{\sum_i \hat\pi_i^{\textup c}\hat\tau_i^{\textup c}}
{\sum_i \hat\pi_i^{\textup c}}
-
\frac{\sum_i (1-\hat e_i)\hat\pi_i^{\textup c}\hat\tau_i^{\textup c}}
{\sum_i (1-\hat e_i)\hat\pi_i^{\textup c}}
=
\frac{
\sum_i
\hat\pi_i^{\textup c}
(\hat e_i-\bar e_{\mathrm{all}}^{\textup c})
\hat\tau_i^{\textup c}
}{
\sum_i (1-\hat e_i)\hat\pi_i^{\textup c}
}
=
\frac{
\hat\beta_{\mathrm{all}}^{\textup c}
\hat V_{\mathrm{all}}^{\textup c}
}{
\sum_i (1-\hat e_i)\hat\pi_i^{\textup c}
}.
\]

Next consider the fit with \(w_i=\hat e_i\hat\pi_i^{\textup c}\). Let
\[
\bar e_{\mathrm{tr}}^{\textup c}
=
\frac{
\sum_i \hat e_i^2\hat\pi_i^{\textup c}
}{
\sum_i \hat e_i\hat\pi_i^{\textup c}
}.
\]
Then
\[
\hat\tau^{\textup c}_{\mathrm{ATT}}
-
\hat\tau^{\textup c}_{\mathrm{ATO}}
=
\frac{\sum_i \hat e_i\hat\pi_i^{\textup c}\hat\tau_i^{\textup c}}
{\sum_i \hat e_i\hat\pi_i^{\textup c}}
-
\frac{
\sum_i \hat e_i(1-\hat e_i)\hat\pi_i^{\textup c}\hat\tau_i^{\textup c}
}{
\sum_i \hat e_i(1-\hat e_i)\hat\pi_i^{\textup c}
}.
\]
Rearranging the right-hand side gives
\[
\hat\tau^{\textup c}_{\mathrm{ATT}}
-
\hat\tau^{\textup c}_{\mathrm{ATO}}
=
\frac{
\sum_i
\hat e_i\hat\pi_i^{\textup c}
(\hat e_i-\bar e_{\mathrm{tr}}^{\textup c})
\hat\tau_i^{\textup c}
}{
\sum_i \hat e_i(1-\hat e_i)\hat\pi_i^{\textup c}
}
=
\frac{
\hat\beta_{\mathrm{tr}}^{\textup c}
\hat V_{\mathrm{tr}}^{\textup c}
}{
\sum_i \hat e_i(1-\hat e_i)\hat\pi_i^{\textup c}
}.
\]

Finally consider the fit with \(w_i=(1-\hat e_i)\hat\pi_i^{\textup c}\). Let
\[
\bar e_{\mathrm{co}}^{\textup c}
=
\frac{
\sum_i (1-\hat e_i)\hat e_i\hat\pi_i^{\textup c}
}{
\sum_i (1-\hat e_i)\hat\pi_i^{\textup c}
}.
\]
Then
\[
\hat\tau^{\textup c}_{\mathrm{ATO}}
-
\hat\tau^{\textup c}_{\mathrm{ATC}}
=
\frac{
\sum_i \hat e_i(1-\hat e_i)\hat\pi_i^{\textup c}\hat\tau_i^{\textup c}
}{
\sum_i \hat e_i(1-\hat e_i)\hat\pi_i^{\textup c}
}
-
\frac{
\sum_i (1-\hat e_i)\hat\pi_i^{\textup c}\hat\tau_i^{\textup c}
}{
\sum_i (1-\hat e_i)\hat\pi_i^{\textup c}
}.
\]
Rearranging gives
\[
\hat\tau^{\textup c}_{\mathrm{ATO}}
-
\hat\tau^{\textup c}_{\mathrm{ATC}}
=
\frac{
\sum_i
(1-\hat e_i)\hat\pi_i^{\textup c}
(\hat e_i-\bar e_{\mathrm{co}}^{\textup c})
\hat\tau_i^{\textup c}
}{
\sum_i \hat e_i(1-\hat e_i)\hat\pi_i^{\textup c}
}
=
\frac{
\hat\beta_{\mathrm{co}}^{\textup c}
\hat V_{\mathrm{co}}^{\textup c}
}{
\sum_i \hat e_i(1-\hat e_i)\hat\pi_i^{\textup c}
}.
\]

The corresponding horizontal differences are
\[
\widehat{\bar e}_{\mathrm{ATT}}^{\textup c}
-
\widehat{\bar e}_{\mathrm{ATE}}^{\textup c}
=
\frac{
\hat V_{\mathrm{all}}^{\textup c}
}{
\sum_i\hat e_i\hat\pi_i^{\textup c}
},
\qquad
\widehat{\bar e}_{\mathrm{ATE}}^{\textup c}
-
\widehat{\bar e}_{\mathrm{ATC}}^{\textup c}
=
\frac{
\hat V_{\mathrm{all}}^{\textup c}
}{
\sum_i(1-\hat e_i)\hat\pi_i^{\textup c}
},
\]
and
\[
\widehat{\bar e}_{\mathrm{ATT}}^{\textup c}
-
\widehat{\bar e}_{\mathrm{ATO}}^{\textup c}
=
\frac{
\hat V_{\mathrm{tr}}^{\textup c}
}{
\sum_i\hat e_i(1-\hat e_i)\hat\pi_i^{\textup c}
},
\qquad
\widehat{\bar e}_{\mathrm{ATO}}^{\textup c}
-
\widehat{\bar e}_{\mathrm{ATC}}^{\textup c}
=
\frac{
\hat V_{\mathrm{co}}^{\textup c}
}{
\sum_i\hat e_i(1-\hat e_i)\hat\pi_i^{\textup c}
}.
\]
Dividing each vertical difference by its corresponding horizontal difference
gives the slope identities.

The weighted least-squares normal equations imply
\[
\widehat L_{\mathrm{all}}^{\textup c}
(\widehat{\bar e}_{\mathrm{ATE}}^{\textup c})
=
\hat\tau_{\mathrm{ATE}}^{\textup c},
\qquad
\widehat L_{\mathrm{tr}}^{\textup c}
(\widehat{\bar e}_{\mathrm{ATT}}^{\textup c})
=
\hat\tau_{\mathrm{ATT}}^{\textup c},
\]
\[
\widehat L_{\mathrm{co}}^{\textup c}
(\widehat{\bar e}_{\mathrm{ATC}}^{\textup c})
=
\hat\tau_{\mathrm{ATC}}^{\textup c}.
\]
The slope identities imply that the local ATT centroid also lies on
\(\widehat L_{\mathrm{all}}^{\textup c}\), the local ATO centroid lies on
both \(\widehat L_{\mathrm{tr}}^{\textup c}\) and
\(\widehat L_{\mathrm{co}}^{\textup c}\), and the local ATC centroid also
lies on \(\widehat L_{\mathrm{all}}^{\textup c}\). These are the displayed
common points, and they are unique whenever the corresponding pairs of lines
are distinct.
\end{proof}

\subsection{Proof of Proposition \ref{prop:linear-ow}}\label{subsec:proof propos linear}
\begin{proof}
If \(e(X)\) is constant, the result is immediate. Hence suppose that \(e(X)\)
is not constant. If \(b=0\), then \(g(e)\) is constant and
\(\tauow=\tauatt=\tauatc\), so \eqref{eq:ow= l atc+ 1-l att} holds for the
displayed \(\lambda\). Now suppose that \(b\ne0\). To derive this weight, write
\[
\frac{\E[e(X)\{1 - e(X)\} \tau(X)]}{\E[e(X)\{1 - e(X)\}]} = \lambda \cdot \frac{\E[e(X) \tau(X)]}{\E[e(X)]} + (1 - \lambda) \cdot \frac{\E[\{1 - e(X)\} \tau(X)]}{\E[1 - e(X)]}.
\]
Substituting \(\E\{\tau(X)\mid e(X)\}=a+be(X)\), canceling the common
intercept, and dividing by \(b\) yields
\[
\frac{\E[e^2(X)\{1 - e(X)\}]}{\E[e(X)\{1 - e(X)\}]} = \lambda \cdot \frac{\E[e^2(X)]}{\E[e(X)]} + (1 - \lambda) \cdot \frac{\E[\{1 - e(X)\} e(X)]}{\E[1 - e(X)]}.
\]
Then, we have 
\begin{align*}
    \lambda = \frac{ \dfrac{\E[e^2(X)(1 - e(X))]}{\E[e(X)\{1 - e(X)\}]} - \dfrac{\E[\{1 - e(X)\} e(X)]}{\E[1 - e(X)]} }{ \dfrac{\E[e^2(X)]}{\E[e(X)]} - \dfrac{\E[\{1 - e(X)\} e(X)]}{\E[1 - e(X)]} }.
\end{align*}
This simplifies to:
\begin{align}\label{eq:lambda,abc}
    \lambda = \frac{ \dfrac{e_2 - e_3}{e_1 - e_2} - \dfrac{e_1 - e_2}{1 - e_1} }{ \dfrac{e_2}{e_1} - \dfrac{e_1 - e_2}{1 - e_1} }.
\end{align}
After algebraic manipulation, we obtain:
\begin{align*}
\lambda = \frac{ (1 - e_1)e_1(e_2 - e_3) - e_1(e_1 - e_2)^2 }{ (e_1 - e_2)(e_2 - e_1^2) } = \frac{e_1[ (1 - e_1)(e_2 - e_3) - (e_1 - e_2)^2 ]}{ (e_1 - e_2)(e_2 - e_1^2) }.
\end{align*}
Because $e_2-e_1^2=\var\{e(X)\}$, we have
\begin{align*}
\lambda=&\frac{\E\{ e(X)\} \left[\E\{1-e(X)\} \E [e^2(X)\{1-e(X)\}]-\E [e(X)\{1-e(X)\}]^2\right]}{\E [e(X)\{1-e(X)\}]\var\{e(X)\}}.
\end{align*}

Let \(A=1-e(X)\) and \(B=e(X)\). By the definition
\[
c^*=\frac{\E(AB)}{\E(A)},
\]
we have
\begin{align*}
\E(A)\E(AB^2)-\{\E(AB)\}^2
&=\E(A)\left\{\E(AB^2)-2c^*\E(AB)+(c^*)^2\E(A)\right\}\\
&=\E(A)\E\{A(B-c^*)^2\}.
\end{align*}
Returning to \(A=1-e(X)\) and \(B=e(X)\), this gives
\[
\E\{1-e(X)\}\E[e^2(X)\{1-e(X)\}]
-\E[e(X)\{1-e(X)\}]^2
=
\E\{1-e(X)\}\E[\{1-e(X)\}\{e(X)-c^*\}^2].
\]
Then we can finally obtain 
\begin{align*}
\lambda=&\frac{\E\{ e(X)\}\E\{ 1-e(X)\} \E[\{1-e(X)\}\{e(X)-c^*\}^2]}{\E \{e(X)(1-e(X))\}\var\{e(X)\}}.
\end{align*}
\end{proof}

\subsection{Proof of Proposition \ref{prop:linear beta}}

\begin{proof}
If \(e(X)\) is constant, the result is immediate. Suppose that \(e(X)\) is
not constant, and define
\[
m_{r,s}
=
\frac{\E\{\beta_{r+1,s}(X)\}}{\E\{\beta_{r,s}(X)\}}.
\]
Tilting a beta-weighted distribution by \(e(X)\) strictly increases its mean
propensity score, whereas tilting it by \(1-e(X)\) strictly decreases that
mean. Therefore,
\[
m_{u,v+1}<m_{u+1,v+1}<m_{u+1,v},
\]
so the value in \eqref{eq:lambda beta} lies in \((0,1)\).

If \(b=0\), all three estimands in
\eqref{eq:tau u+1 v+1 lambda tau u+1 v +(1-lambda) tau u, v+1} are equal,
and the identity follows. Now suppose that \(b\ne0\). Writing the desired
identity in weighted-average form gives
$$
\frac{\E[\beta_{u+1,v+1}(X) \tau(X)]}{\E[\beta_{u+1,v+1}(X)]} = \lambda \cdot \frac{\E[\beta_{u+1,v}(X) \tau(X)]}{\E[\beta_{u+1,v}(X)]} + (1 - \lambda) \cdot \frac{\E[\beta_{u,v+1}(X) \tau(X)]}{\E[\beta_{u,v+1}(X)]}.
$$
Substituting \(\E\{\tau(X)\mid e(X)\}=a+be(X)\), canceling the common
intercept, and dividing by \(b\) yields
$$
\frac{\E[\beta_{u+2,v+1}(X)]}{\E[\beta_{u+1,v+1}(X)]} = \lambda \cdot \frac{\E[\beta_{u+2,v}(X)]}{\E[\beta_{u+1,v}(X)]} + (1 - \lambda) \cdot \frac{\E[\beta_{u+1,v+1}(X)]}{\E[\beta_{u,v+1}(X)]}.
$$
It leads to \eqref{eq:lambda beta}.
\end{proof}

\subsection{Proof of Theorem~\ref{thm:local-beta-one-coordinate-diff}}

\begin{proof}
We first prove part (i). Let \(a=u'-u>0\) and write
\(r=r^{\textup c}_{u,u';v}\). By the definition of
\(r^{\textup c}_{u,u';v}\),
\[
r^a
=
\frac{\E\{\beta_{u',v}(X)\pi^{\textup c}(X)\}}
{\E\{\beta_{u,v}(X)\pi^{\textup c}(X)\}}.
\]
Because
\[
\beta_{u',v}(X)=\beta_{u,v}(X)e(X)^a,
\]
and \(e(X)\in(0,1)\), we have \(r\in(0,1)\). The function
\(t\mapsto t^a\) is strictly increasing on \((0,1)\), so
\[
\frac{e(X)^a-r^a}{e(X)-r}>0
\]
when \(e(X)\ne r\), and its value at \(e(X)=r\) equals the derivative
\(a r^{a-1}>0\). Therefore \(w^{\textup c}_{u,u';v}(X)>0\).

By the definition of \(w^{\textup c}_{u,u';v}(X)\),
\[
w^{\textup c}_{u,u';v}(X)\{e(X)-r\}
=
\beta_{u,v}(X)\pi^{\textup c}(X)
\{e(X)^a-r^a\}.
\]
Taking expectations gives
\[
\begin{aligned}
\E\!\left[
w^{\textup c}_{u,u';v}(X)\{e(X)-r\}
\right]
=
\E\{\beta_{u',v}(X)\pi^{\textup c}(X)\}
-
r^a\E\{\beta_{u,v}(X)\pi^{\textup c}(X)\}=0.
\end{aligned}
\]
Hence
\[
\E_{w^{\textup c}_{u,u';v}}\{e(X)\}=r.
\]

Now,
\[
\begin{aligned}
\tau^{\textup c}_{u',v}-\tau^{\textup c}_{u,v}
&=
\E\left[
\left\{
\frac{\beta_{u',v}(X)\pi^{\textup c}(X)}
{\E\{\beta_{u',v}(X)\pi^{\textup c}(X)\}}
-
\frac{\beta_{u,v}(X)\pi^{\textup c}(X)}
{\E\{\beta_{u,v}(X)\pi^{\textup c}(X)\}}
\right\}
\tau^{\textup c}(X)
\right]\\
&=
\frac{1}{\E\{\beta_{u',v}(X)\pi^{\textup c}(X)\}}
\E\left[
\beta_{u,v}(X)\pi^{\textup c}(X)
\{e(X)^a-r^a\}
\tau^{\textup c}(X)
\right]\\
&=
\frac{1}{\E\{\beta_{u',v}(X)\pi^{\textup c}(X)\}}
\E\left[
w^{\textup c}_{u,u';v}(X)\{e(X)-r\}
\tau^{\textup c}(X)
\right]\\
&=
\frac{\E\{w^{\textup c}_{u,u';v}(X)\}}
{\E\{\beta_{u',v}(X)\pi^{\textup c}(X)\}}
\E_{w^{\textup c}_{u,u';v}}
\left[
\{e(X)-\E_{w^{\textup c}_{u,u';v}}(e(X))\}
\tau^{\textup c}(X)
\right]\\
&=
\frac{
\E\{w^{\textup c}_{u,u';v}(X)\}
\cov_{w^{\textup c}_{u,u';v}}\{\tau^{\textup c}(X),e(X)\}
}{
\E\{\beta_{u',v}(X)\pi^{\textup c}(X)\}
}.
\end{aligned}
\]
This proves part (i).

We next prove part (ii). Let \(b=v'-v>0\) and write
\(s=r^{\textup c}_{u;v,v'}\). By the definition of
\(r^{\textup c}_{u;v,v'}\),
\[
s^b
=
\frac{\E\{\beta_{u,v'}(X)\pi^{\textup c}(X)\}}
{\E\{\beta_{u,v}(X)\pi^{\textup c}(X)\}}.
\]
Because
\[
\beta_{u,v'}(X)=\beta_{u,v}(X)\{1-e(X)\}^b,
\]
and \(e(X)\in(0,1)\), we have \(s\in(0,1)\). The function
\(t\mapsto \{1-t\}^b\) is strictly decreasing on \((0,1)\), so
\[
\frac{s^b-\{1-e(X)\}^b}{e(X)-1+s}>0
\]
when \(e(X)\ne 1-s\), and its value at \(e(X)=1-s\) equals
\(b s^{b-1}>0\). Therefore \(w^{\textup c}_{u;v,v'}(X)>0\).

By the definition of \(w^{\textup c}_{u;v,v'}(X)\),
\[
w^{\textup c}_{u;v,v'}(X)\{e(X)-1+s\}
=
\beta_{u,v}(X)\pi^{\textup c}(X)
\left[
s^b-\{1-e(X)\}^b
\right].
\]
Taking expectations gives
\[
\begin{aligned}
\E\!\left[
w^{\textup c}_{u;v,v'}(X)\{e(X)-1+s\}
\right]
=
s^b\E\{\beta_{u,v}(X)\pi^{\textup c}(X)\}
-
\E\{\beta_{u,v'}(X)\pi^{\textup c}(X)\}=0.
\end{aligned}
\]
Hence
\[
\E_{w^{\textup c}_{u;v,v'}}\{e(X)\}=1-s.
\]

Finally,
\[
\begin{aligned}
\tau^{\textup c}_{u,v}-\tau^{\textup c}_{u,v'}
&=
\E\left[
\left\{
\frac{\beta_{u,v}(X)\pi^{\textup c}(X)}
{\E\{\beta_{u,v}(X)\pi^{\textup c}(X)\}}
-
\frac{\beta_{u,v'}(X)\pi^{\textup c}(X)}
{\E\{\beta_{u,v'}(X)\pi^{\textup c}(X)\}}
\right\}
\tau^{\textup c}(X)
\right]\\
&=
\frac{1}{\E\{\beta_{u,v'}(X)\pi^{\textup c}(X)\}}
\E\left[
\beta_{u,v}(X)\pi^{\textup c}(X)
\left\{
s^b-\{1-e(X)\}^b
\right\}
\tau^{\textup c}(X)
\right]\\
&=
\frac{1}{\E\{\beta_{u,v'}(X)\pi^{\textup c}(X)\}}
\E\left[
w^{\textup c}_{u;v,v'}(X)\{e(X)-1+s\}
\tau^{\textup c}(X)
\right]\\
&=
\frac{\E\{w^{\textup c}_{u;v,v'}(X)\}}
{\E\{\beta_{u,v'}(X)\pi^{\textup c}(X)\}}
\E_{w^{\textup c}_{u;v,v'}}
\left[
\{e(X)-\E_{w^{\textup c}_{u;v,v'}}(e(X))\}
\tau^{\textup c}(X)
\right]\\
&=
\frac{
\E\{w^{\textup c}_{u;v,v'}(X)\}
\cov_{w^{\textup c}_{u;v,v'}}\{\tau^{\textup c}(X),e(X)\}
}{
\E\{\beta_{u,v'}(X)\pi^{\textup c}(X)\}
}.
\end{aligned}
\]
This proves part (ii).
\end{proof}

\subsection{Proof of Corollary~\ref{coro:cate beta weight}}

\begin{proof}
Recall \(\pi^{\textup{c}}(X)=\p(U=\textup{c}\mid X)\) in Section \ref{sec:iv}.
Define \(\E^{\textup{c}}\) as expectation with respect to the distribution of \(X\) among compliers:
\begin{align*}
\E^{\textup{c}}\{\varphi(X)\}
=
\frac{\E[\varphi(X)\pi^{\textup{c}}(X)]}{\E\{\pi^{\textup{c}}(X)\}}
\end{align*}
for any integrable function \(\varphi(X)\).
Then
\[
\tau^{\textup{c}}_{u,v}
=
\frac{\E^{\textup{c}}\{\beta_{u,v}(X)\tau^{\textup{c}}(X)\}}
{\E^{\textup{c}}\{\beta_{u,v}(X)\}}.
\]
Moreover,
\[
g^{\textup{c}}(e)
=
\E^{\textup{c}}\{\tau^{\textup{c}}(X)\mid e(X)=e\}
=
\E\{Y(1)-Y(0)\mid U=\textup{c}, e(X)=e\}.
\]
Thus the local beta-weighted estimands have the same form as the beta-weighted estimands in Corollary~\ref{coro:beta weight}, with the population distribution replaced by the complier distribution. We give the details for completeness.

Define the normalized beta weight under the complier distribution by
\[
\omega^{\textup{c}}_{u,v}(X)
=
\frac{\beta_{u,v}(X)}
{\E^{\textup{c}}\{\beta_{u,v}(X)\}},
\qquad
\E^{\textup{c}}\{\omega^{\textup{c}}_{u,v}(X)\}=1.
\]
Then
\[
\tau^{\textup{c}}_{u,v}
=
\E^{\textup{c}}\{\omega^{\textup{c}}_{u,v}(X)\tau^{\textup{c}}(X)\}.
\]
Since \(\omega^{\textup{c}}_{u,v}(X)\) is a function of \(e(X)\),
\[
\tau^{\textup{c}}_{u,v}
=
\E^{\textup{c}}\{\omega^{\textup{c}}_{u,v}(X)g^{\textup{c}}(e(X))\}.
\]

We first prove monotonicity in \(u\). Fix \(v\ge1\) and take \(u_2>u_1\ge1\). Define
\[
h(X)=\omega^{\textup{c}}_{u_2,v}(X),
\qquad
k(X)=\omega^{\textup{c}}_{u_1,v}(X).
\]
Then \(\E^{\textup{c}}\{h(X)\}=\E^{\textup{c}}\{k(X)\}=1\), and
\[
\tau^{\textup{c}}_{u_2,v}
-
\tau^{\textup{c}}_{u_1,v}
=
\E^{\textup{c}}\!\left[\{h(X)-k(X)\}g^{\textup{c}}(e(X))\right].
\]
Because \(\E^{\textup{c}}\{h(X)-k(X)\}=0\), for any constant \(C\),
\[
\tau^{\textup{c}}_{u_2,v}
-
\tau^{\textup{c}}_{u_1,v}
=
\E^{\textup{c}}\!\left[
\{h(X)-k(X)\}\{g^{\textup{c}}(e(X))-g^{\textup{c}}(C)\}
\right].
\]
Moreover,
\[
\frac{h(X)}{k(X)}
=
\frac{\E^{\textup{c}}\{\beta_{u_1,v}(X)\}}
{\E^{\textup{c}}\{\beta_{u_2,v}(X)\}}
e(X)^{u_2-u_1},
\]
which is non-decreasing in \(e(X)\). Let
\[
C
=
\left[
\frac{\E^{\textup{c}}\{\beta_{u_2,v}(X)\}}
{\E^{\textup{c}}\{\beta_{u_1,v}(X)\}}
\right]^{1/(u_2-u_1)}.
\]
Then
\[
h(X)\ge k(X)
\quad\Longleftrightarrow\quad
e(X)\ge C.
\]
If \(g^{\textup{c}}\) is non-decreasing, then
\[
\{h(X)-k(X)\}\{g^{\textup{c}}(e(X))-g^{\textup{c}}(C)\}\ge0
\]
almost surely under the complier distribution. Therefore,
\[
\tau^{\textup{c}}_{u_2,v}
\ge
\tau^{\textup{c}}_{u_1,v}.
\]
If \(g^{\textup{c}}\) is non-increasing, the same argument gives the reverse inequality. Hence \(\tau^{\textup{c}}_{u,v}\) is non-decreasing in \(u\) under Assumption~\ref{assump:iv mono}(i) and non-increasing in \(u\) under Assumption~\ref{assump:iv mono}(ii).

We next prove monotonicity in \(v\). Fix \(u\ge1\) and take \(v_2>v_1\ge1\). Define
\[
h(X)=\omega^{\textup{c}}_{u,v_2}(X),
\qquad
k(X)=\omega^{\textup{c}}_{u,v_1}(X).
\]
Then
\[
\tau^{\textup{c}}_{u,v_2}
-
\tau^{\textup{c}}_{u,v_1}
=
\E^{\textup{c}}\!\left[
\{h(X)-k(X)\}\{g^{\textup{c}}(e(X))-g^{\textup{c}}(C')\}
\right]
\]
for any constant \(C'\). Moreover,
\[
\frac{h(X)}{k(X)}
=
\frac{\E^{\textup{c}}\{\beta_{u,v_1}(X)\}}
{\E^{\textup{c}}\{\beta_{u,v_2}(X)\}}
\{1-e(X)\}^{v_2-v_1},
\]
which is non-increasing in \(e(X)\). Let
\[
C'
=
1-
\left[
\frac{\E^{\textup{c}}\{\beta_{u,v_2}(X)\}}
{\E^{\textup{c}}\{\beta_{u,v_1}(X)\}}
\right]^{1/(v_2-v_1)}.
\]
Then
\[
h(X)\ge k(X)
\quad\Longleftrightarrow\quad
e(X)\le C'.
\]
If \(g^{\textup{c}}\) is non-decreasing, then
\[
\{h(X)-k(X)\}\{g^{\textup{c}}(e(X))-g^{\textup{c}}(C')\}\le0
\]
almost surely under the complier distribution. Therefore,
\[
\tau^{\textup{c}}_{u,v_2}
\le
\tau^{\textup{c}}_{u,v_1}.
\]
If \(g^{\textup{c}}\) is non-increasing, the same argument gives the reverse inequality. Hence \(\tau^{\textup{c}}_{u,v}\) is non-increasing in \(v\) under Assumption~\ref{assump:iv mono}(i) and non-decreasing in \(v\) under Assumption~\ref{assump:iv mono}(ii).

Combining the two parts completes the proof.
\end{proof}
\end{document}